\newcount\preprintOn  \preprintOn=1   
\ifnum\preprintOn=1 
	\documentclass[a4paper,fleqn,preprint]{cas-sc_arxiv}
\else 
	\documentclass[a4paper,fleqn,preprint]{cas-sc}
\fi
\newcount\measles  \measles=1   
\newcount\append  \append=1   

\usepackage[authoryear,longnamesfirst]{natbib}
\usepackage{graphicx}
\graphicspath{{./figures/}}
\usepackage[utf8]{inputenc}
\usepackage{dsfont}
\usepackage{ifthen}
\usepackage[section]{placeins}
\usepackage{color}
\usepackage{float}
\usepackage{enumerate} 
\renewcommand{\baselinestretch}{1} 
\usepackage{booktabs}
\usepackage{tabularx}
\usepackage{ltablex}
\usepackage{bm}
\usepackage{amssymb,amsmath,amsbsy} 
\usepackage{longtable} 
\usepackage{textcomp}
\usepackage{lscape}
\usepackage[english]{babel}
\usepackage[T1]{fontenc}
\usepackage{helvet}
\usepackage{etoolbox}
\usepackage{booktabs}
\usepackage{xcolor} 
\usepackage{caption}
\usepackage{scrextend}
\usepackage{mathtools}
\mathtoolsset{showonlyrefs}
\usepackage{subfig}
\usepackage[normalem]{ulem}
\usepackage{enumitem}
\allowdisplaybreaks

\newtheorem{theorem}{Theorem}{\bfseries}{\itshape}
\newtheorem{corollary}[theorem]{Corollary}{\bfseries}{\itshape}
\newtheorem{lemma}[theorem]{Lemma}{\bfseries}{\itshape}
\newtheorem{proposition}[theorem]{Proposition}{\bfseries}{\itshape}
\newtheorem{definition}[theorem]{Definition}{\bfseries}{\itshape}
\newdefinition{remark}[theorem]{Remark}   
\newtheorem{assumption}[theorem]{Assumption}{\bfseries}{\itshape}
\newproof{pf}{proof}
\counterwithin{theorem}{section}
\counterwithin{equation}{section}
\counterwithin{figure}{section}
\counterwithin{table}{section}
\usepackage{tikz}
\usetikzlibrary{positioning,fit,arrows.meta,backgrounds}
\definecolor{cascblue}{RGB}{30,90,200}
\definecolor{cascred}{RGB}{200,30,30}
\definecolor{cascpurple}{RGB}{128,0,185}
\definecolor{cascpink}{RGB}{255,205,225}
\definecolor{cascbrown}{RGB}{140,85,30}
\definecolor{cascmagenta}{RGB}{195,0,135}
\definecolor{cascgreen}{RGB}{0,145,0}

\def \R{\mathbb{R}}               
\def \N{\mathbb{N}}               
\def \P{\mathbb{P}}             
\def \1{{\bf 1}}                
\def \0{{\bf 0}}

\def\covid{\text{Covid-19 }}
\def\covidk{\text{Covid-19}}

\definecolor{myred}{rgb}{0.9,0,0}  
\definecolor{mygreen}{rgb}{0,0.7,0}  
\definecolor{myblue}{rgb}{0.2,0,0.8}  
\definecolor{orange}{rgb}{1,0.6,0}  
\definecolor{olive}{rgb}{0.5,0.5,0}  
\definecolor{mylila}{rgb}{0.8,0.5,0.2}  
\definecolor{mygrey}{rgb}{0.6,0.6,0.6}  
\definecolor{mybrown}{rgb}{0.65,0.16,0.16}  
\definecolor{mymaroon}{rgb}{0.11,0.0,0.0}  
\definecolor{mygold}{rgb}{0.85,0.58,0.00}

\def \one{\mathcal{I}}

\newcommand{\myparagraph}[1]{%
	\vspace{1em}
	\noindent\textbf{#1.}~~}

\renewcommand{\P}{\mathbb{P}}

\newcommand{\countP}{\Theta_k}  

\newcommand{\E}{\mathbb{E}}

\def\Gauss{\mathcal{Z}}
\def\driftX{f_X}
\def\diffX{\sigma_X}
\def\diffXrel{{\sigma}_U}
\newcommand{\Noise}{{\mathcal{B}}}
\newcommand{\NoiseB}{{\mathcal{E}}}

\newcommand{\XDrel}{U^N}
\newcommand{\XD}{X^D}
\newcommand{\xDrel}{u}
\newcommand{\xD}{x}
\newcommand{\Xrel}{\underline{X}}
\newcommand{\XrelN}{\Xrel^N}
\newcommand{\Xrelinfty}{\Xrel^{\infty}}
\newcommand{\xrel}{\underline{x}}
\newcommand{\Irel}{\underline{I}}
\newcommand{\IrelN}{\Irel^N}
\newcommand{\Irelinfty}{\Irel^{\infty}}

\def\L{L}
\def\LV{\L^V}
\def\LR{\L^R}
\def\LG{\L^G}
\def\LD{\L^\dagger}
\def\K{d}
\def\KV{\K^V}
\def\KR{\K^R}
\def\KG{\K^G}
\def\KD{\K^\dagger}
\def\PP{P}
\def\PV{\PP^V}
\def\PR{\PP^R}
\def\PG{\PP^G}
\def\PD{\PP^\dagger}
\def\psiV{\psi^V}
\def\psiR{\psi^R}
\def\psiG{\psi^G}
\def\psiD{\psi^\dagger}

\def \Repro{\mathcal{R}}

\def \Repronull{\Repro_0}

\def \Reproeff{\Repro_E}
\def \Reproeffn{\Repro_{E,n}}

\def \regupar{\zeta}

\newcommand{\mymarginpar}[1]{ \marginpar{{\tiny #1}}}
\definecolor{darkgreen}{rgb}{0.0,0.40,0.0}
\newcommand{\prepcolor}{}

\renewcommand{\mymarginpar}[1]{}

\usepackage[section]{algorithm}
\usepackage{algpseudocode}
\usepackage[algo2e]{algorithm2e} 

\begin{document}
\let\WriteBookmarks\relax
\def\floatpagepagefraction{1}
\def\textpagefraction{.001}

	\title[mode = title]{Stochastic Epidemic Models with Partial Information}
	\shorttitle{Stochastic Epidemic Models with Partial Information}	
	\tnotemark[1]
	\tnotetext[1]{This document is an extended version of a manuscript that was submitted for publication in a journal.}

\shortauthors{Florent Ouabo Kamkumo, Ibrahim Mbouandi Njiasse, Ralf Wunderlich}



\author[]{Florent {Ouabo Kamkumo}}[orcid=0000-0003-3319-975X]
\ead{Florent.OuaboKamkumo@b-tu.de}
\credit{Conceptualization, Data curation, Formal analysis, Investigation, Methodology,   Software,  Validation, Visualization, Writing - original draft}
\affiliation[]{organization={Brandenburg University of Technology Cottbus-Senftenberg, Institute of Mathematics},
                addressline={P.O. Box 101344}, 
                city={Cottbus},
                postcode={03013}, 
                country={Germany}}

\author[]{Ibrahim {Mbouandi Njiasse}}[orcid=0009-0007-6039-1583]
\ead{Ibrahim.MbouandiNjiasse@b-tu.de}
\credit{Formal analysis, Investigation, Methodology, Writing - review \& editing}

\author[]{Ralf Wunderlich}[orcid=0000-0001-5841-6293]
\cormark[1]
\ead{Ralf.Wunderlich@b-tu.de}
\credit{Conceptualization, Formal analysis, Funding acquisition, Investigation, Methodology, Project administration, Supervision,  Writing - review \& editing}
\cortext[cor1]{Corresponding author}


\begin{abstract}
Mathematical models of epidemics often use compartmental models dividing the population into several compartments. Based on a microscopic setting describing the temporal evolution of the subpopulation sizes in the compartments by stochastic counting processes, one can derive macroscopic models for large populations describing the average behavior by associated ordinary differential equations as in the celebrated SIR model. Further, diffusion approximations allow one to address fluctuations from the average,  which are also of interest in smaller populations, using stochastic differential equations. 		

In general, not all state variables are directly observable, and we face the so-called ``dark figure'' problem, which concerns, for example, the unknown number of asymptomatic and unreported infections. The present study addresses this problem by developing partially observable stochastic epidemic models that account for uncertainty about the current state of the epidemic, also known as nowcast uncertainty. Examples include a simple extension of the SIR model, a model for a disease in which infection or vaccination confers lifelong immunity, and a \covid model. For the latter, we propose a ``cascade state approach'' that allows us to exploit the information contained in formally hidden compartments with observable inflow but unobservable outflow. Furthermore, parameter estimation and calibration are performed using ridge regression for the \covid model. The results of the numerical simulations illustrate the theoretical findings.
\end{abstract}



\begin{keywords}
Continuous-time Markov chain \sep Stochastic epidemic model \sep Stochastic differential equation \sep Diffusion approximation \sep  Partial observation \sep Covid-19\sep  Parameter estimation 
\ifthenelse{\preprintOn=0}{}{{\sep~ \sep ~ \sep \sep~ \sep ~ \sep Version of \today }}
\end{keywords}

\maketitle
\ifthenelse{\preprintOn=0}{}{{\prepcolor
\setcounter{tocdepth}{2}
\tableofcontents
}}

\section{Introduction}
For many years, mathematical models have been an important tool for scientists to understand and analyze the dynamics of epidemics. Epidemics have always been a major public health challenge and require effective strategies to understand and control the spread of infectious diseases. Traditional deterministic models such as SIR, SIRS and SEIR have provided valuable insights into disease transmission and progression. Many researchers have developed deterministic epidemic models to assist public health officials in understanding the spread of infectious diseases. \citet{ogilvy1927contribution} introduced the deterministic SIR model, which today forms the basis of much modern epidemic modeling. Their work was originally motivated by epidemics of bubonic plague in India. The goal was to provide a mathematical explanation for observed epidemic behaviors, such as the rise and fall in the number of cases, using a theoretical framework that could guide public health interventions. Since then, numerous compartmental mathematical models, both deterministic and stochastic, have been developed to analyze and understand  the course of the epidemics.

 The underlying idea of these models is to divide the population into several subpopulations collected in compartments.
	On a microscopic level, the temporal evolution of the compartment sizes is described by stochastic counting processes. One can then derive macroscopic models for large populations that describe the average behavior by associated ordinary differential equations (ODEs), such as in the celebrated SIR model. To describe fluctuations around this average behavior, which can be relevant for smaller populations, diffusion approximations can be derived. These approximations describe the state dynamics by stochastic differential equations (SDEs). 	
	
	However, in general, some state variables are not directly observable because they involve individuals with asymptomatic or unreported infections. This is also known as ``nowcast uncertainty'' or the ``dark figure'' problem. Such unobservable states are problematic for computing epidemic characteristics such as the effective reproduction number and infection prevalence. 
	Furthermore, the management and containment of epidemics based on stochastic optimal control problems and associated feedback controls require observations of the current state as input. Even deterministic ODE models   suffer from the fact that the initial conditions, which describe the initial sizes of the compartments, cannot be precisely specified, since some of them cannot be directly observed.

	This article addresses the above problems with not directly observable or hidden states of an epidemic and proposes partially observable stochastic epidemic models, in which only some states of the epidemic are observed. They are used in the companion papers \citep{ouabo2025estimating,achola2025stochastic,mbouandi2025convergence} for estimating the hidden states based on the available information from observable states, and in \citep{mbouandi2025stochastic} 
	for solving stochastic optimal control problems under partial information arising from the decision-making problem of a social planner seeking cost-optimal containment of an epidemic.

\myparagraph{Literature Review on Deterministic Epidemic Models}
Deterministic epidemic modeling originates from the seminal work of  \citet{ogilvy1927contribution}, who introduced the classical SIR framework that laid the foundation for modern compartmental modeling of infectious diseases. Since then, deterministic models have been extensively developed and generalized. Comprehensive surveys of compartmental epidemic models, their mathematical properties, and threshold dynamics can be found in \citet{hethcote2000mathematics} and \citet{brauer2017mathematical,brauer2019mathematical}, which provide systematic overviews of stability analysis, reproduction numbers, endemic equilibria, and model extensions.

Beyond classical SIR-type structures, an important strand of the literature incorporates compartments that are not directly observable, such as exposed, asymptomatic, or  unreported infected individuals. These extensions allow for a more realistic representation of transmission dynamics, particularly when testing is incomplete or reporting is delayed. Such modeling strategies became especially prominent during the \covid pandemic. For example, 
\citet{charpentier2020covid} proposed a deterministic \covid model that implicitly accounts for undetected infected individuals and formulated an optimal control problem to determine strategies that maintain a low level of infectious individuals.

Deterministic epidemic models have also frequently been combined with optimal control theory to evaluate intervention strategies such as vaccination and treatment \citep{brauer2019mathematical,charpentier2020covid}, see \citet{behncke2000optimal} and the monograph \citet{lenhart2007optimal} for the underlying methodology, \citet{sharomi2017optimal} for a review, and \citet{kantner2020beyond,perkins2020optimal} for applications to non-pharmaceutical interventions against \covidk. These approaches provide insight into how public health policies affect disease dynamics and allow researchers to quantify trade-offs between epidemiological outcomes and resource allocation. 

\myparagraph{Literature Review on Stochastic Epidemic Models}
 While deterministic models offer structural clarity and analytical tractability, they do not capture stochastic fluctuations that may play an important role during the early phase of an outbreak or in smaller populations. This limitation motivates the formulation of stochastic models that can be undertaken in various ways. \citet{allen2008introduction} presents different approaches to formulating stochastic models. She discusses three primary methods: discrete-time Markov chains, continuous-time Markov chains, and stochastic differential equations. Furthermore, she explores  properties for analyzing stochastic epidemiological models, such as the probability of disease extinction and the probability of disease outbreak \citep{allen2015stochastic}.

In \citet{allen2012extinction}, 
the authors propose a stochastic multigroup model based on continuous-time Markov chains (CTMC) to address the issue of disease extinction.  \citet{allen2021stochastic} develops a stochastic model for cholera using branching process theory to analyze the probability of disease extinction.
\citet{didelot2017model} proposed a new stochastic model describing the course of bubonic plague epidemics in rat and human populations. In their modeling, they considered rat and human populations separately, providing a better understanding of the dynamics of bubonic plague in both populations. 
\citet{qureshi2022modeling} propose a stochastic model combining the Auto-Regressive Integrated Moving Average (ARIMA) methodology with a neural network based approach to predict the global cumulative cases of monkeypox virus (MPXV) at peak levels. 
\citet{finkenstadt2002stochastic} study  a stochastic model to explain extinction and recurrence of epidemics observed in measles. The work 
\citet{guo2023numerical} present a detailed investigation of a stochastic model based on the SEIR framework with vaccination. This model governs the spread of the measles virus while incorporating white noise and the effects of immunizations. The latest pandemic of \covid has been the focus of much attention from researchers, who have proposed numerous studies to understand the evolution of the disease, see  \citep{barbarossa2020modeling,bardina2020stochastic,mamis2023stochastic,tesfaye2021stochastic},  and how to control it as in \citep{el2024stochastic,he2020discrete,li2022dynamics}. The \covid pandemic emphasized the critical role of such models in predicting the spread of infectious diseases and supporting public health decisions. Accurate modeling has become increasingly important to guide interventions and mitigate the impact of outbreaks.
Traditional epidemic models, such as the SIR model and SEIR model  are often limited by their assumption of complete knowledge about the state of the epidemic, which is not realistic in the face of incomplete or noisy data, and also when we are using a model like SEIR with  an exposed  compartment E denoting individuals exposed to the disease, but not yet infectious. Particularly, unobserved cases, such as asymptomatic individuals and delays in testing, contribute to the significant challenges in accurately estimating the current state of an epidemic, see \citep{colaneri2022invisible, zhu2021extended}, and the key model parameters as in \citep{hasan2022new,song2021maximum}.  Summarizing, stochastic epidemic models offer a flexible alternative, allowing the incorporation of randomness and uncertainty in epidemic dynamics.

\myparagraph{Our Contribution}
  This paper provides a comprehensive introduction to stochastic epidemic models with incomplete information that explicitly account for ``nowcast uncertainties'' caused by limited access to information about the current values of certain state variables.   We start with a CTMC framework inspired by the concept of chemical reaction networks, as in \citet{anderson2011continuous} and \citet{britton2019stochastic}, and derive a diffusion approximation using the law of large numbers and the central limit theorem for Poisson processes. This makes it possible to describe the dynamics of the state vector, which is decomposed into a hidden and an observable component, using systems of SDEs and, after time discretization, stochastic recursions.

The theoretical approach is illustrated using various examples. We start with a simple extension of the  SIR model\ifthenelse{\measles=1}{, followed by a model  for a disease in which infection or vaccination confers lifelong immunity}{}. Finally,  we develop a sophisticated stochastic model for \covidk,  that accounts for states that are only partially observable due to asymptomatic and unreported infections, as well as for the waning immunity following recovery and vaccination. To capture as much available information as possible, the model is extended to include so-called cascade states in order to capture information from reported inflows into formally hidden compartments, such as through vaccinations or recoveries following confirmed infections.

The \covid model was calibrated using real-world data from Germany provided by the Robert Koch Institute (RKI). Particular attention was paid to time-dependent parameters that govern infection, testing, and vaccination rates.  For this purpose, we applied ridge regression methods to regularize the raw estimates. The calibrated model was used to numerically simulate the course of the \covid epidemics. The resulting numerical results are consistent with observations in Germany during the pandemic, particularly with regard to infection dynamics.

The proposed models are ready for the application of stochastic filtering methods to estimate hidden states based on past and current data for the observable states, see our articles \citep {ouabo2025estimating,mbouandi2025convergence,achola2025stochastic},  as well as for use in stochastic optimal control problems with partial observation, as described in \cite{mbouandi2025stochastic}.

	\myparagraph{Paper Organization}
	Section \ref{Micro_Model} begins with mathematical modeling of epidemics and deals with microscopic models described by CTMCs. Large population limits of CTMC models  are examined in Section \ref{Macros_Model}. This allows for the derivation of diffusion approximations that approximate the state dynamics using systems of SDEs and stochastic recursions resulting from time discretization.  Section \ref{Section_Time_Discr} introduces partially observable models, in which the state is split into hidden and observable components. Section \ref{Model_example} contains various model examples, with a particular focus on the \covid model, for which a cascade state approach is proposed. Section \ref{D_calibration} examines the calibration of the key parameters of the \covid model using real-world data. In Section \ref{sec:NumericalResults}, we present the results of numerical simulations of the calibrated  \covid model. 
	The paper concludes in Section \ref{sec:conclusion} with a short summary and outlook. An appendix provides a list of commonly used notations.

\section{Microscopic Models}\label{Micro_Model}
In this section, we want to derive a  structured population model in which a  population of constant size $N \in \mathbb{N}$ is subdivided into $d \in \mathbb{N}$ \textit{compartments}, each representing a subpopulation with a distinct epidemiological status. Individuals in this model can undergo  $K \in \mathbb{N}$ distinct \textit{transitions} between these compartments.
Let $T > 0$ denote a fixed time horizon. The stochastic processes that will describe the model are defined on a complete filtered probability space $(\Omega, \mathcal{F}, \mathbb{F}, \mathbb{P})$, where the filtration $\mathbb{F} = (\mathcal{F}_t)_{t \in [0,T]}$ satisfies the usual conditions of right-continuity and completeness. Let $X = (X(t))_{t \in [0,T]}$ denote the $\mathbb{Z}^d$-valued   \textit{state process}, for which each component $X^i(t) \in \{0, \ldots, N\}$ represents  the  number of individuals or the absolute size of the subpopulation in compartment $i \in \{1, \ldots, d\}$ at time $t$. The filtration $\mathbb{F}$  is assumed to be generated by the full state process $X$ (the information actually available to an observer is considered in Section \ref{Section_Time_Discr}), i.e., $\mathcal{F}_t = \sigma(X(s): 0 \leq s \leq t)$.

\subsection{Counting Processes}\label{Counting_Proc}
The stochastic dynamics of the state process $X$ are determined by the transition of individuals between compartments. Examples of this are the infection of a susceptible person, who moves to the ``infected'' compartment, or the recovery of an infected person, who moves to the ``recovered'' compartment. These transitions occur at random times and are recorded using counting processes, i.e., stochastic integer-valued  processes. These processes are the building blocks of microscopic epidemic models and are explained in this subsection.  

Let \(K\) represent the total number of possible transitions that individuals can undergo.
In order to systematically track these transitions over time, we introduce for \( k = 1, \dots, K \) the counting processes \( \countP=(\countP(t))_{t\in[0,T]}  \) , which record the total number of transitions of type \( k \) that occurred within the time interval \([0,t]\). 
Each transition type corresponds to a distinct epidemiological event, such as infection, recovery, or loss of immunity. The  counting processes \( \countP\)  are non-negative and non-decreasing on \( [0,T] \), and start at initial time \( t = 0 \)  with \( \countP(0) = 0 \).  Further, for $0\le s<t\le T$ the increment \( \countP(t) - \countP(s) \)  represents the number of transitions of type \( k \) that have occurred within the time interval \((s,t]\).

Transitions occur at random time points, and  each transition refers to an individual  who enters a specific compartment after leaving another compartment.
Thus, transition \( k \)  modifies at time $t$ the  component \( X_i,  i = 1, \dots, d,  \) of the state process according to  
\(X_{i}(t) = X_{i}(t^{-}) + \xi_{k}^{i},\) where the increment  \( \xi_{k}^{i} \) is given by  
\begin{align}
\label{trans_component}	
\xi_{k}^{i}=
\begin{cases}  
	-1, & \text{if an individual leaves compartment } i, \\  
	+1, & \text{if an individual enters compartment } i, \\  
	\phantom{-}0, & \text{otherwise,}  
\end{cases}  
\end{align}
 and $X_{i}(t^{-})$ denotes the state immediately before the transition.  By forming the \( d \)-dimensional “transition vectors” \( \xi_k \) with the above entries \( -1 \), \( 0 \), or \( 1 \), which indicate the effect of transition \( k \) on each of the \( d \)-compartments, the dynamics of the state process \( X \) can be described as follows: 
\begin{align}
\label{stateX}	
X(t) = X(0) + \sum\limits_{k=1}^{K} \xi_k \countP(t),\quad t\in[0,T].
\end{align}

\subsection{Continuous-time Markov Chain Model }
\label{sec:CTMC}
 The stochastic nature of the state process \( X \) given in  \eqref{stateX} results from the random temporal occurrence  of transitions, which are described by the counting processes \( \countP(t) \). 
 Common assumptions in epidemic compartment models are that the different transition mechanisms act independently of each other and that the counting processes are memoryless. The frequency with which transitions occur depends on  the current state. For example, the number of infections in a short period of time is proportional to the current number of infected individuals and of susceptible individuals.  Furthermore, this transition frequency depends on disease-specific parameters that can change over time. Such non-stationary behavior is justified by biological and environmental considerations. For example, infection rates may fluctuate due to behavioral changes, pathogen mutations, or seasonal effects, see e.g. \cite{colaneri2022invisible,hasan2022new}.
 
 A natural mathematical framework for modeling the processes $\countP$ that count the randomly occurring transition events with the above properties is provided by
 independent Markov processes, and in particular by non-homogeneous Poisson processes. They are characterized by intensity functions of the form \(
 \lambda_{k}(t,x), k=1,\ldots,K, \)	depending both on time \(t\) and on the current state \(x\). Then, for a small time increment
 \(\Delta t\ll 1\), the probability of exactly one transition of type \(k\) occurring in  \((t,t+\Delta t]\) satisfies 
 \( \mathbb{P} \big(\countP(t+\Delta t)-\countP(t)=1\big) =\lambda_k(t,X(t))\,\Delta t+o(\Delta t), \)
 while the probability of two or more transitions is \(o(\Delta t)\). Further, for $0\le s<t\le T$ it holds 
 	\[
 \countP(t)-\countP(s)\sim 
 \mathrm{Poi}\Big(\int_s^{t}\lambda_k(u,X(u))\,du\Big),
 \qquad 0\le s<t\le T,
 \]
 where $\mathrm{Poi}(a)$ denotes the Poisson distribution with parameter $a>0$, given by the   probability mass function  $\frac{a^j}{j!}e^{-a}, j=0,1,\ldots$.
 It is well-known  that a non-homogeneous Poisson process can be expressed by a standard Poisson process with unit intensity, using a time change defined by the 
   cumulative intensity \(	\tau_k(t)=\int_0^{t}\lambda_k(s,X(s))\,ds\),	see \cite{brown2002time}. This allows  us to model the counting processes $\countP, k=1,\ldots,K$, as \( \countP(t)=\Pi_k(\tau_k(t)),\)  where  $\Pi_1,\ldots,\Pi_K$ are independent standard Poisson processes with unit intensity. 
By this  time–change argument, the state process can be represented as 
	\begin{equation}\label{Eq_CTMC}
		X(t)
		= X(0)+\sum_{k=1}^{K}\xi_k\,\Pi_k\Big(\int_0^{t}\lambda_k(s,X(s))\,ds\Big),
	\end{equation}
	where $\xi_k$ denotes the transition vector of transition type \(k\) given in \eqref{trans_component}.
	
	The stochastic state process \(X(t)\) takes values in a finite subset of $\mathbb{N}^d$ and, by
	construction, satisfies the Markov property. Further, it holds 
		\(
	\mathbb{P}(X(t+\Delta t)-X(t)=\xi_k\mid X(t))
	=\lambda_k(t,X(t))\,\Delta t +o(\Delta t).
	\)
	In the literature, a stochastic  process of this type is called  (non-homogeneous) continuous–time Markov chain (CTMC), see
	\cite{anderson2011continuous,britton2019stochastic}.  For the numerical simulation of such CTMCs using the Gillespie algorithm, we refer  to \cite{gillespie1977exact} for the case of homogeneous CTMCs.
\ifthenelse{\preprintOn=0}{ A modification for non-homogeneous  CTMCs and implementation details can be found in \citep[Section 2.3]{ouabo2026partial}.}{ A modification for non-homogeneous  CTMCs and implementation details are described in the next subsection.}

\ifthenelse{\preprintOn=0}{}{{\prepcolor
\subsection{Adapted Gillespie Algorithm for CTMC Simulation }\label{sec_Gillespie}
In this subsection we explain how the Gillespie algorithm works and how its standard version can be generalized to the case of time-dependent transition intensities, as they occur in our model.

Recall that for \(0\leq s\leq t\) and  $k=1,\ldots,K$, we have $\countP(t)-\countP(s)\sim \textrm{Poi}(\int\nolimits_{s}^{t}\lambda_{k}(u,X(u))du)$. We define \( \overline{\tau}_{k} \) as the random waiting time until transition \( k \) occurs. We  introduce $h> 0$ in such a way that $h< \overline{\tau}_{k}$. Then the process $X(t)$ remains  constant during the waiting time before the transition.  We assume that $\countP(t)$ is known, then the probability that the waiting time $\overline{\tau}_{k}$ is greater than $h$ given  $\countP(t)$, is the conditional survival function of $h$ which takes the following form 
\[ \overline{F}_{k,t}^{}(h):=\mathbb{P}[\overline{\tau}_{k} >h|\countP(t)] =\mathbb{P}[\countP(t+h)-\countP(t)=0|\countP(t)]=\exp\Big[-\int\nolimits_{t}^{t+h}\lambda_{k}(u,X(u))du\Big].\]
The conditional density can be derived through the relation between the survival function  $\overline{F}_{k,t}$ and the density given by $f_{k,t}^{}(h)=-\frac{\partial}{\partial h} \overline{F}_{k,t}^{}(h)$. Based on the Leibniz integral rule, we obtain 
\[ f_{k,t}^{}(h)= \lambda_{k}(t+h,X(t+h)) \exp\Big[-\int\nolimits_{t}^{t+h}\lambda_{k}(u,X(u))du\Big].\]
Using the density function and the survival function, we can compute the intensity 
$ \lambda_{k}(t+h,X(t+h))=
{f_{k,t}^{}(h)}/{\overline{F}_{k,t}^{}(h)},$  
which characterizes the probability of transition in an infinitesimal period of time after $t+h$, conditional on the absence of transition before $t+h$.
Now we introduce $\overline{\tau}$, the waiting time until the next transition (of any type ) occurs, for which we have  
$ \overline{\tau}=\overline{\tau}_{1}\wedge \overline{\tau}_{2}\wedge\ldots\wedge \overline{\tau}_{K}.$
Since \(\{\overline{\tau}>h\}=\{\overline{\tau}_{1}>h\}\cap \{\overline{\tau}_{2}>h\}\cap\ldots\cap \{\overline{\tau}_{K}>h\}\) and knowing that all the counting processes $\countP(t)$, $k=1,\ldots,K$ are independent, then the survival function with respect to $h$ can be computed as follows 
\[\overline{F}_{t}^{}(h)=1-F_{t}(h)=\mathbb{P}(\{\overline{\tau}>h\})=\prod\limits_{k=1}^{K}\mathbb{P}(\{\overline{\tau}_{k}>h\})=\exp\Big[-\int\nolimits_{t}^{t+h}\sum_{k=1}^{K}\lambda_{k}(u,X(u))du\Big].\]
Let the total transition rate be defined as \( \Lambda(t, X(t)) = \sum_{k=1}^{K} \lambda_k(t, X(t)) \). It is the rate for the occurrence of one of the $K$ possible transitions at time \( t \). When such a transition occurs, then it is of type \( k \) with probability \( {\lambda_k(t, X(t))}/{\Lambda(t, X(t))} \). Then the above integral can be written in the following way 
\[\overline{F}_{t}^{}(h)=\exp\Big[-\int\nolimits_{t}^{t+h}\Lambda_{}(u,X(u))du\Big]=\exp\Big[-\int\nolimits_{0}^{h}\Lambda_{}(t+u,X(t+u))du\Big].\]
Now, knowing the expression of the survival function $\overline{F}_{t}^{}(h)$, we can simulate the waiting time $\overline{\tau}$ by using the inversion method (or quantile transformation). We draw $U$ uniformly in the interval $(0,1)$ and solve the equation $U=\overline{F}_{t}^{}(h)$ which is equivalent to 
\[\log U=-\int\nolimits_{0}^{h}\Lambda_{}(t+u,X(t+u))du.\]
Assuming that $ \Lambda_{}(t,X(t))=\Lambda_{0}$ is constant, the waiting time is $\overline{\tau}=-{\log U}/{\Lambda_{0}}$, which aligns with the  setting  in the conventional Gillespie algorithm. For $\Lambda_{}(t,X(t))$ slowly varying, then we  can approximate the integral $ \int\nolimits_{0}^{h}\Lambda_{}(u+t,X(u+t))du $ by $h\Lambda_{}(t,X(t))$ and derive an approximation of the waiting time $\overline{\tau}\approx -{\log U}/{\Lambda(t,X(t))}$. 
For an exact solution, we introduce $B_{t}(h)=\int\nolimits_{0}^{h}\Lambda_{}(t+u,X(t+u))du$.  Then, the waiting time is obtained by solving the equation 
$B_{t}(h) + \log U=0.$
Often, $B_{t}(h)$ is not known explicitly, or the form of $\Lambda_{}(t,X(t))$ is complex, making it difficult to find an exact solution of the above equation. One solution is to use an approximation method to solve the corresponding non-linear equation. As an example, assuming \(\Lambda(t,X(t)) > 0\), the Newton-Raphson method can be utilized to estimate the solution of the equation above as follows:
\[\overline{\tau}_{n+1}=\overline{\tau}_{n} - \dfrac{B_{t}(\overline{\tau}_{n}) + \log U}{B_{t}^{'}(\overline{\tau}_{n})},~~~~~\overline{\tau}_{0}=-\dfrac{\log U}{\Lambda(t,X(t))}.\]
Now, we know how to simulate the waiting time $\overline{\tau} $ in the case where the intensity  is constant or time-dependent. The next step is to show how to determine which transition  among all the $k$, $k= 1,\ldots,K$, transitions  will take place. At any moment \( t \), transitions occur at a rate of \( \Lambda(t, X(t)) \). When a transition takes place, it is identified as type \( k \) with a probability of \( {\lambda_k(t, X(t))}/{\Lambda(t, X(t))} \).

Following a transition, the system moves from the current state \( x \) to a new state \( x' \), determined by a specific transition vector \( x' = x + \xi_{k^*} \), where \( \xi_{k^*} \) corresponds to the selected transition type \( k^* \). The choice of \( k^* \in \{1, \ldots, K\} \) can be simulated using the cumulative sum approach. First, for each transition type \( k \), the cumulative probability is computed as \( p_k = \sum_{i=1}^{k} \lambda_{i}(t + \overline{\tau}, X(t)) / \Lambda(t + \overline{\tau}, X(t )) \), where \( \Lambda(t + \overline{\tau}, X(t )) = \sum_{i=1}^{K} \lambda_{i}(t + \overline{\tau}, X(t )) \) is the total transition rate. A random number \( r \) is then drawn from a standard uniform distribution, and the index \( k^* \) is chosen such that \( p_{k^*-1} \leq r < p_{k^*} \), with \( p_0 = 0 \). Once \( k^* \) is determined, the state and time dynamics are updated accordingly: \( X(t + \overline{\tau}) := X(t) + \xi_{k^*} \) and \( t := t + \overline{\tau} \). This ensures that transitions are simulated consistently with their respective probabilities and rates.

 Since the paths are simulated on the finite horizon $[0,T]$, the next transition time $t+\overline{\tau}$ may exceed $T$. Between two transitions the state is constant, $X(t+u)=X(t)$, so that $B_{t}(h)=\int\nolimits_{0}^{h}\Lambda(t+u,X(t))du$ can be evaluated from the current state, and $h\mapsto B_{t}(h)$ is non-decreasing. Hence $t+\overline{\tau}>T$ holds if and only if $B_{t}(T-t)<-\log U$, i.e., no transition occurs in $(t,T]$. In this case the state remains constant, $X(s)=X(t)$ for $s\in(t,T]$, and the simulation stops at $t=T$. Otherwise, the equation $B_{t}(h)+\log U=0$ has a solution $\overline{\tau}\in(0,T-t]$, and it suffices to search for this root in the bounded interval $(0,T-t]$. This test also covers the case that the total rate vanishes, e.g., after the extinction of the epidemic, where $B_{t}(h)$ remains bounded and the equation may have no solution at all. Moreover, the intensities are only evaluated on $[0,T]$, where the time-dependent parameters are defined. The complete procedure is summarized in Algorithm \ref{alg:simul_CTMC}.

\begin{algorithm}
	\caption{Simulation of Non-homogeneous CTMC}\label{alg:simul_CTMC}
	\DontPrintSemicolon
	\KwIn{$X_{0}$, $\xi_k$, $k=1,\ldots,K$, $T$}
	\KwOut{$X=(X(t))_{t\in[0,T]}$}
	\textbf{Initialization :}~~$t:=0$ ; $X(t):=X_{0}$
	
	\While{$t<T$}
	{
		Draw $U$ from standard uniform distribution\;\\
		\eIf{ $B_{t}(T-t)<-\log U$}
		{
			$X(s):=X(t)$ for $s\in(t,T]$, ~~$t:=T$ ~~(no further transition in $(t,T]$)\;
		}
		{
			Solve $B_{t}(h) + \log U =0$  for $h\in(0,T-t]$ and obtain $\overline{\tau}$\;\\
			\For{$k:=1$ to $K$}
			{
				$p_{k}:=\sum\limits_{i=1}^{k}\dfrac{\lambda_{i}(t+\overline{\tau}, X(t ))}{\Lambda(t+\overline{\tau}, X(t ))}$\;
			}
			Draw $r$ from standard uniform distribution\;\\
			Pick $k= 1,\ldots,K$ such that $p_{k^{}-1}\leq r< p_{k^{}}$\;\\
			 $X(s):=X(t)$ for $s\in(t,t+\overline{\tau})$\;\\
			$X(t+\overline{\tau}):=X(t)+\xi_k$,~~$t:=t+\overline{\tau}$\;
		}
	}
\end{algorithm}
}}

\section{Macroscopic Models }\label{Macros_Model}
 In this section, we review the large population limits of  stochastic CTMC epidemic models. They are based  on the  law of large numbers  and the central limit theorem for Poisson processes, and  lead to models based on ODEs and SDEs instead of CTMCs.

\subsection{Law of Large Numbers and ODE Models}\label{LLN}
 For the study of the asymptotic behavior  of the state $X$ for increasing total population size $N$ it is helpful to emphasize the dependence of $X$ on  $N$ using the notation $X^N$. We recall the CTMC state dynamics from \eqref{Eq_CTMC}, which reads 
$
	X^N(t)=X^N(0)+\sum\nolimits_{k=1}^{K}\xi_{k} \Pi_k\big(\int\nolimits_{0}^{t}\lambda_{k}(s,X^N(s))ds\big). 
$ 
 These dynamics can also be expressed in terms of relative subpopulation sizes denoted by \( \XrelN \) and defined as the normalized  state vector given by   \( \XrelN\!(t) =  N^{-1}X(t) \). The components of \( \XrelN(t) \)  take values in $[0,1]$, and represent the proportions of the total population within various compartments at time \( t \). The dynamics of \( \XrelN \)  follow from \eqref{Eq_CTMC} and  are given by
\begin{align}
	\XrelN(t)=\XrelN(0)+\sum_{k=1}^{K}\frac{1}{N}\xi_{k} \Pi_k\Big(\int\nolimits_{0}^{t}\lambda_{k}(s,N\XrelN(s))ds\Big).
\end{align} 
For the application of the announced limit theorems, the following regularity assumptions are necessary.
\begin{assumption}~\label{ass_DA}
There exist  functions $\nu_{1},\ldots,\nu_{K}: [0,T]\times[0,1]^{d}\rightarrow \mathbb{R}^{+}$   that satisfy	$\nu_{k}\in C^{2}([0,T]\times[0,1]^{d})$, $k= 1,\ldots,K$, and the scaling property
\begin{align}
\label{H2r}
\lambda_{k}(t, N\xrel )=N\nu_{k}(t,\xrel ),\quad \text{for all $k=1,\ldots,K,  ~t\in [0,T], ~\xrel \in [0,1]^{d}$.}
\end{align}
\end{assumption}
Within the context of epidemic models, this scaling property typically holds true for each transition, see \citep[Section 2.2]{britton2019stochastic}. In particular, this relation is fulfilled when the intensity rate $\lambda_{k}(t,x)$ is linear  in $x$, or quadratic  in $x$, e.g.~proportional to  ${x_{i}x_{j}}/{N}$ with $x_{i},x_{j}\in \R$, \(i,j\in \{1,\ldots,d\},~i\neq j\). With these functions and suitable possibly time-dependent model parameters, the intensities $\nu_k$ are also  $C^{2}$ functions  as required in Assumption \ref{ass_DA}.  
\begin{remark}\label{rem:scaling_prop}
Theorems \ref{LLN_Theo} and \ref{theo:CLT} below also hold under slightly weaker assumptions regarding the smoothness and boundedness of the functions $\nu_k$, and when \eqref{H2r} holds only asymptotically for $N\to\infty$, i.e., $\frac{1}{N}\lambda_{k}(t,\lfloor N\xrel \rfloor)$ converges to $\nu_{k}(t,\xrel)$ uniformly in $(t,\xrel)$ for the law of large numbers, and at a rate faster than $1/\sqrt{N}$ for the central limit theorem, see \citep{kurtz1971limit}. Here $\lfloor N\xrel\rfloor$ is  the vector of integers \((\lfloor N\xrel_{1}\rfloor,\ldots,\lfloor N\xrel_{d}\rfloor)^{\top}\) and  $\lfloor z\rfloor$ represents the largest integer not exceeding $z$.  Under additional conditions, they can also be extended to an infinite time horizon $T=\infty$, see \citep[Chapter 2]{britton2019stochastic}. 
\end{remark}
Using \eqref{H2r},  we can rewrite $\XrelN(t)$ as 
\begin{align}
	\XrelN(t)=\XrelN(0)+\sum\limits_{k=1}^{K}\frac{1}{N}\xi_{k} \Pi_k\Big(\int\nolimits_{0}^{t}N\nu_{k}(s,\XrelN(s))ds\Big).
\end{align}
The normalized state process $\XrelN$ exhibits similar behavior to $X^{N}$, but the jump sizes are scaled by a factor of $1/N$. Thus, the jumps of $\XrelN$ decrease as $N$ increases, and it is to be expected that $\XrelN$ approaches a continuous function in the limit case for $N\to\infty$. This intuitive concept can be formulated mathematically in a rigorous way, as in  \citep{kurtz1971limit,kurtz2009limit,kurtz1978strong} for the case of time-independent intensities of Poisson processes. Generalized results that also take time-dependent intensities into account were obtained by \citet{guy2015approximation,guy2016approximation} and \citet{narci2021inference}.

For  a standard Poisson process $\Pi_k$, $k=1,\ldots,K$ with unit intensity, the  so-called compensated process  $\widetilde{\Pi}_{k}$ defined by $\widetilde{\Pi}_{k}(t)=\Pi_{k}(t)-t,  t\in[0,T]$, is known to be a continuous-time martingale. Based on that, we can again  rewrite the dynamics of $\XrelN$  for $t\in[0,T]$ as
\begin{align}
	\begin{split}	
	\XrelN(t)&=\XrelN(0)+\int\nolimits_{0}^{t}F(s,\XrelN(s))ds+\sum\limits_{k=1}^{K}\frac{1}{N}\xi_{k} \widetilde{\Pi}_{k}\Big(\int\nolimits_{0}^{t}N\nu_{k}(s,\XrelN(s))ds\Big) ,\\[-0.5ex]
	\text{where}~~ F(t,x)&=\sum\limits_{k=1}^{K}\nu_{k}(t,x)\xi_{k}.
\end{split}
	\label{Fq_def}
\end{align}
This semi-martingale representation allows  us to derive the following asymptotic result. For the  proof we refer to \citet[Chapter 2, Theorem 2.2.7]{britton2019stochastic}.
\begin{theorem}[Law of Large Numbers]\label{LLN_Theo}
	  Let Assumption \ref{ass_DA} hold, and let $\xrel_0  \in [0,1]^{d}$ be such that  $\XrelN(0)\to \xrel_0  $ as $N\to \infty$. Then $\XrelN(t)$ converges almost surely and uniformly on $[0,T]$ to the solution $\Xrelinfty(t)$ of the ODE
	\begin{align}
		\label{ODE_LLN}
		\dfrac{d\Xrelinfty(t)}{dt}=F(t,\Xrelinfty(t)), ~~~~\Xrelinfty(0)=\xrel_0  .
	\end{align}
\end{theorem}

\subsection{Central Limit Theorem}
In this section, we examine random deviations of the normalized state process $\XrelN$ from its deterministic limit $\Xrelinfty$.
To this end, we introduce $\Gauss^{N}$ as the scaled deviation between $\XrelN$ and $\Xrelinfty$, defined by 
\begin{equation}
	\Gauss^{N}(t)=\sqrt{N}\Big(\XrelN(t)-\Xrelinfty(t)\Big). \label{Gaussian_Proc}
\end{equation}
From \eqref{Fq_def} and Theorem \ref{LLN_Theo} it follows that the above relation can be rewritten as 
\[ \Gauss^{N}(t)=\Gauss^{N}(0) +\sqrt{N}\int\nolimits_{0}^{t}\Big[F(s,\XrelN(s))-F(s,\Xrelinfty(s))\Big]ds + \sum_{k=1}^{K}\xi_{k}\widetilde{\mathcal{N}}_{k}^{N}(t),\]
where
$\Gauss^{N}(0)=\sqrt{N}\big(\XrelN(0)-\Xrelinfty(0)\big)$ and $\widetilde{\mathcal{N}}_{k}^{N}(t)={N}^{-1/2}\widetilde{\Pi}_{k}
\big(N\int\nolimits_{0}^{t}\nu_{k}(s,\XrelN(s))ds\big).$

In order to formulate the following limit theorems, we introduce the space \(\mathbb{D}([0,T],\mathbb{R}^{d})\) of cadlag functions, i.e., functions that are right-continuous and have left limits on \([0,T]\). This space is endowed with the Skorokhod topology, which provides a natural framework for convergence of stochastic processes with jumps or discontinuities; see \citet{guy2016approximation} and \citet[Chapter 6]{jacod2013limit}.

Applying the central limit theorem for the Poisson process leads to the following
\begin{proposition}\label{prop_CLT0}
	 As $N\to \infty$, the stochastic processes $\big\{\widetilde{\mathcal{N}}_{k}^{N}(t)\big\}_{t\in[0,T]}$ converge in distribution in \(\mathbb{D}([0,T],\mathbb{R})\), endowed with the Skorokhod topology, to $\big\{\widetilde{\mathcal{N}}_k(t)\big\}_{t\in[0,T]}$, where
	$$\widetilde{\mathcal{N}}_{k}(t)=\int\nolimits_{0}^{t}\sqrt{\nu_{k}(s,\Xrelinfty(s))}dW_{k}(s), \quad k=1,\ldots,K,$$
	 and $W_{1}(t),\ldots,W_{K}(t)$ are mutually independent standard Brownian motions.
\end{proposition}
\begin{pf}
	The proof of this proposition is mainly based on the central limit theorem for the Poisson process and can be found in  \citep[Chapter 2]{britton2019stochastic} for  time-homogeneous intensities,  and  in \citep{guy2015approximation,guy2016approximation} for the general case.
\end{pf}
Based on this proposition, the main result of this section can be derived.  The following central limit theorem  states the convergence of the stochastic process \((\Gauss^{N})\) to a Gaussian process as \( N \to \infty \). 

\begin{theorem}[Central Limit Theorem]\label{theo:CLT}
	 Let the assumptions of Theorem \ref{LLN_Theo} hold  with $\sqrt{N}\big(\XrelN(0)-\xrel_0\big)\to 0$ as $N\to\infty$ (in particular for $X^N(0)=\lfloor N\xrel_0\rfloor$). Then $\Gauss^{N}$ converges in distribution in \(\mathbb{D}([0,T],\mathbb{R}^d)\), endowed with the Skorokhod topology, to the stochastic process $\Gauss$ defined by
	\[ \Gauss(t)=\int\nolimits_{0}^{t}\nabla_x F(s,\Xrelinfty(s))\Gauss(s)\,ds+
	\sum\limits_{k=1}^{K}\xi_{k}\int\nolimits_{0}^{t}\sqrt{\nu_{k}(s,\Xrelinfty(s))}\,dW_{k}(s), ~~ t\in[0,T]. \] 
	 Here, $W_{1}(t),\ldots,W_{K}(t)$ are mutually independent standard Brownian motions, and $\nabla_x$ denotes the gradient w.r.t.~ $x$.
\end{theorem}

\begin{pf}
	For the proof we refer to \citet[Chapter 2, Section 2.3]{britton2019stochastic} and \citet[Section 3]{anderson2011continuous},  which consider time-homogeneous transition rates. The proof for the general case of time-dependent intensities can be found in \citep{guy2015approximation,guy2016approximation}.
\end{pf}
\ifthenelse{\preprintOn=0}{In \citet[Section~{3.3}]{ouabo2026partial} we present some results of a numerical simulation study for a SIRS model that illustrate the above-mentioned limit theorems for large populations.}{}
\ifthenelse{\preprintOn=0}{}{{\prepcolor
\subsection{Numerical Illustration}\label{SIRS_CTMC_Model}
\begin{table}[htbp]
	\centering	
	\begin{minipage}{0.35\textwidth}
		\includegraphics[width=1\textwidth]{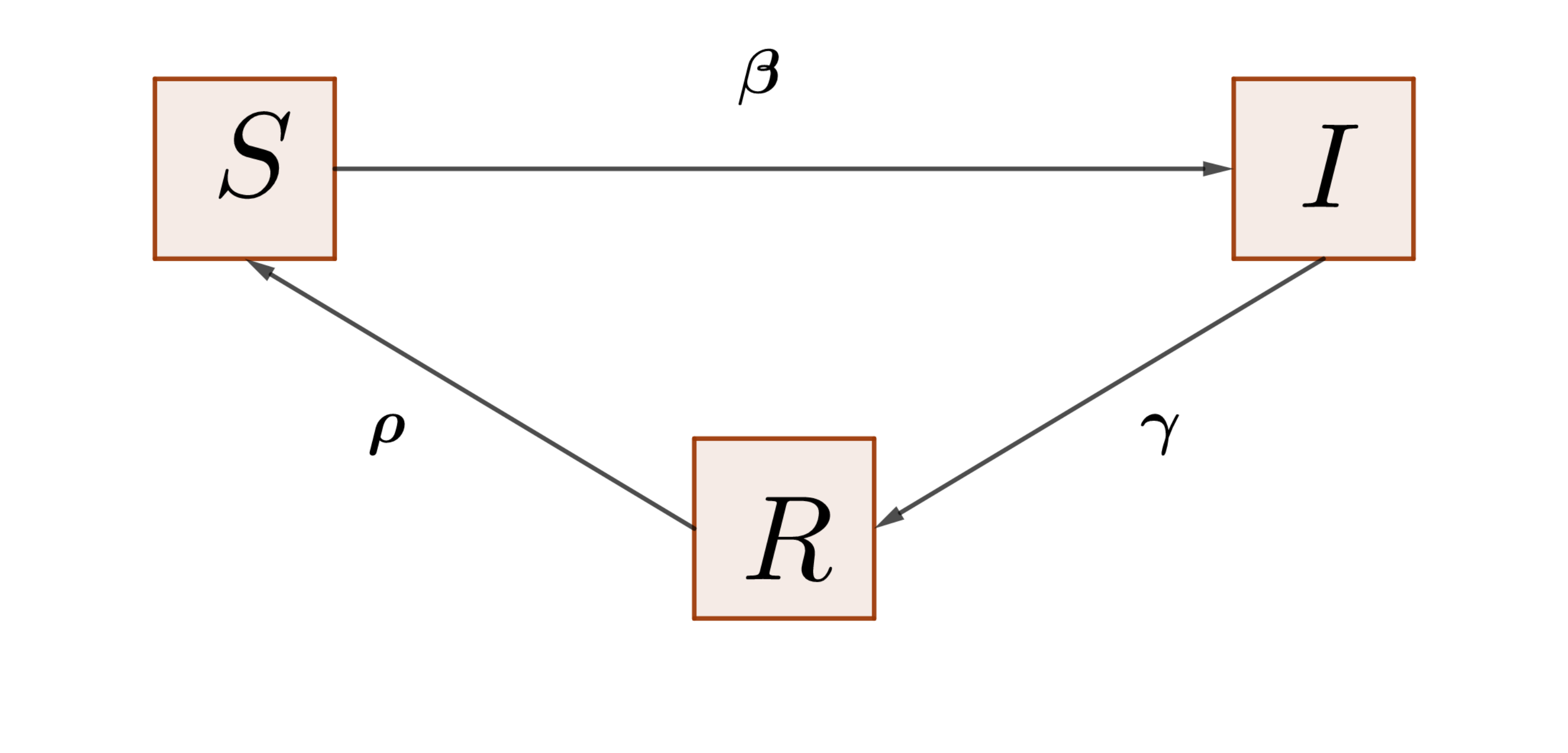}	
	\end{minipage}
	\begin{minipage}{0.64\textwidth}
		 \small \setlength{\tabcolsep}{5pt}
		\begin{tabular}[t]{l|ll|c|c}
			\hline
			k & Transition & &Transition vector $\xi_{k}$& intensity $\lambda_{k}(t,X)$ \\ 
			\hline
			&&&\\[-1em]
			1& Infection& $S\to I$   & $(\phantom{-}1,\phantom{-}0,-1)^{\top}$ &  {$\beta(t) S\frac{I}{N}=\beta(t)\, X_{3}\frac{X_{1}}{N}$}
			\\[1em] 
			\hline
			&&&\\[-1em]
			2 & Recovery & $I\to R$   & $(-1,\phantom{-}1,\phantom{-}0)^{\top}$ & {$\gamma I=\gamma X_{1}$}
			\\
			\hline
			&&&\\[-1em]
			3 & Losing immunity & $R\to S$  & $(\phantom{-}0,-1,\phantom{-}1)^{\top}$ & {$\rho R=\rho X_{2}$}
			\\
			\hline	
		\end{tabular}			
	\end{minipage}
	\captionof{figure}{$SIRS$ model with  compartments, transition vectors, and transition intensities.\newline 
		State process $X=(X_1,X_2,X_3)^\top=(I,R,S)^\top$,  number of states $ d= 3 $, number of transitions $ K = 3 $. }
	\label{SIRS_Model}
\end{table}

In this subsection, we present some results of a numerical simulation study for a SIRS model shown in Figure \ref{SIRS_Model}. It shows simulated paths of the CTMC state process, followed by a visualization of its asymptotic behavior for $N\to\infty$, which is  governed by the law of large numbers and the central limit theorem.
 
For the simulations, we consider a time-dependent infection rate given by \( \beta(t) = \beta_{0}(1 + A \sin({2\pi t}/{T^{\beta}})) \), where \( \beta_{0} \) represents the baseline infection rate, \( A \in (0,1)\) is the amplitude of the seasonal effect, \( T^{\beta} \) is the period of the seasonal variation in transmission, and \( t \) denotes time. For the model parameters, we use a constant  recovery rate \( \gamma \) and  loss of immunity rate \( \rho \),  both fixed at \( \gamma = \rho = 0.05 \), while \( \beta_{0} = 0.2 \), \( A = 0.4 \), and \( T^{\beta} = 365 \).

\begin{table}[htbp]
	\centering
	\includegraphics[width=0.7\textwidth]{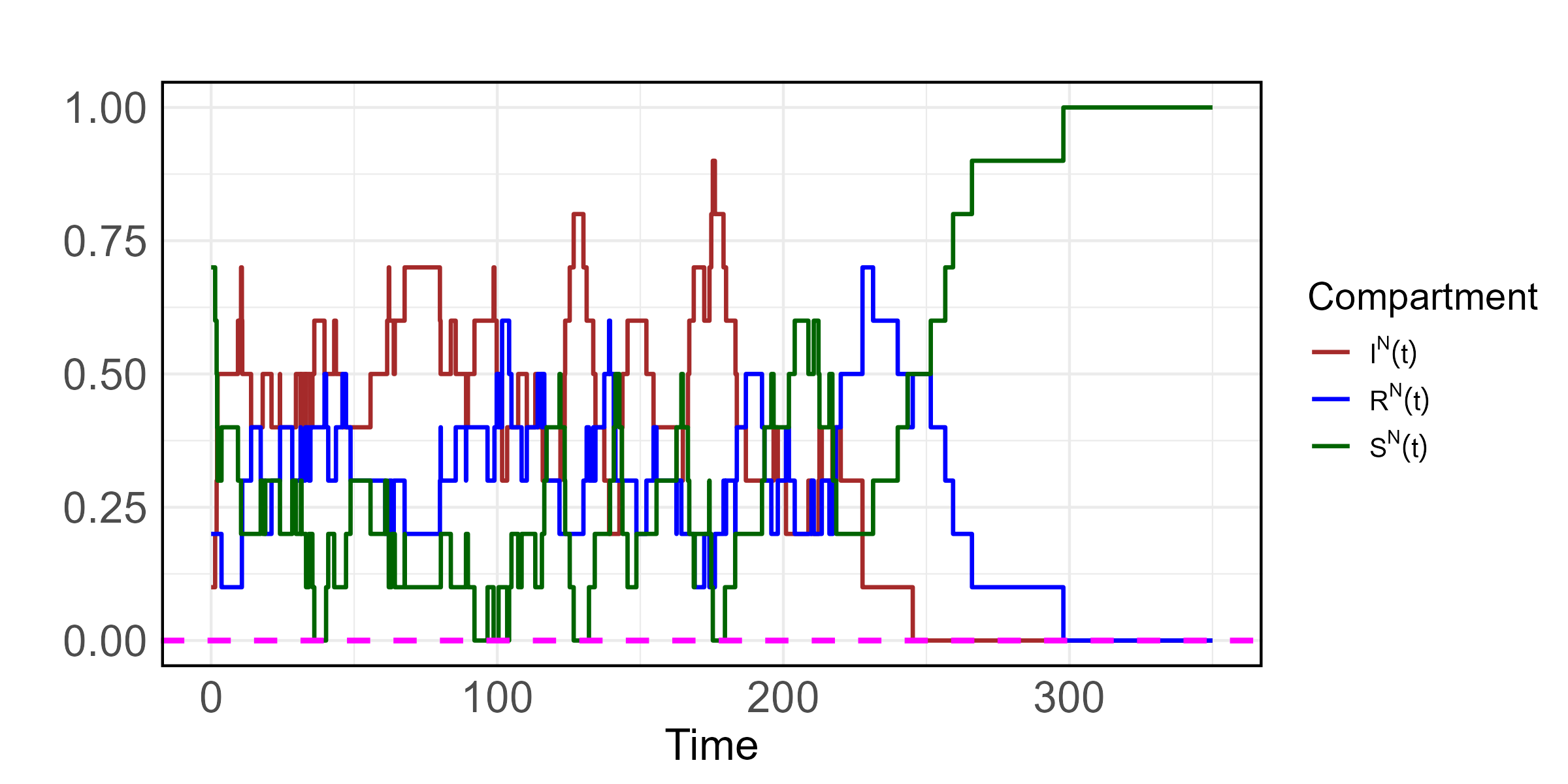}
	\captionof{figure}{{SIRS model with population size $N=10$: Simulated CTMC paths of $S^N,I^N,R^N$.}}
	\label{CTMC_SIRS_Simu}
\end{table}

Based on the  Gillespie Algorithm \ref{alg:simul_CTMC}, we generated and  visualized the paths of this SIRS model presented in  Figure \ref{CTMC_SIRS_Simu}. Note that a small population size ($N=10$) was used to clearly illustrate the different jumps that occur over time. 

Next, we focus on the infected compartment and show in the left panels of Figure \ref{Chap2_LLN} how the paths of the normalized state process $\IrelN$ (which are  step functions) become smoother as the population size $N$ increases and eventually approach the limit $\Irelinfty$. The latter is the solution of the ODE given in the law of large numbers,  see Theorem \ref{LLN_Theo}.
To illustrate the central limit theorem, in the right panels of Figure \ref{Chap2_LLN}, we show paths of the first component of $\Gauss^{N}$ defined in \eqref{Gaussian_Proc}, i.e., the scaled deviations $\sqrt{N}(\IrelN(t)- \Irelinfty(t))$.
While the proportion $\IrelN$ approaches the deterministic limit $\Irelinfty$ as $N$ increases, the scaled deviations have no pathwise limit: they keep fluctuating around $0$ and converge in distribution to a Gaussian process, see Theorem \ref{theo:CLT}.
\begin{figure}[pos=!htbp]
	\centering
	\subfloat[\(10\) paths of $\IrelN$ with $N=10$  ]{\includegraphics[width=0.48\textwidth]{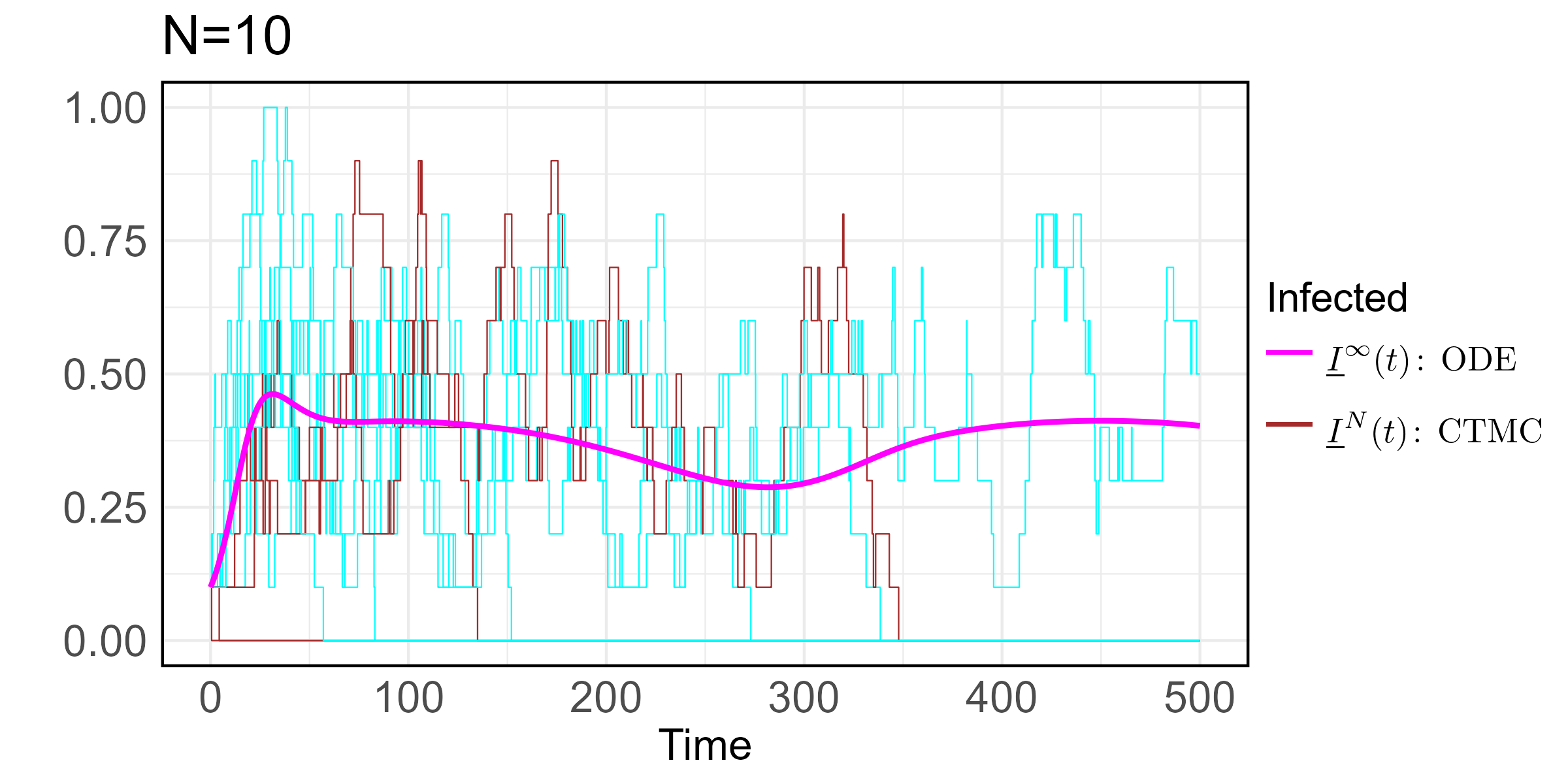}\label{fig:f51}}
	\hfill
	\subfloat[$ 10 $  paths of $\sqrt{N}(\IrelN- \Irelinfty)$ with $N=10$]{\includegraphics[width=0.48\textwidth]{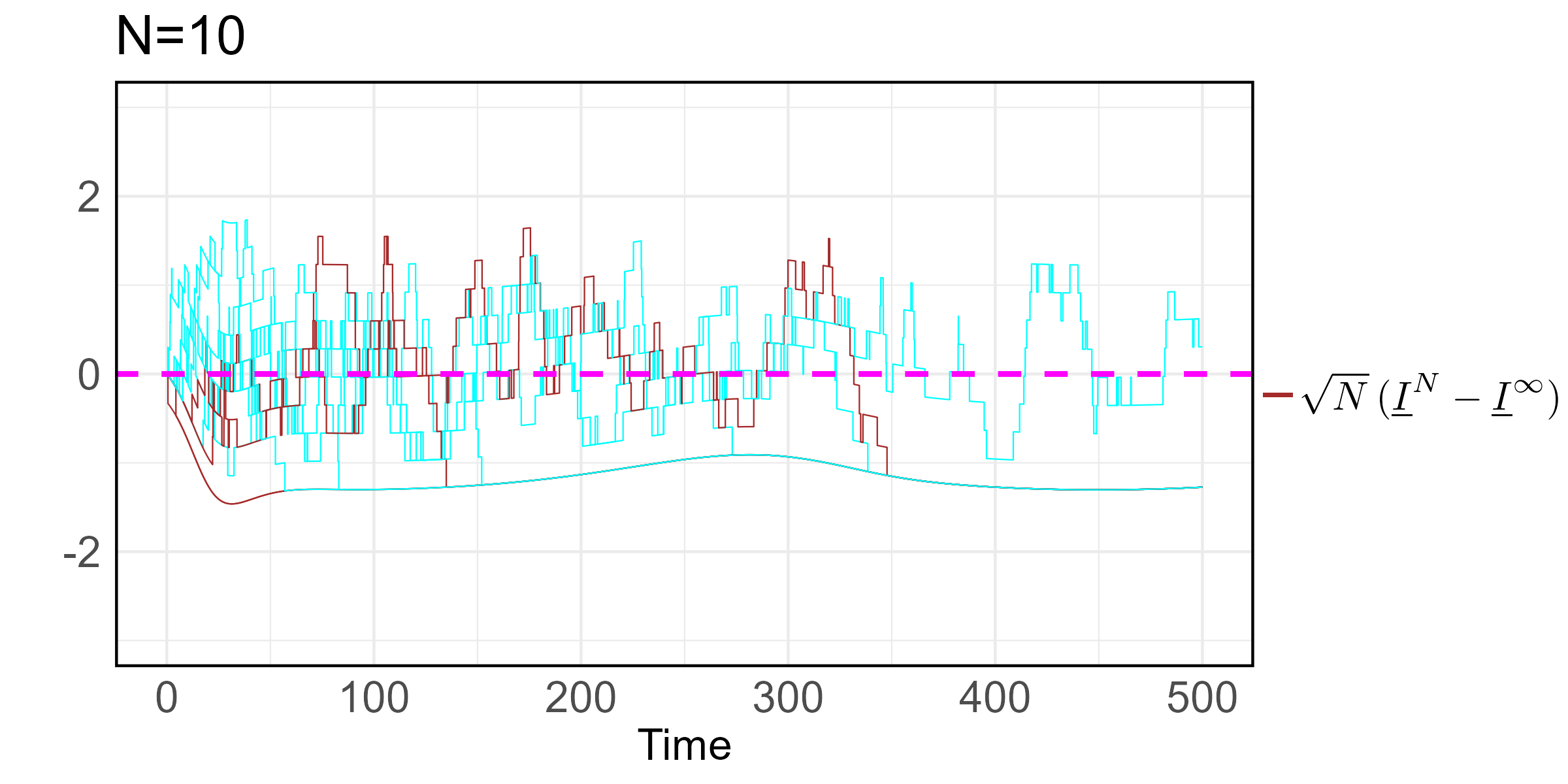}\label{fig:f61}}\\
	\subfloat[\(10\) paths of $\IrelN$ with $N=50$  ]{\includegraphics[width=0.48\textwidth]{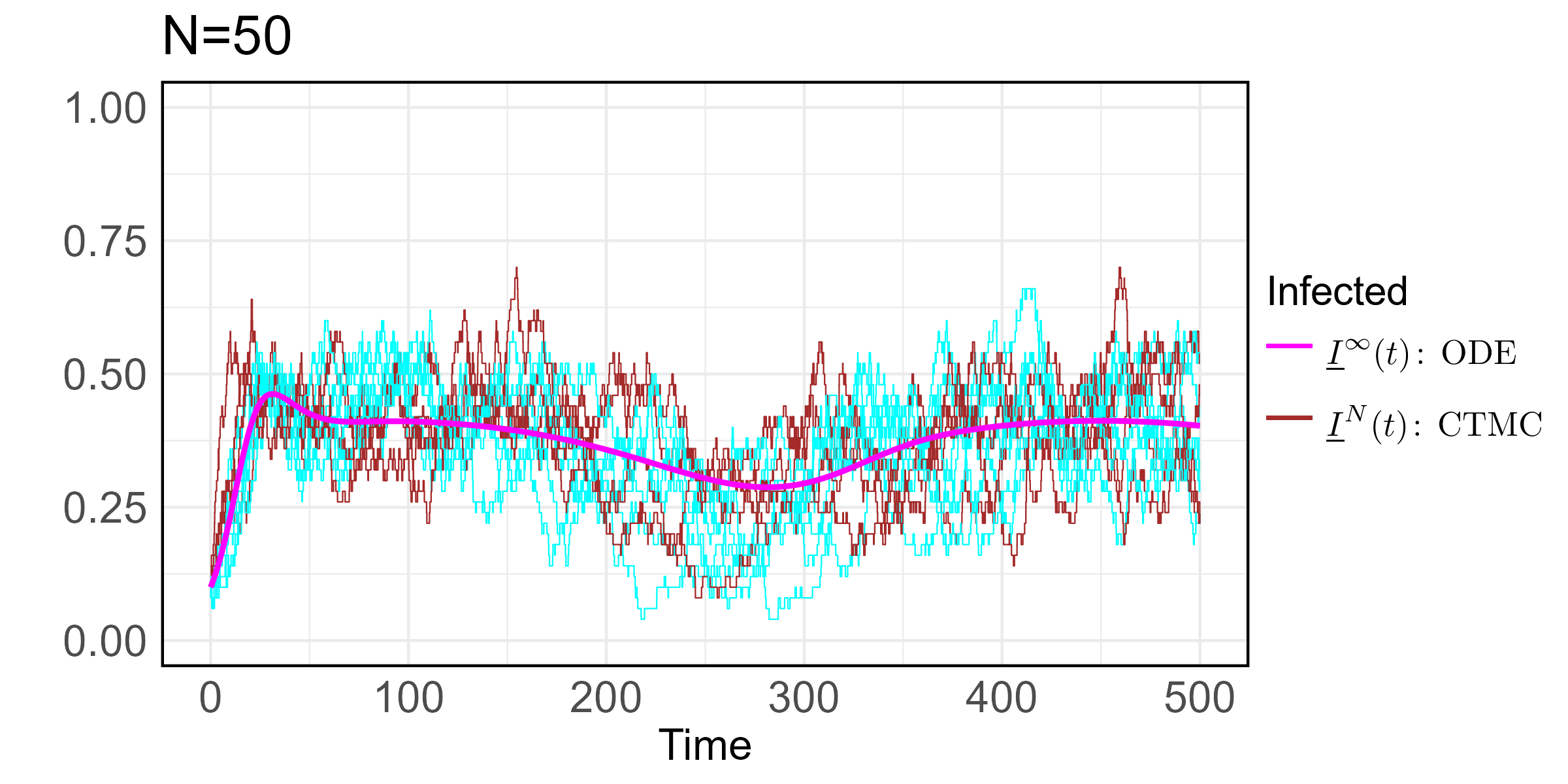}\label{fig:f511}}
	\hfill
	\subfloat[$ 10 $  paths of $\sqrt{N}(\IrelN- \Irelinfty)$ with $N=50$]{\includegraphics[width=0.48\textwidth]{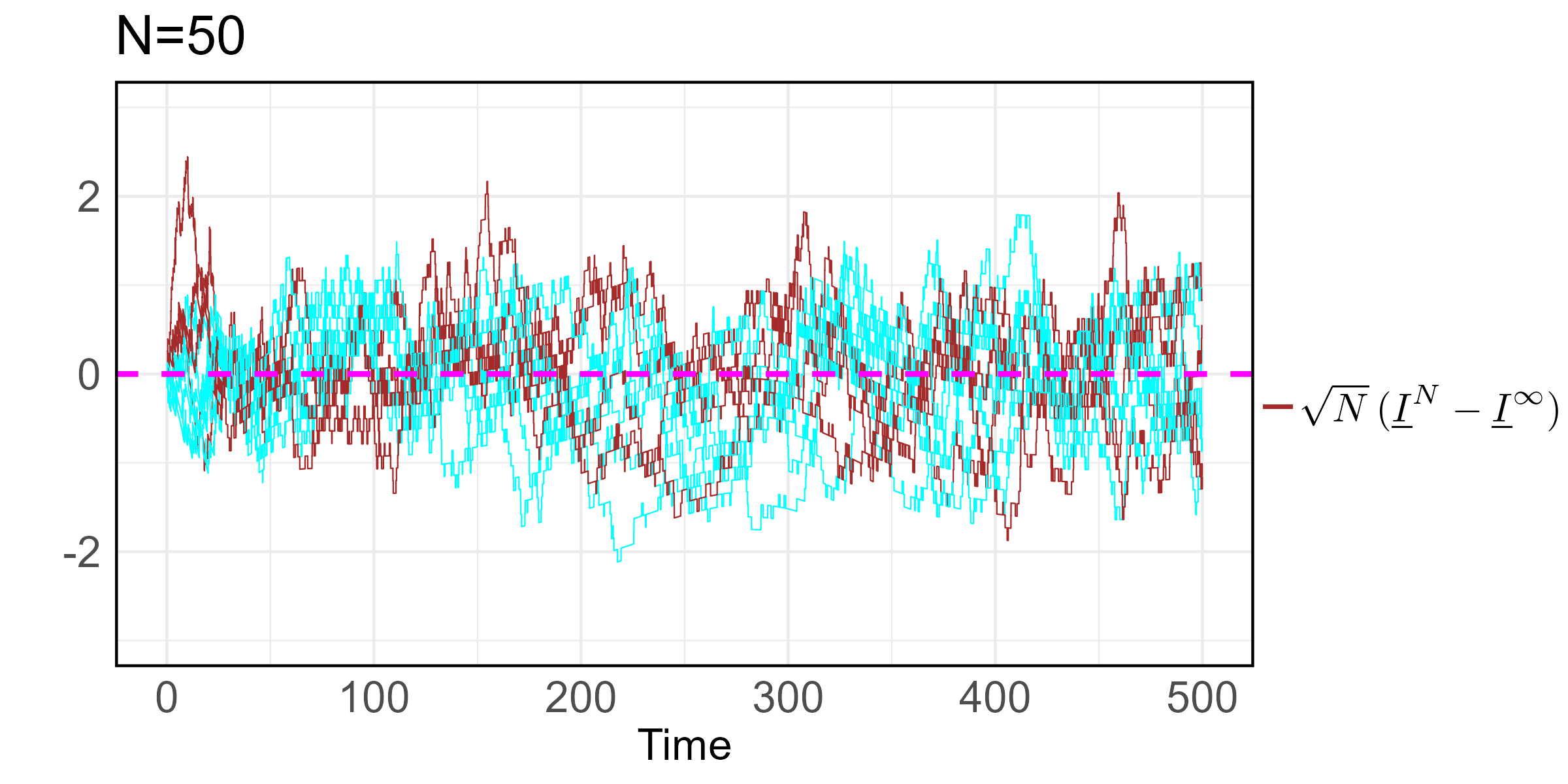}\label{fig:f611}}\\
	\subfloat[\(10\) paths of $\IrelN$ with $N=500$  ]{\includegraphics[width=0.48\textwidth]{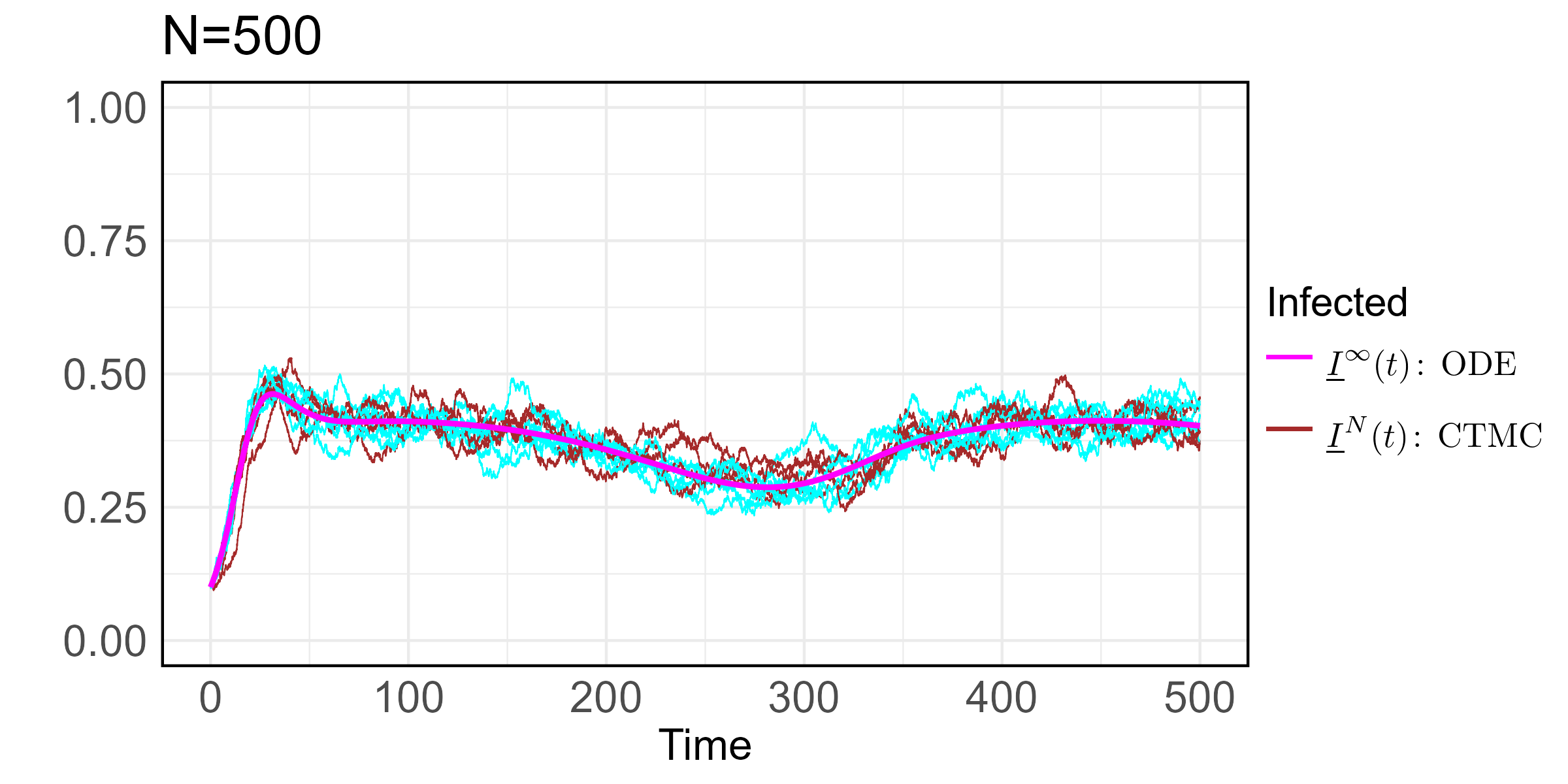}\label{fig:f41}}
	\hfill
	\subfloat[$ 10 $  paths of $\sqrt{N}(\IrelN- \Irelinfty)$ with $N=500$]{\includegraphics[width=0.48\textwidth]{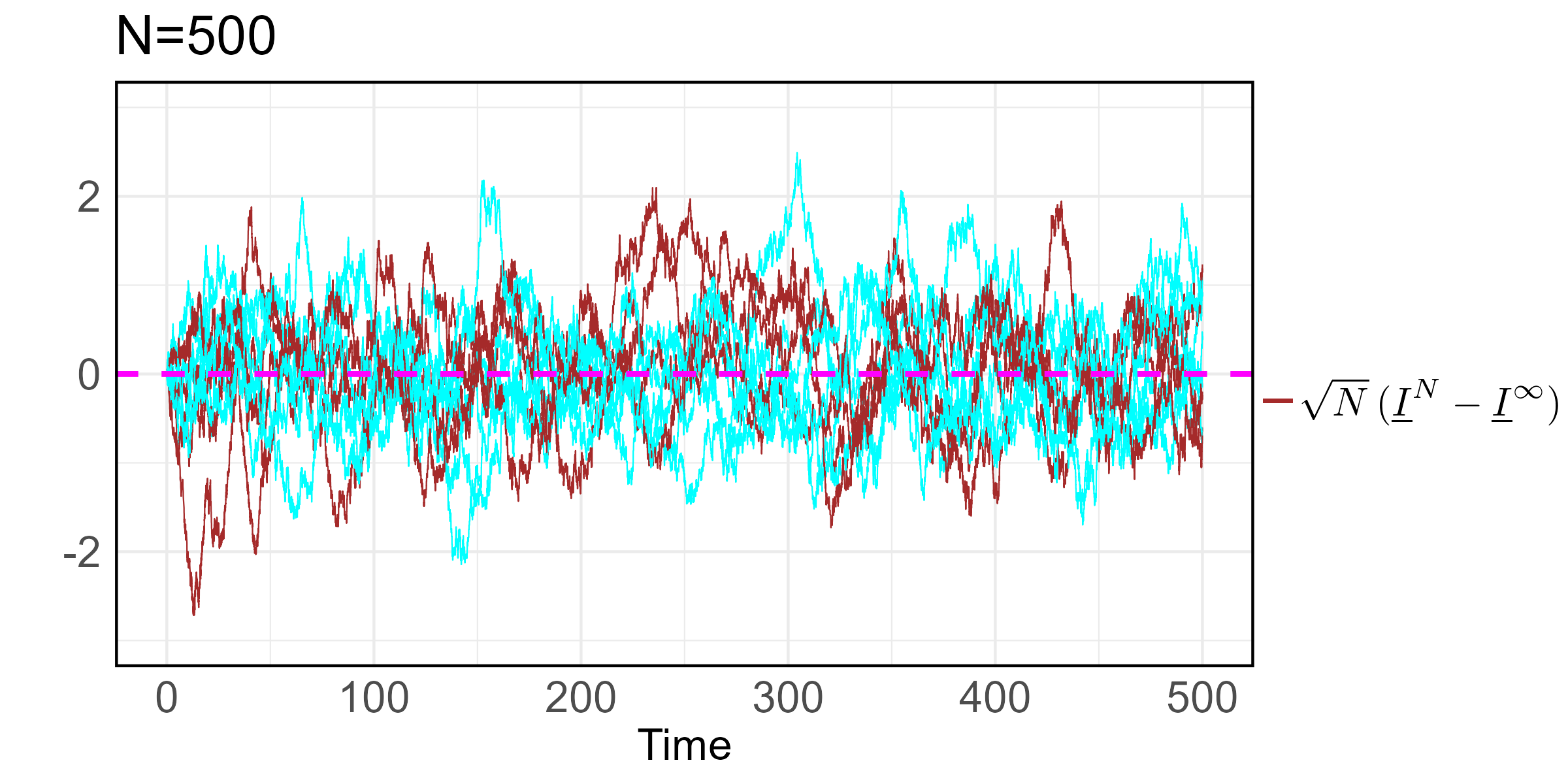}\label{fig:f81}}\\
	\subfloat[\(10\) paths of $\IrelN$ with $N=10000$  ]{\includegraphics[width=0.48\textwidth]{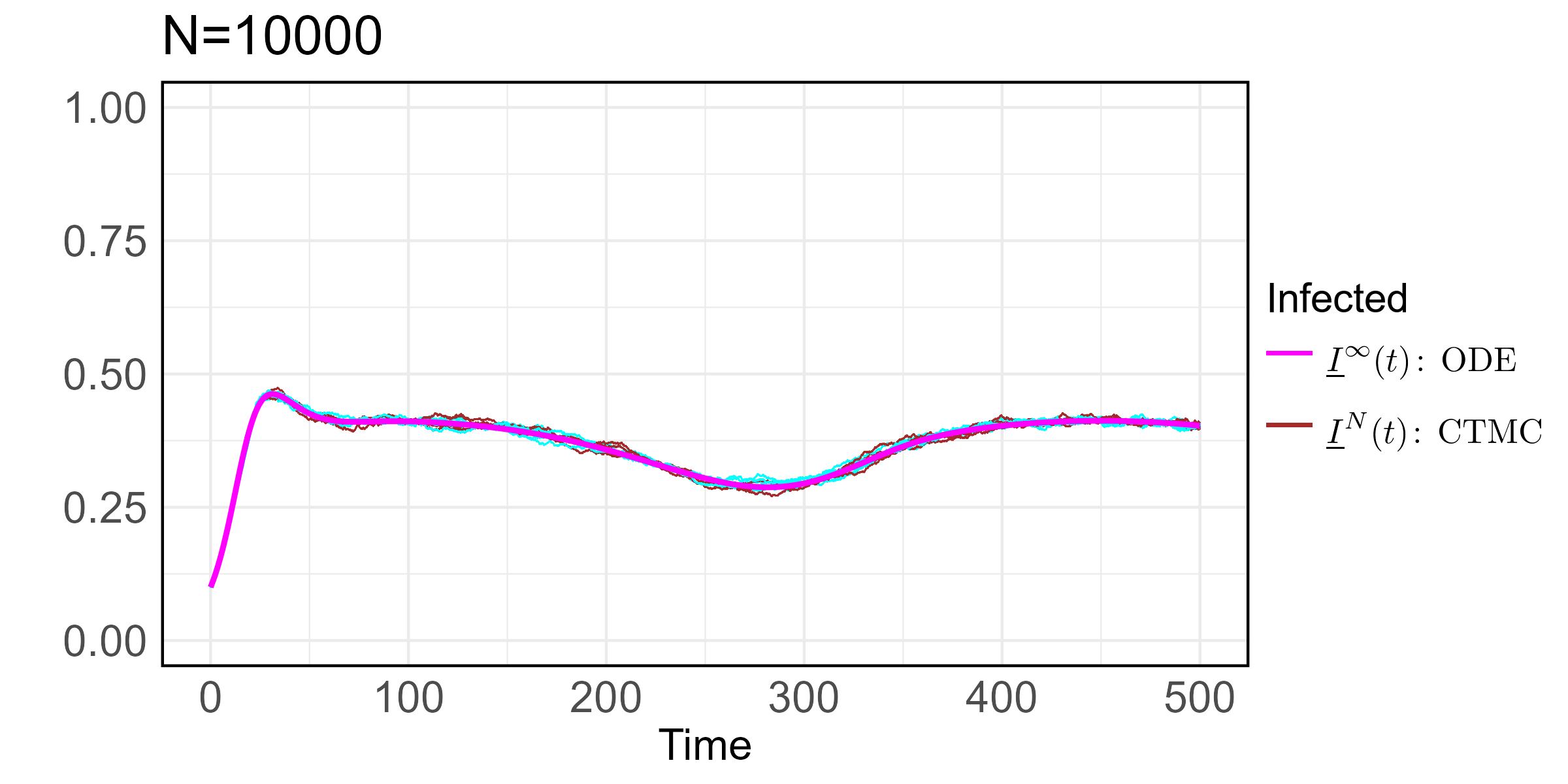}\label{fig:f101}}
	\hfill
	\subfloat[$ 10 $  paths of $\sqrt{N}(\IrelN- \Irelinfty)$ with $N=10000$]{\includegraphics[width=0.48\textwidth]{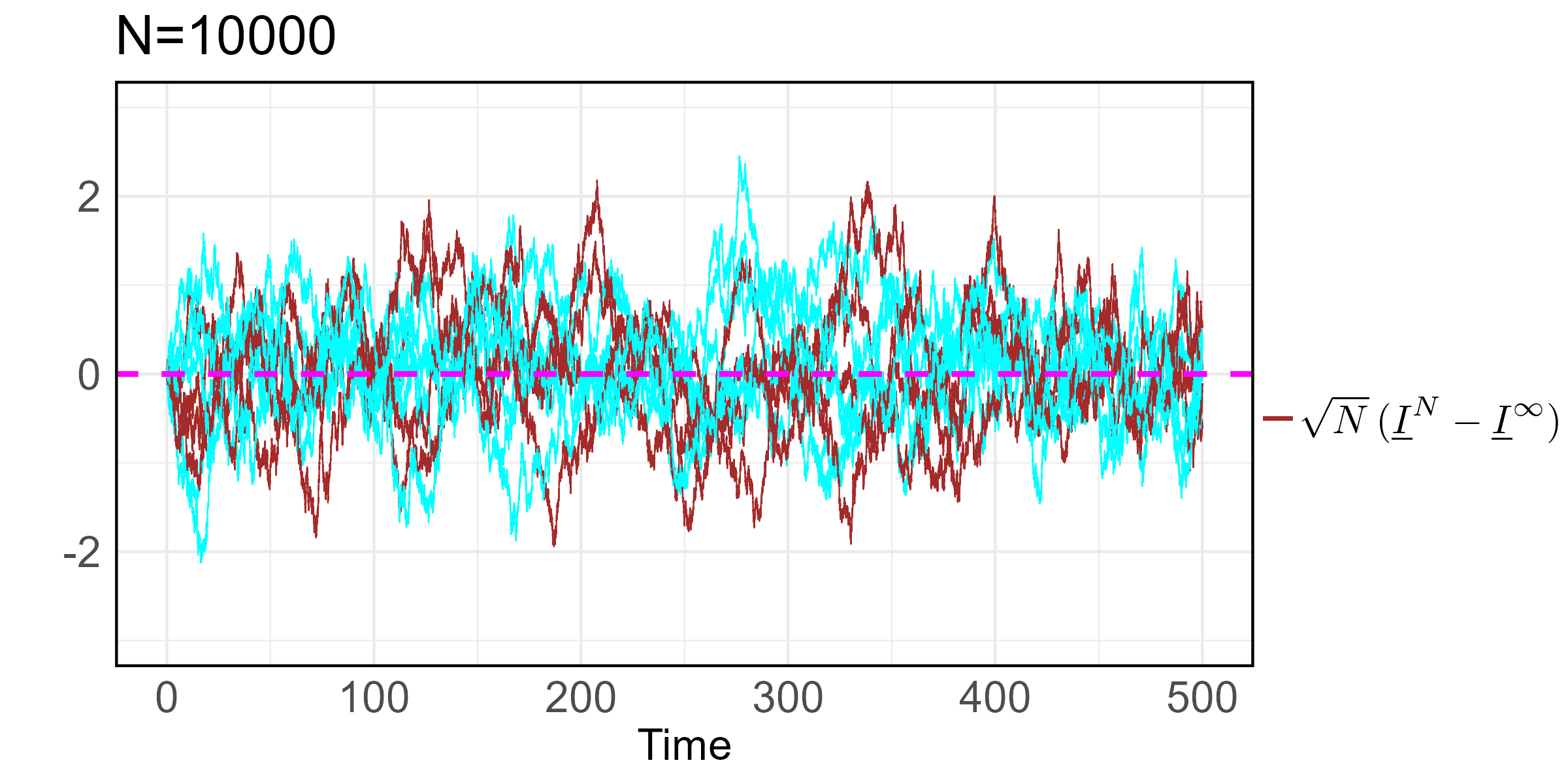}\label{fig:f91}}
	\caption{
		Left panel: Visualization of the  law of large numbers  in the SIRS model. As the population size \( N \) increases, the sample paths of  \( \IrelN \) progressively become smoother and approach the deterministic trajectory of \( \Irelinfty\).  \newline
		Right panel: Visualization of the central limit theorem. As the population size \( N \) increases, the scaled deviations of the sample paths of \( \IrelN \) from their deterministic limit \( \Irelinfty(t) \) approach a Gaussian process with zero mean, reflecting stochastic fluctuations around the deterministic solution.  
	}
	\label{Chap2_LLN}
\end{figure}
}}

\subsection{Diffusion Approximation and SDE Models}\label{Section_Diffusion_Approx}
The central limit theorem allows us to approximate the CTMC state process $X^N$  for  finite $N$ by the solution of a suitably chosen SDE. Starting point is the  normalized state process  $\XrelN$ for which \citet[Section 2.4, Theorem  2.4.1]{britton2019stochastic} consider the diffusion process $\XDrel$ satisfying the SDE 
\begin{align}\label{SDE_XDrel}
	d\XDrel(t)={F}\big(t,\XDrel(t)\big)dt+{\dfrac{1}{\sqrt{N}}}\diffXrel\big(t,\XDrel(t)\big)\,dW(t), \quad 
	\XDrel(0)=\xDrel_0= \frac{1}{N}x_0.
\end{align}
Here, $W$ is a $K$-dimensional standard Brownian motion, the drift coefficient $F$ is given in \eqref{Fq_def}, and diffusion coefficient is given for $(t,u)\in [0,T]\times [0,1]^d$ by 
\begin{align}
	\diffXrel(t,\xDrel)=\Big(\xi_{1}\sqrt{\nu_{1}(t,\xDrel)},\ldots,\xi_{K}\sqrt{\nu_{K}(t,\xDrel)}\Big).
\end{align}
In  \cite{britton2019stochastic} it is proven that $ \sqrt{N}(\XrelN-\XDrel)$ converges to $0$ in the Wasserstein metric. Equivalently, the Wasserstein distance between $\XrelN$ and $\XDrel$ vanishes faster than $1/\sqrt{N}$, so $\XDrel$ provides a diffusion approximation of $\XrelN$ for large $N$.

The associated diffusion approximation of the state process $X^N=N\XrelN$ describing the absolute subpopulation sizes is  obtained by $\XD=N\XDrel$ for which one obtains the SDE
\begin{align}\label{DA1}
	\begin{split}
		d\XD(t)&=\driftX\big(t,\XD(t)\big)dt+\diffX\big(t,\XD(t)\big)\, dW(t),\quad 
		\XD(0)=x_0.
	\end{split}
\end{align}
The drift coefficient  and diffusion coefficient are given for $(t,x)\in [0,T]\times [0,N]^d$ by 
\begin{align}
	\driftX(t,\xD)=\sum_{k=1}^{K}\xi_{k}\lambda_{k}(t,\xD) 
	\quad \text{and} \quad  
	\diffX(t,\xD)=\Big(\xi_{1}\sqrt{\lambda_{1}(t,\xD)},\ldots,\xi_{K}\sqrt{\lambda_{K}(t,\xD)}\Big).
\end{align}
Here we used the scaling property \eqref{H2r} of the intensities $\lambda_k$ and $\nu_k$ from Assumption  \ref{ass_DA}, i.e., $\lambda_{k}(t, N\xDrel)=N\nu_{k}(t,\xDrel)$. It is fulfilled in the model examples considered in Section \ref{Model_example} below. If \eqref{H2r} holds only asymptotically for $N\to\infty$, as discussed in Remark  \ref{rem:scaling_prop}, one has to accept an additional approximation error that vanishes only for  $N\to\infty$.
\begin{remark}~
\begin{enumerate}
\item 
Contrary to the SDE \eqref{SDE_XDrel} for $\XDrel$, the coefficients of SDE \eqref{DA1} for $\XD$  contain no separate $N$-dependent scaling of the noise terms. Thus $\XD$ serves as a diffusion approximation of $X^N$  for all $N$, and we omit the superscript $N$ in the notation. Note that the approximation property  only holds for the normalized version, i.e.,  the Wasserstein distance of  $\XrelN$ and $\XDrel$ goes faster to zero than $1/\sqrt{N}$.
\item Note that the ODE \eqref{ODE_LLN} for $\Xrelinfty$ given in the law of large numbers, see Theorem  \ref{LLN_Theo}, is formally contained in SDE \eqref{SDE_XDrel} for $\XDrel$  as limiting case for $N\to\infty$, as the diffusion term vanishes.
\item
From the perspective of model calibration, it should be emphasized that it is sufficient to specify the transition vectors $\xi_k$ and the transition intensities $\lambda_k$, $k=1,\ldots,K$, or their scaled counterparts $\nu_k$, to parameterize the CTMC model \eqref{Eq_CTMC} of the state process $X$, the ODE model \eqref{ODE_LLN} of $\Xrelinfty $ as well as the SDE models \eqref{SDE_XDrel} and \eqref{DA1}, which describe the diffusion approximations $\XDrel$ and $\XD$, respectively. In particular, the transition from the deterministic ODE model \eqref{ODE_LLN} to the SDE models \eqref{SDE_XDrel} and \eqref{DA1} does not require any new parameters, since the diffusion coefficients in the SDEs depend exclusively on the vectors  $\xi_k$ and the intensities $\lambda_k$ (or $\nu_k$).
\end{enumerate}
\end{remark}

\section{Partially Observable Models}\label{Section_Time_Discr}
After providing an introduction to stochastic modeling of epidemics in the previous two sections, this section focuses on models in which some of the states cannot be directly observed. This is a typical case for many diseases in which infections can be asymptomatic,  remain undetected or are not fully reported, or in which immunity is lost over time after recovery without being noticed. The assumption made in many studies that the current value of all components of state $X$ is known is therefore unrealistic. Instead, one is confronted with so-called 	``nowcast uncertainty'', which describes the lack of accuracy of the current state, as well as the  ``dark figure'' problem inherent in many disease courses. It therefore becomes difficult to  compute reliable  characteristics of an epidemic, such as the effective reproduction number	or  the prevalence of the infection within the population. Furthermore, the management and containment of epidemics, which are based on solutions to optimal control problems and the associated feedback controls, require information about the current state as input. To solve these problems, hidden or unobservable states must be estimated based on observable variables. Mathematically, this can be treated as a filtering problem.
	
To effectively solve such stochastic filter problems, it is helpful to describe the state of the epidemic using SDE models from diffusion approximations.  	Then one can rely on the well-developed Kalman filter theory and construct efficient estimators for the unobservable states; see \cite{ouabo2025estimating,achola2025stochastic,mbouandi2025convergence},  which can also be integrated very elegantly into the treatment of stochastic optimal control problems; see \cite{mbouandi2025stochastic}. Such filter techniques also exist for CTMC models, but their practical application usually fails due to the curse of dimensionality. For these reasons, we will work below with the diffusion approximation $\XD$ of the CTMC state process $X$ given in \eqref{DA1}.

\myparagraph{State Decomposition}  
We  distinguish between observable and hidden states, and assume that among the $d$ states there are $d_1$ states that are hidden, where  $d_1\in \N, d_1<d$. That is, their actual subpopulation size is not directly observable, and has to be estimated from the observable states. They are also referred to as the ``dark figures''. The remaining  $d_2=d-d_1$ states are supposed to be observable. The $d$-dimensional state vector $\XD$ is then decomposed into a $d_1$-dimensional vector $Y$, whose entries are the hidden states, and a $d_2$-dimensional vector $Z$, which contains the observable states. The SDE \eqref{DA1} for $\XD$  can then be rewritten in the form of the following system of  coupled SDEs
	\begin{align}\label{fp11} 
		\begin{array}{rlrl}
			dY(t) &= \overline{f}(t, Y(t), Z(t))dt + \overline{\sigma}(t, Y(t), Z(t))dW^{1}(t) + \overline{g}(t, Y(t), Z(t))dW^{2}(t), &Y(0) &= y_0,  \\[0.5ex]			
			dZ(t) &= \overline{h}(t, Y(t), Z(t))dt + \overline{\ell}(t, Y(t), Z(t))dW^{2}(t),  &	 Z(0) &= z_0. 
		\end{array}
	\end{align}
In terms of the drift and diffusion coefficients $\driftX, \diffX$ of the SDE \eqref{DA1}, the coupled system \eqref{fp11} corresponds to the following block decomposition:
\begin{align}\label{block_SDE}
	\driftX=\begin{pmatrix}\overline{f} \\ \overline{h}\end{pmatrix}, \quad
	\diffX =\begin{pmatrix}\overline{\sigma} & \overline{g}\\ 0_{d_2\times k_1} & \overline{\ell}\end{pmatrix}, \quad
	W=\begin{pmatrix} W^1 \\ W^2\end{pmatrix}, \quad
	X=\begin{pmatrix} Y \\ Z\end{pmatrix}.
\end{align}
Here, the first SDE in \eqref{fp11} governs the hidden process \( Y \), also referred to as the signal process, while the second SDE  describes the observation process \( Z \). The Brownian motions \( W^{1} \) and \( W^{2} \) in \eqref{fp11} are independent Wiener processes so that \( W^{1} \) contains those $k_1\in\{1,\ldots,K-1\}$ entries of $W$  that only appear in the  SDE for the hidden process $Y$, while $W^2$ collects the $k_2=K-k_1$ entries of $W$ that appear in the  SDEs for both $Y$ and $Z$.  The drift coefficients $\overline{f}, \overline{h}$ collect the corresponding entries of the drift coefficient $\driftX$ in \eqref{DA1}. The same applies to the diffusion coefficients $\overline{\sigma},\overline{g},\overline{\ell}$, which are derived from $\diffX$. Note that all these coefficients usually depend on both the hidden state  $Y$ and the observable state  $Z$. 

\begin{remark}
 In order to obtain a meaningful partially observable model, the number $d_1$ of hidden states should be at least $2$. A single hidden state ($d_1=1$) can always be completely reconstructed from the normalization property $X_1+\ldots+X_d=N$, since we assume a constant total population size $N$. Further,  to allow the estimation of hidden states from observable quantities,  the model should contain at least one  observable state, i.e. $d_2\ge 1$. Thus  we obtain $d=d_1+d_2\ge 3 $ for the dimension $d$ of the state.
	
 In many examples of epidemic models, including our examples in Section \ref{Model_example}, there are compartments ``susceptible'' and (asymptomatic or undetected) ``infected'', which represent two hidden states, and there is also a transition between them, in this case the infection of a susceptible person. In the SDE model \eqref{fp11}, the randomness of such hidden-to-hidden transitions is captured by the diffusion term driven by the $k_1$-dimensional Wiener process $W^{1}$, which only appears in the SDE of the hidden state $Y$. Therefore, it makes sense to assume $k_1\ge 1$. On the other hand, the stochastic fluctuations of transitions between hidden and observable states, as well as between observable states, are captured by diffusion terms driven by the $k_2$-dimensional Wiener process $W^{2}$, which appears in both SDEs in \eqref{fp11}. This is why we assume $k_1\in\{1,\ldots,K-1\}$ above, as this ensures that $k_2=K-k_1\ge 1$.

\end{remark}

\myparagraph{Time Discretization}
Given that, in practice, measurements are usually not recorded continuously, but only at specific points in time, we perform time discretization of the continuous time models \eqref{DA1} and \eqref{fp11} to make them more suitable for practical applications. 
For this purpose,   the time interval  $[0,T]$ is  divided into $N_t\in \mathbb{N}$ uniformly spaced subintervals of length $\Delta t = {T}/{N_t}$, resulting in the time grid points  $t_n = n \Delta t$ for $n =0,\ldots,N_t$. 
Applying the Euler-Maruyama scheme to \eqref{DA1}  yields the following discrete-time  stochastic recursion for $ n = 0, \ldots, N_{t}-1$
\[
\XD_{n+1} = \XD_n + \driftX(n, \XD_n)\Delta t + \diffX(n, \XD_n)\sqrt{\Delta t} \,\Noise_{n+1}, \quad \XD_0=x_0.
\]
Here,  the sequence \( (\Noise_n) _{n = 1, \ldots, N_{t}}\), \(\) consists of   $K$-dimensional independent and identically distributed  (i.i.d.) standard normal random vectors, \( \Noise_{n} \sim \mathcal{N}(0_{K}, \one_K) \).  
Similarly, for  the coupled  SDEs \eqref{fp11} for the hidden and observable state components, we obtain the corresponding recursive expressions that describe the evolution of both the signal and observation processes given by 
\begin{align}
	\begin{array}{rlrl}	
	Y_{n+1} &= Y_n + \overline{f}(n, Y_n, Z_n)\Delta t + \overline{\sigma}(n, Y_n, Z_n)\sqrt{\Delta t}\,\Noise_{n+1}^1
	 + \overline{g}(n, Y_n, Z_n)\sqrt{\Delta t}\,\Noise_{n+1}^2, & Y_0 &= y_0,   
	\\[0.5ex]
	Z_{n+1} &= Z_n + \overline{h}(n, Y_n, Z_n)\Delta t + \overline{\ell}(n, Y_n, Z_n)\sqrt{\Delta t}\,\Noise_{n+1}^2, &  Z_0 & = z_0.
	\end{array} 
	\label{md4}	
\end{align}
 Here, $(\Noise^1_n)$ and  $(\Noise^2_n)$ are independent sequences of i.i.d.~$\mathcal{N}(0_{k_1},\one_{k_1})$ and $\mathcal{N}(0_{k_2},\one_{k_2})$ random vectors, respectively, so that $\Noise^1_n$ contains those $k_1\in\{1,\ldots,K-1\}$ entries of $\Noise_n$  that only appear in the recursion for $Y$, while $\Noise^2_n$ collects the $k_2=K-k_1$ entries of $\Noise_n$ that appear in the recursion for both $Y$ and $Z$.
To simplify the notation, we reformulate the scheme as follows:
\begin{align}
	\label{Discr_Model}
	\begin{array}{rlrl}			
		Y_{n+1} &=  f(n, Y_n, Z_n) + \sigma(n, Y_n, Z_n)\,\Noise^1_{n+1} + g(n, Y_n, Z_n)\,\Noise^2_{n+1}, & Y_0 &= y_0, \\[1ex]		
		Z_{n+1} &=   h(n, Y_n, Z_n) + \ell(n, Y_n, Z_n)\,\Noise^2_{n+1}, &  Z_0 &= z_0. 
\end{array} 
\end{align}
 The coefficients in the above recursions are defined as follows
\begin{align}
	\label{Discr_Model_coeff}
	\begin{split}
		\begin{array}{rcl@{\hspace*{2em}}rclrcl}
		 f(n, y, z) & =& y + \overline{f}(n, y, z)\Delta t, & \sigma(n, y, z)&=& \overline{\sigma}(n, y, z)\sqrt{\Delta t}, & g(n, y, z)&=& \overline{g}(n, y, z)\sqrt{\Delta t}, \\[0.5ex]
		h(n, y, z) & = &z + \overline{h}(n, y, z)\Delta t, & \ell(n, y, z)&=& \overline{\ell}(n, y, z)\sqrt{\Delta t}.
		\end{array}
	\end{split}
\end{align}
In the examples below, and in many related applications, the drift coefficient of the observation equation is linear with respect to the signal and can be decomposed as \( h(n,y,z) = h_0(n,z) + h_1(n,z)y \). This structure is useful for applications in the field of Kalman filtering.
\section{Model Examples}\label{Model_example}
In this section, we present \ifthenelse{\measles=1}{three}{two} examples of partially observable stochastic epidemic models. The first is the \( SI^{\pm}RS \) model, which extends the well-known $SIR$ and $SIRS$ models. It is considered in Subsect.~\ref{sec:SIpmRS}, and it divides the compartment of infected individuals into two groups: one for hidden infections and one for confirmed infections. \ifthenelse{\measles=1}{In the second example, studied in Subsect.~\ref{Measles_model} we consider the  \( SEI^{\pm}R^{\pm}V \) model for a disease in which both  recovery from infection and vaccination result in lifelong immunity. A typical example is the measles disease. It extends the \( SI^{\pm}RS \) model  to include  compartments for exposed and vaccinated individuals, and in addition to the compartment of infected persons, it also divides the compartment of  recovered  individuals into two groups: one for identified  cases and one for unidentified cases. Based on this model  our works \cite{mbouandi2025stochastic} and  \cite{MbouandiNjiasse2025} study stochastic optimal control problems arising in the cost-optimal management of epidemics.}{}

\ifthenelse{\measles=1}{
Our third example is a \covid model, see Subsect.~\ref{sec_base_covid}.
Here, we would like to take into account in detail the presence of asymptomatic carriers and the significant proportion of hidden cases, as well as the waning immunity following previous infection or vaccination. 
As in the \( SEI^{\pm}R^{\pm}V \) model, we distinguish between confirmed and hidden infected and recovered individuals and take vaccinations into account. 
}
{
The second example is a \covid model, see Subsect.~\ref{sec_base_covid}. It extends the \( SI^{\pm}RS \) model  to include  a compartment for vaccinated individuals, and in addition to the compartment of infected persons it also divides the compartment of  recovered  individuals into two groups: one for identified  cases and one for unidentified cases.
} 
Since recovery from \covid infection or vaccination does not guarantee lifelong immunity, we take into account a loss of immunity. In addition, we extend the model to include compartments that record infected individuals who require treatment in a hospital or intensive care unit, as well as individuals who have died from the disease. 

In Subsects.~\ref{sec:cascade_model} and \ref{sec:extended_covid} we investigate an extended \covid model and introduce additional so-called cascade states in order to use the information from observable inflows to compartments that are formally hidden because of  their unobservable outflow.		
These  \covid models are used for the statistical estimation of dark figures using extended Kalman filter techniques in our work \cite{ouabo2025estimating}, \cite{achola2025stochastic}, and in \cite{OuaboKamkumo2025}.

\subsection{$SI^{\pm}RS$  Model}\label{sec:SIpmRS}
	The \( SI^{\pm}RS \) model extends the well-known $SIR$ and $SIRS$ models to account for hidden infections. 
	The latter two models divide the population into three groups: susceptible individuals $S$, infected individuals $I$, and recovered individuals $R$. Here, the compartment $S$ contains individuals who are susceptible to the disease, i.e., they can become infected if they come into contact with an  individual from compartment $I$ who is currently infected with the disease and can transmit it to susceptible individuals. Finally, $R$ comprises individuals who have recovered from the disease and developed immunity. They are no longer susceptible to infection with the disease. The $SIR$ model assumes lifelong immunity, so $R$ can be considered as an absorbing state. In contrast, in the $SIRS$ model, recovered individuals can lose their immunity. They then return to the $S$ compartment and become susceptible again.
	
	\myparagraph{Model structure}
	The \( SI^{\pm}RS \) model depicted in  Figure  \ref{SImpRS_Model1} extends the $SIRS$ model and  divides the compartment  $I$ of infected individuals into the following two groups. 
	
	\myparagraph{$\mathbf I^-$, Hidden infected}
	This compartment comprises  infected individuals whose infection status is unknown to the health system. This may occur for several reasons. Some individuals may be asymptomatic or present only mild symptoms and therefore do not seek medical attention. Others may have limited access to testing facilities or may refuse to undergo diagnostic testing due to personal beliefs, fear of social stigma, or mistrust of health authorities. In addition, individuals may perform self-tests at home and obtain a positive result but choose not to report it to health authorities or medical institutions. In such situations, these individuals remain epidemiologically infectious and continue their normal social interactions, thereby contributing to the transmission of the disease within the population.
		
	\myparagraph{$\mathbf I^+$, Confirmed infected}	It represents individuals whose infection has been confirmed through official diagnostic procedures such as laboratory testing, clinical diagnosis, or reported positive self-tests. Once identified, these individuals are typically advised or required to isolate, receive medical supervision, or follow quarantine protocols. Consequently, although they remain biologically infectious, their effective contact with susceptible individuals is significantly reduced due to isolation measures, treatment protocols, or public health interventions. 
		For simplicity, this model assumes perfect quarantine, i.e., individuals in $I^+$ do not contribute to the spread of the disease, i.e., they do not infect susceptible individuals.

	\smallskip
	The dynamics of this model can be described by a CTMC with the  state process $X=(X_{1},X_{2},X_{3},X_{4})^\top  = (I^-,R,S,I^+)^\top$ of dimension $d=4$.
	Please note that we use the notation $S, I^{\pm}, R$ for both the name of the compartment and the current number of individuals in that compartment. Also, recall that we are assuming  a constant total population size $N$, thus it holds  $X_{1}+X_{2}+X_{3} +X_{4}=N$. Only one of these four states can be observed, namely $X_4=I^+$, since both the inflow of positively tested persons from $I^-$ and the outflow of recovered persons to $R$ can be tracked. The other states are hidden.
	For $I^-$, only the outflow of positively tested individuals to $I^+$ is recorded, but not the inflow of infected individuals from $S$. For $R$, the inflow of recovered individuals from $I^+$ is recorded, but not that from $I^-$. Furthermore, the outflow of individuals to $S$ who have lost their immunity is also not observable. For $S$, both the inflow from $R$ and the outflow of infected individuals to $I^-$ cannot be tracked. 
	
	\begin{table}[h]
		\begin{center}
			\includegraphics[width=0.4\textwidth]{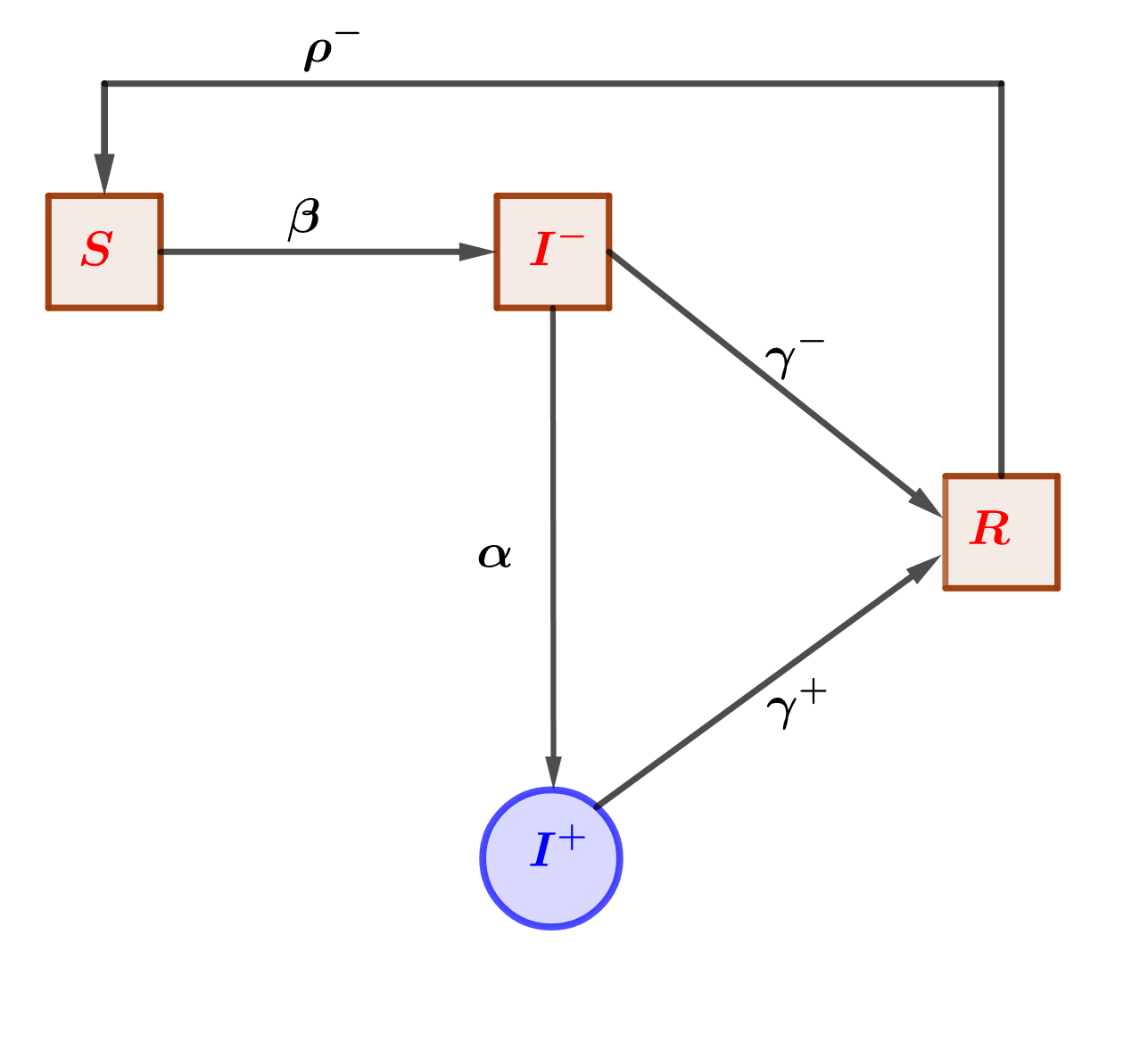}
			\captionof{figure}{Flow diagram of the partially observable $SI^{\pm}RS$ model, $d=4$ states, and $K=5$ transitions. The states include  $d_1=3$ hidden states ($S,I^-,R$, colored red), and $d_2=1$ observable state ($I^+$, colored blue). The transitions include $k_1=3$ internal transitions between hidden states and $k_2=2$ transitions between hidden and observable states.}
			\label{SImpRS_Model1}
		\end{center}
	\end{table}
	
	\begin{table}[h]
		
		\begin{center}
			
			\setlength{\tabcolsep}{5pt}
			 \begin{tabular}{llccc}
\midrule
$k$ & Transition && $\xi_k$ & $\lambda_k(t,Y,Z)$ \\
\midrule
1 & Infection of susceptible & $S\,\,\to I^-$ & $(\phantom{-}1,\phantom{-}0,-1,\phantom{-}0)^\top$ & $\beta Y^1Y^3/N$ \\
2 & Recovery of hidden infected & $I^-\!\to R$ & $(-1,\phantom{-}1,\phantom{-}0,\phantom{-}0)^\top$ & $\gamma^-Y^1$ \\
3 & Loss of immunity & $R\,\to S$ & $(\phantom{-}0,-1,\phantom{-}1,\phantom{-}0)^\top$ & $\rho Y^2$ \\\hline
4 & Positive test of hidden infected & $I^-\,\to I^+$ & $(-1,\phantom{-}0,\phantom{-}0,\phantom{-}1)^\top$ & $\alpha Y^1$ \\
5 & Recovery of confirmed infected & $I^+\!\to R$ & $(\phantom{-}0,\phantom{-}1,\phantom{-}0,-1)^\top$ & $\gamma^+Z^1$ \\
\midrule
\end{tabular}
			
		\end{center}
		\caption{Transition vectors $\xi_k$ and transition intensities $\lambda_k(t,x)$ for the  $K=5$ transitions of the partially observable $SI^{\pm}RS$ model; see Figure \ref{SImpRS_Model1}. The state process $X$ of dimension $d=4$ is decomposed into the hidden state  $Y=(I^-,R,S)^\top$ and the observable state $Z=I^+$.  Transitions 1 through 3 are internal transitions between hidden states, the remaining two transitions are between hidden and observable states. The parameters  $\beta,\gamma^-,\gamma^+,\alpha,\rho$ are supposed to be positive and may be time-dependent. }		
		\label{Table_Info_SIpmRS_model}
	\end{table}

	In summary, for the diffusion approximation we have $d_1=3$ hidden components that form $Y=(I^{-},R,S)^{\top}$ and $d_2=1$ observable component that forms $Z=I^{+}$. As it can be seen from the flowchart in Figure  \ref{SImpRS_Model1}, there are $K=5$ transitions for which we give the corresponding transition vectors $\xi_k$ and transition intensities $\lambda_k(t,x)$ in Table \ref{Table_Info_SIpmRS_model}. The choice of intensities $\lambda_k$ corresponds to that in the standard SIRS model.
	The transition intensity $\lambda_1$ of the first transition, i.e., the infection of a susceptible person, is proportional to the number of susceptible persons  and to the proportion $I^-/N$ of undetected infected persons, with a proportionality factor $\beta>0$, which is referred to as the infection or transmission rate. The latter can be time-dependent and can also be controlled, e.g., by social distancing measures. This is taken into account by the time dependence of $\lambda_1(t,x)$. The other four intensities are proportional to the number of people in the compartment that is being left, and there are proportionality factors $\gamma^-,\gamma^+,\alpha,\rho>0$, which represent the recovery rate for undetected and detected infected individuals, the testing rate, and the rate at which immunity is lost, respectively. In particular, the testing rate $\alpha$ can be time-dependent and controlled here, i.e., it is subject to time-varying countermeasures implemented by health authorities.
	
	\myparagraph{State dynamics}
	With the above settings the dynamics of the CTMC state process given in \eqref{Eq_CTMC},   of its  diffusion approximation $\XD$ given by the SDE  \eqref{DA1} as well as of the hidden and observable components $Y,Z$ given in the system of coupled SDEs \eqref{fp11} are fully specified.
	In particular the diffusion approximation \eqref{DA1} reads 
	\[d\XD(t)=f_{X}(t,\XD(t))dt+\sigma_{X}(t,\XD(t))\,dW(t),\]
	where 
			 $$\driftX(t,y,z)= \left(\begin{array}{c}		
			\beta{y^1y^3}/{N} - (\alpha + \gamma^{-})y^1\\
			\gamma^{-}y^1 +\gamma^{+} z^1 - \rho y^2\\
			-\beta {y^1y^3}/{N} +\rho y^2 \\ [0.4ex]
			\alpha y^1 -\gamma^{+}z^1\\
		\end{array} \right),~~~~ \diffX(t,y,z)= \left(\begin{array}{c@{\hspace*{+0.05cm}}c@{\hspace*{-0.01cm}}c@{\hspace*{-0.01cm}}c@{\hspace*{-0.01cm}}c}			
			\sqrt{\beta{y^1y^3}/{N}} & -\sqrt{\gamma^{-}y^1} & 0 & -\sqrt{\alpha y^1} & 0\\
			0 & \sqrt{\gamma^{-}y^1} & -\sqrt{\rho y^2} & 0 & \sqrt{\gamma^{+}z^1}\\
			-\sqrt{\beta{y^1y^3}/{N}} & 0 & \sqrt{\rho y^2} & 0 & 0\\
			0 & 0 & 0 & \sqrt{\alpha y^1} & -\sqrt{\gamma^{+}z^1}\\
		\end{array} \right)
		$$ 
		and $W=(W_1,\ldots,W_5)^{\top}$. The coefficients $\driftX$ and $\diffX$ are specified in terms of the components $y,z$ of the state vector  $x=(y^1,y^2,y^3,z)^\top$ to keep the notation consistent with the coupled system of SDEs \eqref{fp11} for $Y$ and $Z$, which is given by
	\begin{align}
		\begin{array}{rlrl}
			dY(t) &= \overline{f}(t, Y(t), Z(t))dt + \overline{\sigma}(t, Y(t), Z(t))\,dW^{1}(t) + \overline{g}(t, Y(t), Z(t))\,dW^{2}(t),\\[0.5ex]
			dZ(t) &= \overline{h}(t, Y(t), Z(t))dt + \overline{\ell}(t, Y(t), Z(t))\,dW^{2}(t).
		\end{array}
	\end{align} 
	The coefficients $ \overline{f}, ~ \overline{\sigma},~ \overline{g} ,~ \overline{\ell} $  are  non-linear functions of the  hidden  state $Y$, while the drift coefficient $\overline h$ of the SDE for the observable state $Z$ is linear and of the form $\overline{h}(t,y,z) = \overline{h}_0(t,z) + \overline{h}_1(t,z) y$.  They are given as follows
	 \begin{align}
		\overline{f}(t,y,z)&=\left(\begin{array}{c}		
			\beta{y^1y^3}/{N}-(\gamma^{-} + \alpha)y^1\\
			\gamma^{-}y^1 -\rho y^2 +\gamma^{+}z^1\\
			-\beta{y^1y^3}/{N} + \rho y^2\\  [0.4ex]
		\end{array} \right),\hspace*{2em}
		\overline{h}_{0}(t,z)=-\gamma^{+}z^1, \hspace*{2em}
		\overline{h}_{1}(t,z)=\left(\begin{array}{ccc}
			\alpha & 0 & 0
		\end{array} \right),
		\\
		\overline{\sigma}(t,y,z)& =\left(\begin{array}{ccc}
			\sqrt{\beta{y^1y^3}/{N}}  & -\sqrt{\gamma^{-} y^1} & 0 \\
			0 & \sqrt{\gamma^{-} y^1 } & -\sqrt{\rho y^2}\\
			-\sqrt{\beta{y^1y^3}/{N}} & 0 & \sqrt{\rho y^2} 
		\end{array} \right),\quad \overline{g}(t,y,z)=\left(\begin{array}{cc}
			-\sqrt{\alpha y^1}  & 0 \\
			0  & \sqrt{\gamma^{+} z^1}\\
			0  & 0 \\
		\end{array} \right), \\ 
		\overline{\ell}(t,y,z)&=\left(\begin{array}{cc}
			\sqrt{\alpha y^1} &-\sqrt{\gamma^{+}z^1}
		\end{array} \right), \quad 
		 W^{1} = (W_1, W_2, W_3)^{\top}, ~~W^{2}=(W_4, W_5)^{\top}.
	\end{align}
	  The entries of $W^1$ drive the transitions between hidden states, see transitions 1 through 3 in Table~\ref{Table_Info_SIpmRS_model}, while the entries of $W^2$ drive transitions 4 and 5 from and to the observable state $I^+$. 	
	 Note that the coefficients  $\driftX,\diffX$ given above can be also obtained from $\overline{f},\overline{h},\overline{\sigma},\overline{g},\overline{\ell}$ using block matrix decomposition \eqref{block_SDE}. Further, the coefficients of the recursions \eqref{Discr_Model} describing the discrete-time state dynamics follow   from  \eqref{Discr_Model_coeff}.	

\subsection{$SEI^{\pm}R^{\pm}V$  Model}\label{Measles_model}
\myparagraph{Model structure} This model extends the $SI^{\pm}RS$ model presented in the previous subsection to include compartments for exposed (E) and vaccinated (V) individuals. Furthermore, the group $R$ of recovered individuals is divided into two groups. The first group is group $R^-$, which consists of persons who are recovering from an undetected infection and are transitioning from $I^-$ to $R^-$.	The second group is denoted by $R^+$ and comprises  those individuals who recover from a detected infection, i.e.,  they  move from $I^+$ to $R^+$. In contrast to the $SI^{\pm}RS$ model, it is assumed that recovery from the disease and vaccination provide lifelong immunity. 
	
A typical example of such a disease is measles. Many earlier studies on the spread of measles within a population were based on modified $SEIRV$ models, which divided a population into susceptible, exposed, infected, recovered, and vaccinated individuals, see \citep{dayan2023reliable,ibrahim2023stability,minta2023progress}. Such models usually assume that all compartments are fully observable, which is not realistic, as we have some compartments such as the exposed compartment $E$, whose individuals are not directly observable.  The  $SEI^{\pm}R^{\pm}V$ model proposed in this subsection can be considered an extension of the $SEIRV$ models to account for the more realistic case of unobservable compartments.
\ifthenelse{\preprintOn=0}{}{{\prepcolor
In Remark \ref{rem:measles}, we provide further explanations on specific characteristics of this disease.
\begin{remark}\label{rem:measles} \quad \textbf{Measles disease} \\
Measles, also referred to as rubeola, is an extremely contagious viral illness caused by the measles morbillivirus \citep{CDC2026MeaslesClinicalOverview}. Once a person has been infected, the course of the disease  from infection to recovery  can be divided into three distinct phases. The first stage is the incubation period, which lasts $7$--$14$ days. During this period, there are no symptoms and the individual is not yet infectious. The second stage is characterized by the initial symptoms, followed by a rash. There is no specific treatment for measles. During this time, infected individuals are contagious and can transmit the virus to susceptible people. Therefore, it is crucial to isolate infected individuals to prevent the virus from spreading to others.
	 Most people fully recover within two to three weeks of the rash appearing and progress to the third stage. Recovery  results in lifelong immunity.
	 
Vaccination, in particular the highly effective MMR (measles, mumps and rubella) vaccine, is the most effective way of preventing measles, as it provides lifelong immunity, and the herd immunity resulting from widespread vaccination is crucial to preventing epidemics.

 The most common methods for confirming a measles infection include serology (blood tests) and PCR (Polymerase Chain Reaction). Serology tests involve the IgM antibody test, which detects measles-specific IgM antibodies, indicating a current or recent infection. The  PCR methods, specifically RT-PCR (Reverse Transcription PCR), detect the genetic material of the measles virus in clinical samples such as throat swabs, nasopharyngeal swabs, urine, or blood, offering high sensitivity and the ability to confirm infection even before antibodies are detectable. 
\end{remark}
}}

\myparagraph{Compartments and Transitions}
We are now going to give a detailed description of the compartments and the transitions of the  $SEI^{\pm}R^{\pm}V$ model for which the flowchart is presented in Figure \ref{Measles_Model}. Furthermore, the transitions with the associated transition vectors and intensities are given in Table \ref{Table_Info_Measle_Model}.

\begin{table}[htbp]
	\centering
	\includegraphics[width=0.7\textwidth]{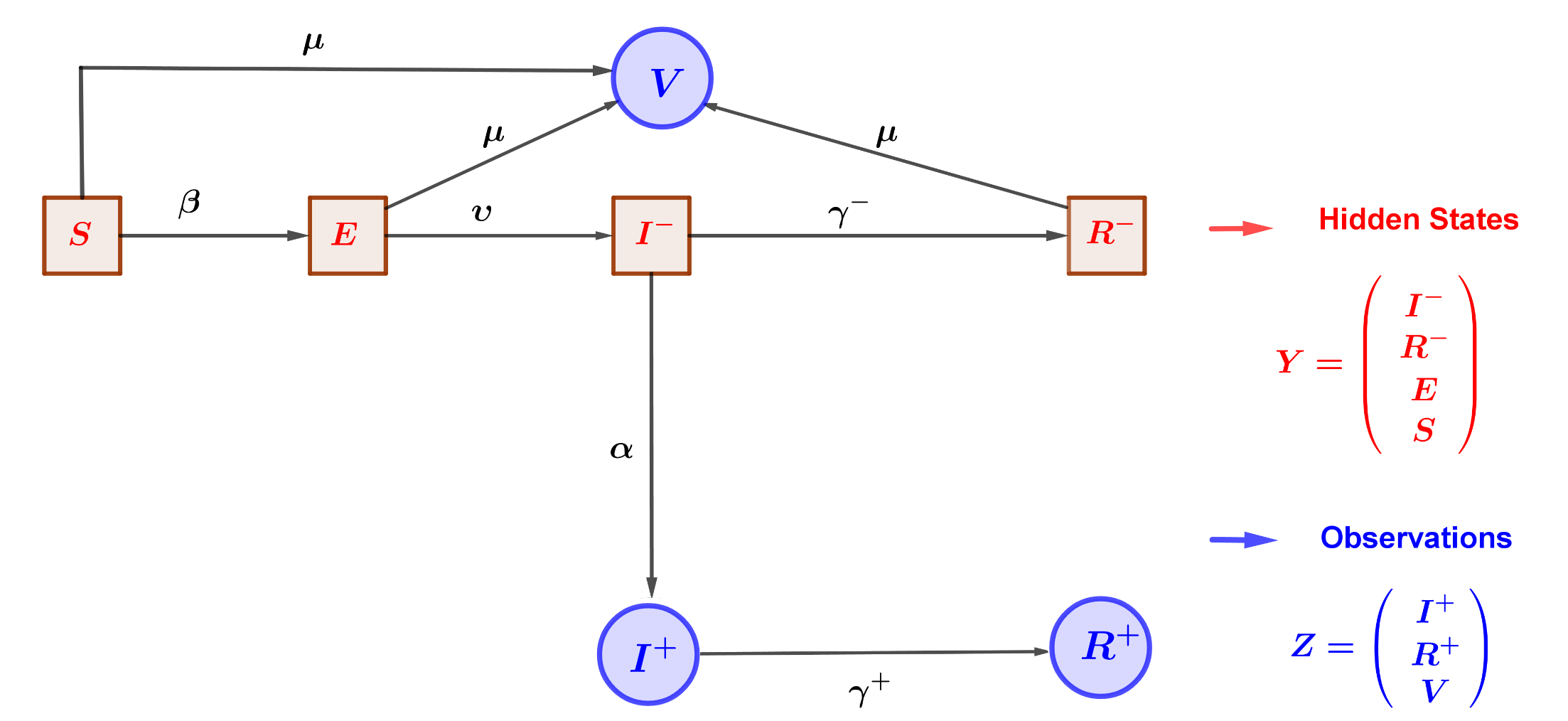}
	\captionof{figure}{Flow diagram of the partially observable $SEI^{\pm}R^{\pm}V$ model, 
		$d=7$ states, and $K=8$ transitions. The states include  $d_1=4$ hidden states ( $I^-,R^-,E,S$, colored red), and $d_2=3$ observable states ($I^+,R^+,V$, colored blue), forming $Y=(I^-,R^-,E,S)^\top$ and $Z=(I^+,R^+,V)^\top$. The transitions comprise $k_1=3$ transitions between hidden states and $k_2=5$ transitions involving at least one observable state.}	
	\label{Measles_Model}
\end{table} 
\begin{table}[!h]
	\begin{center}
		\setlength{\tabcolsep}{5pt}
		 \begin{tabular}{llccc}
\midrule
$k$ & Transition && $\xi_k$ & $\lambda_k(t,Y,Z)$ \\
\midrule
1 & Progression to infectious & $E\to I^-$ & $(\phantom{-}1,\phantom{-}0,-1,\phantom{-}0,\phantom{-}0,\phantom{-}0,\phantom{-}0)^\top$ & $\upsilon Y^3$ \\
2 & Recovery of hidden infected & $I^-\!\to R^-$ & $(-1,\phantom{-}1,\phantom{-}0,\phantom{-}0,\phantom{-}0,\phantom{-}0,\phantom{-}0)^\top$ & $\gamma^-Y^1$ \\
3 & Infection & $S\to E$ & $(\phantom{-}0,\phantom{-}0,\phantom{-}1,-1,\phantom{-}0,\phantom{-}0,\phantom{-}0)^\top$ & $\beta Y^1Y^4/N$ \\\hline
4 & Positive test of hidden infected & $I^-\!\to I^+$ & $(-1,\phantom{-}0,\phantom{-}0,\phantom{-}0,\phantom{-}1,\phantom{-}0,\phantom{-}0)^\top$ & $\alpha Y^1$ \\
5 & Recovery of confirmed infected & $I^+\to R^+$ & $(\phantom{-}0,\phantom{-}0,\phantom{-}0,\phantom{-}0,-1,\phantom{-}1,\phantom{-}0)^\top$ & $\gamma^+Z^1$ \\
6 & Vaccination of susceptible & $S\to V$ & $(\phantom{-}0,\phantom{-}0,\phantom{-}0,-1,\phantom{-}0,\phantom{-}0,\phantom{-}1)^\top$ & $\mu Y^4$ \\
7 & Vaccination of exposed & $E\to V$ & $(\phantom{-}0,\phantom{-}0,-1,\phantom{-}0,\phantom{-}0,\phantom{-}0,\phantom{-}1)^\top$ & $\mu Y^3$ \\
8 & Vaccination of hidden recovered & $R^-\!\to V$ & $(\phantom{-}0,-1,\phantom{-}0,\phantom{-}0,\phantom{-}0,\phantom{-}0,\phantom{-}1)^\top$ & $\mu Y^2$ \\
\midrule
\end{tabular}
		
	\end{center}
	\caption{Transition vectors $\xi_k$ and transition intensities $\lambda_k(t,x)$ for the  $K=8$ transitions of the partially observable $SEI^{\pm}R^{\pm}V$ model; see Figure \ref{Measles_Model}. The state process $X$ of dimension $d=7$ is decomposed into the hidden state  $Y=(I^-,R^-,E,S)^\top$ and the observable state $Z=(I^+,R^+,V)^\top$.  Transitions 1 through 3 occur between hidden states, the remaining transitions involve at least one observable state. The parameters $\beta,\upsilon,\gamma^-,\gamma^+,\alpha,\mu$ are positive and may be time-dependent.
	}
	\label{Table_Info_Measle_Model}
\end{table}

\myparagraph{$\mathbf S $, Susceptible}
These individuals can become infected by infectious individuals who are not in quarantine, i.e., by individuals in $ I^{-} $. They then transition to $ E $, and when they receive the vaccine, they transition to the compartment $ V $.

\myparagraph{$\mathbf E $,  Exposed}
Individuals in E are infected but do not yet show any symptoms, they are not infectious and will not test positive. After a certain incubation period, they develop symptoms, become infectious, and move to compartment $I^-$. Since individuals in this compartment are unaware of their infection, they may receive a vaccination because they are considered susceptible, and subsequently move to compartment $V$.

\myparagraph{$\mathbf I^{-} $, Hidden infected}	  
These individuals are infected with measles and are typically symptomatic (fever, rash), but remain unreported because they do not seek care or are not captured by surveillance systems. They may later be detected and move to $I^{+}$, or recover naturally and move to $R^{-}$, thereby acquiring immunity. Individuals in $I^{-}$ are not vaccinated  as they are  symptomatic, and they are unlikely to attend vaccination. Further, vaccination programs target individuals who are presumed susceptible. In this model, these include $S$ and the unrecognized states $E,R^-$, but exclude symptomatic $I^-$ and known infected or recovered individuals in $I^+,R^+$. 
For more details, we refer to the description of $I^-$ in Subsect.~\ref{sec:SIpmRS}.

\myparagraph{$\mathbf R^{-} $, Hidden recovered}	  
These individuals have recovered from the illness, were unaware they had been infected, or refused to be tested, i.e., they have moved from $I^-$ to $R^-$. Although their infection-induced immunity does not fade,  their actual state is unknown and  they may therefore be vaccinated during a vaccination campaign and then move to $V$.

\myparagraph{$\mathbf I^{+}$, Confirmed infected }
This compartment  comprises infected persons who have tested positive. They are placed in quarantine until they fully recover and then move into the $R^+$
compartment as confirmed recovered. Although they are infectious, they cannot transmit the virus to others, as they are assumed to be completely isolated.
For more details, we refer to the description of $I^+$ in Subsect.~\ref{sec:SIpmRS}. 

\myparagraph{$\mathbf R^{+}$, Confirmed recovered}
Here, individuals have recovered from  confirmed  infections, i.e.,  they have transitioned from $I^+$ to $R^+$. Since they have lifelong immunity, this is an absorbing state from which no transitions to other states occur.

\myparagraph{$\mathbf V$, Vaccinated}
This state includes vaccinated individuals. Since the vaccine provides lifelong immunity, this is also an absorbing state from which no transitions occur to other states. There are vaccination transitions from $S,E,R^-$ to this state. 

\smallskip
Of the $d=7$ states mentioned above, $d_2=3$ are observable. This is $I^+$, since both the inflow of people who have tested positive from $I^-$ and the outflow due to recovery to $R^+$ can be tracked. In addition, $R^+$ and $V$ are observable in this model because they are reached through recorded recoveries and vaccinations and do not have an unobservable outflow. The remaining $d_1=4$ states  $I^-,R^-,E,S$ are hidden, since they are associated with at least one unobservable transition. As it can be seen in Figure \ref{Measles_Model} and Table \ref{Table_Info_Measle_Model}, the transition intensities again follow standard epidemic modeling principles: infection is proportional to the number of susceptible individuals and to the proportion of undetected infected individuals, while exposure progression, testing, recovery, and vaccination transitions are linear in the compartment from which individuals leave. The positive parameters in Table \ref{Table_Info_Measle_Model} therefore determine the CTMC transition rates and the corresponding diffusion coefficients.

\myparagraph{State dynamics}
With the settings in Table \ref{Table_Info_Measle_Model} the dynamics of the state process $X=(I^{-},R^{-},E,S,I^{+},R^{+},V)^{\top}$ of the  $SEI^{\pm}R^{\pm}V$ depicted in Figure \ref{Measles_Model} are fully specified. They can be described  using the  CTMC approach as in \eqref{Eq_CTMC},  or the   diffusion approximation $\XD$ given by the SDE  \eqref{DA1}, or its decomposition into  the hidden component  $Y=(I^{-},R^{-},E,S)^{\top}$ and the observable component $Z=(I^{+},R^{+},V)^{\top}$ satisfying the  system of coupled SDEs \eqref{fp11}. 

 These coefficients  of the latter system of SDEs are given as follows
 \begin{align*}
\overline f&=\begin{pmatrix}\upsilon y^3-(\alpha+\gamma^-)y^1\\\gamma^-y^1-\mu y^2\\\beta y^1y^4/N-(\mu+\upsilon)y^3\\-\beta y^1y^4/N-\mu y^4\end{pmatrix},
~~\overline h_0=\begin{pmatrix}-\gamma^+z^1\\\gamma^+z^1\\0\end{pmatrix},
~~\arraycolsep=2.5pt\overline h_1=\begin{pmatrix}\alpha&0&0&0\\0&0&0&0\\0&\mu&\mu&\mu\end{pmatrix},~~
\arraycolsep=1.5pt\overline\sigma=\begin{pmatrix}\sqrt{\upsilon y^3}&-\sqrt{\gamma^-y^1}&0\\0&\sqrt{\gamma^-y^1}&0\\-\sqrt{\upsilon y^3}&0&\sqrt{\beta y^1y^4/N}\\0&0&-\sqrt{\beta y^1y^4/N}\end{pmatrix},\\
\overline g&=\begin{pmatrix}-\sqrt{\alpha y^1}&0&0&0&0\\0&0&0&0&-\sqrt{\mu y^2}\\0&0&0&-\sqrt{\mu y^3}&0\\0&0&-\sqrt{\mu y^4}&0&0\end{pmatrix},~~
\overline\ell=\begin{pmatrix}\sqrt{\alpha y^1}&-\sqrt{\gamma^+z^1}&0&0&0\\0&\sqrt{\gamma^+z^1}&0&0&0\\0&0&\sqrt{\mu y^4}&\sqrt{\mu y^3}&\sqrt{\mu y^2}\end{pmatrix}.
\end{align*}
Here, as for the  $SI^{\pm}RS$  model,    the drift coefficient $\overline h$ of the SDE for $Z$ is linear and of the form $\overline{h}(t,y,z) = \overline{h}_0(t,z) + \overline{h}_1(t,z) y$. The other coefficients $ \overline{f}, ~ \overline{\sigma},~ \overline{g} ,~ \overline{\ell} $  are  non-linear in the  hidden  state $Y$. 
 One of the two multi-dimensional Brownian motions driving the SDE for $Y$ is  $W^1$. Its  $k_1=3$  entries  are related to the transitions 1 through 3   in Table \ref{Table_Info_Measle_Model}, i.e., the transitions $E\to I^-$, $I^-\to R^-$, $S\to E$ between hidden states. The Brownian motion $W^2$ of dimension $k_2=5$ also drives the SDE for $Y$, but also the SDE for $Z$. Its entries are related to the remaining 5 transitions, which also involve observable states.		
		  
The coefficients  $\driftX,\diffX$ of the SDE \eqref{DA1} for the dynamics of $\XD$ are obtained from $\overline{f},\overline{h},\overline{\sigma},\overline{g},\overline{\ell}$ using the block matrix decomposition \eqref{block_SDE}.  Note that the coefficients of the recursions \eqref{Discr_Model} describing the discrete-time state dynamics follow   from  \eqref{Discr_Model_coeff}.

\subsection{Base \covid Model}\label{sec_base_covid}
 We will now propose partially observable stochastic epidemic models that capture specific characteristics of the \covid pandemic. The first characteristic is the high proportion of asymptomatic cases, which are contagious and contribute to the spread of the disease. In addition, some people develop symptoms and receive a preliminary positive test result, for example, from rapid tests they purchased at local stores, but for one reason or another refuse to get tested at a testing center to definitively confirm whether or not they are infected. Combined with the often low testing rates in many countries, this poses a major problem for health authorities in their fight against the spread of the disease and hinders their efforts, as it is very difficult to obtain reliable information about the actual number of infections.

A second key feature is the advent of a vaccine that provides temporary immunity. Third, \covid is characterized by the fact that immunity wanes following a previous infection or vaccination. This leads to a further increase in cases with unknown status. Finally, a certain percentage of infections can take a severe course, requiring treatment in a hospital or, in some cases, in an intensive care unit (ICU), and some people even die from the disease.
 In this context, we would like to recall  the concept of ``flattening the curve'' during the \covid pandemic, as discussed in \citet{qualls2017community} and \citet{ferguson2020impact}. This means that one of the goals of public health interventions  should be to prevent the healthcare system from becoming overburdened to the point where it exceeds the capacity of hospitals and, in particular, ICUs.
  
   In this subsection, we begin by describing our first \covid model, which is referred to as the base model. Later in Subsect.~\ref{sec:extended_covid}, a second, refined model is presented, which is referred to as the extended \covid model.

\myparagraph{Model structure}
To capture the characteristics mentioned above, we use the compartments $I^+$ and $I^-$ for confirmed and hidden infections, respectively, as we did in the
\ifthenelse{\measles=1}{
 $SI^{\pm}RS$ and $SEI^{\pm}R^{\pm}V$ models.}
 {
 $SI^{\pm}RS$ model.}
If a person from $I^{-}$ tests positive, they move to $I^{+}$. 
\ifthenelse{\preprintOn=0}{}
{{\prepcolor We refer to Remark \ref{rem:CovidTests} for details about \covid tests. }}
We also subdivide the group of recovered individuals into those who have recovered from a confirmed infection ($R^+$) and those who have recovered from an 
\ifthenelse{\measles=1}{
	 hidden infection ($R^-$), as in the $SEI^{\pm}R^{\pm}V$ model.}
{
	hidden infection ($R^-$).}
 In addition, we introduce a compartment $V$ consisting of vaccinated individuals. For simplicity, $V$ aggregates all individuals who have received at least one dose of the vaccine, regardless of the number of doses or the vaccine type administered. The model   therefore does  not distinguish between partially and fully vaccinated individuals, but instead focuses on the collective impact of vaccination on the dynamics of the disease.
 As none of the available \covid vaccines  (e.g., mRNA vaccines such as Pfizer-BioNTech and Moderna) offer lifelong immunity,  it is possible for vaccinated individuals to lose their immunity over time and become susceptible.

 To capture individuals’ interactions with the healthcare system, we introduce the compartments $H$ (hospital patients) and $C$ (patients requiring intensive care). Using these two compartments, we can distinguish the use of healthcare services based on the severity of the illness. Finally, the \( D \)-compartment (death) represents those who succumb to the disease. We assume that deaths due to  \covid occur exclusively among individuals treated in ICUs, i.e.,  in the \( C \)-compartment.   For the sake of simplicity, natural mortality, recruitment (birth), and migration effects are neglected, allowing us to work with a constant total population size.   

\begin{table}[hb]
	\centering
	\includegraphics[width=14cm,height=8cm]{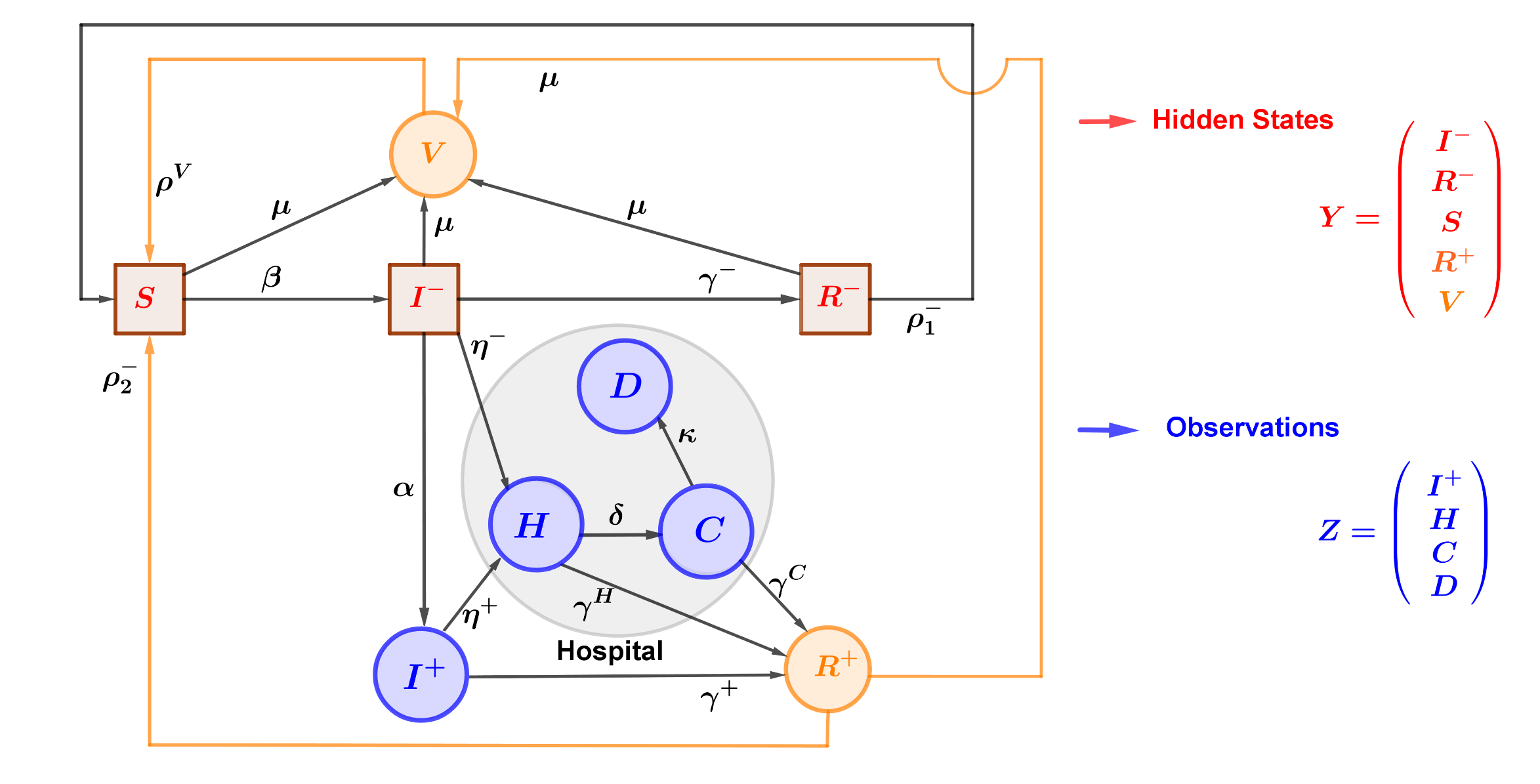}	
	\captionof{figure}{Flow diagram of the partially observable base \covid model, 
		$d=9$ states, and $K=17$ transitions. The states include  $d_1=5$ hidden states $I^-,R^-,S,R^+,V$, and $d_2=4$ observable states $I^+,H,C,D$, colored blue. Among the hidden states,  $I^-,R^-,S$ are fully hidden (colored red), while $R^+,V$ are partially hidden and highlighted in orange because they exhibit an observable inflow but an unobservable outflow.  Of the $K$ transitions, $k_1=9$ connect hidden states, whereas $k_2=8$ involve at least one observable state; see Table \ref{Table_Info_Model1}.
		}
	\label{fig:BaseCovid}
\end{table}	

\begin{table}[h]
	\begingroup\centering\small
	\begin{tabular}{c@{\hspace*{0.6em}}l@{\hspace*{0.8em}}r@{\hspace*{0.3em}}c@{\hspace*{0.3em}}lc@{\hspace*{0.2em}}c}\toprule
		$k$ & Transition &&&& Transition vector & Transition intensity \\
		&&&&& $\xi_k$ &$\lambda_k(t,X)=\lambda_k(t,Y,Z)$\\\midrule
		1 & Infection of susceptible &$S$ & $\to$ & $I^-$    & $(\phantom{-}1,\phantom{-}0,-1,\phantom{-}0,\phantom{-}0,\phantom{-}0,\phantom{-}0,\phantom{-}0,\phantom{-}0)^\top$ & $\beta S\dfrac{I^-}{N}=\beta Y^{3}\dfrac{Y^{1}}{N}$ \\
		2 & Recovery of hidden infected & $I^-$ & $\to$ & $R^-$  & $(-1,\phantom{-}1,\phantom{-}0,\phantom{-}0,\phantom{-}0,\phantom{-}0,\phantom{-}0,\phantom{-}0,\phantom{-}0)^\top$ & $\gamma^- I^- = \gamma^- Y^{1}$ \\
		3 & Immunity loss of hidden recovered & $R^-$ & $\to$ & $S$    & $(\phantom{-}0,-1,\phantom{-}1,\phantom{-}0,\phantom{-}0,\phantom{-}0,\phantom{-}0,\phantom{-}0,\phantom{-}0)^\top$ & $\rho^-_1 R^- = \rho^-_1 Y^{2}$ \\
		4 & Immunity loss  of confirmed recovered & $R^+$ & $\to$ & $S$    & $(\phantom{-}0,\phantom{-}0,\phantom{-}1,-1,\phantom{-}0,\phantom{-}0,\phantom{-}0,\phantom{-}0,\phantom{-}0)^\top$ & $\rho^-_2 R^+ = \rho^-_2 Y^{4}$ \\
		5 & Immunity loss  of vaccinated & $V$ & $\to$ & $S$      & $(\phantom{-}0,\phantom{-}0,\phantom{-}1,\phantom{-}0,-1,\phantom{-}0,\phantom{-}0,\phantom{-}0,\phantom{-}0)^\top$ & $\rho^V V = \rho^V Y^{5}$ \\
		6 & Vaccination of susceptible & $S$ & $\to$ & $V$      & $(\phantom{-}0,\phantom{-}0,-1,\phantom{-}0,\phantom{-}1,\phantom{-}0,\phantom{-}0,\phantom{-}0,\phantom{-}0)^\top$ & $\mu S = \mu Y^{3}$ \\
		7 & Vaccination of hidden infected &$I^-$ & $\to$ & $V$    & $(-1,\phantom{-}0,\phantom{-}0,\phantom{-}0,\phantom{-}1,\phantom{-}0,\phantom{-}0,\phantom{-}0,\phantom{-}0)^\top$ & $\mu I^- = \mu Y^{1}$ \\
		8 & Vaccination of hidden recovered &$R^-$ & $\to$ & $V$    & $(\phantom{-}0,-1,\phantom{-}0,\phantom{-}0,\phantom{-}1,\phantom{-}0,\phantom{-}0,\phantom{-}0,\phantom{-}0)^\top$ & $\mu R^- = \mu Y^{2}$ \\
		9 & Vaccination of recovered from $R^+$ & $R^+$ & $\to$ & $V$    & $(\phantom{-}0,\phantom{-}0,\phantom{-}0,-1,\phantom{-}1,\phantom{-}0,\phantom{-}0,\phantom{-}0,\phantom{-}0)^\top$ & $\mu R^+ = \mu Y^{4}$ \\
		\midrule
		10 & Positive test of hidden infected & $I^-$ & $\to$ & $I^+$ & $(-1,\phantom{-}0,\phantom{-}0,\phantom{-}0,\phantom{-}0,\phantom{-}1,\phantom{-}0,\phantom{-}0,\phantom{-}0)^\top$ & $\alpha I^- = \alpha Y^{1}$ \\
		11 & Hospitalization of hidden infected & $I^-$ & $\to$ & $H$   & $(-1,\phantom{-}0,\phantom{-}0,\phantom{-}0,\phantom{-}0,\phantom{-}0,\phantom{-}1,\phantom{-}0,\phantom{-}0)^\top$ & $\eta^- I^- = \eta^- Y^{1}$ \\
		12 & Recovery of confirmed infected & $I^+$ & $\to$ & $R^+$ & $(\phantom{-}0,\phantom{-}0,\phantom{-}0,\phantom{-}1,\phantom{-}0,-1,\phantom{-}0,\phantom{-}0,\phantom{-}0)^\top$ & $\gamma^+ I^+ = \gamma^+ Z^{1}$ \\
		13 & Recovery from hospitalization & $H$ & $\to$ & $R^+$   & $(\phantom{-}0,\phantom{-}0,\phantom{-}0,\phantom{-}1,\phantom{-}0,\phantom{-}0,-1,\phantom{-}0,\phantom{-}0)^\top$ & $\gamma^H H = \gamma^H Z^{2}$ \\
		14 & Recovery from ICU & $C$ & $\to$ & $R^+$   & $(\phantom{-}0,\phantom{-}0,\phantom{-}0,\phantom{-}1,\phantom{-}0,\phantom{-}0,\phantom{-}0,-1,\phantom{-}0)^\top$ & $\gamma^C C = \gamma^C Z^{3}$ \\
		15 & Hospitalization of confirmed infected & $I^+$ & $\to$ & $H$   & $(\phantom{-}0,\phantom{-}0,\phantom{-}0,\phantom{-}0,\phantom{-}0,-1,\phantom{-}1,\phantom{-}0,\phantom{-}0)^\top$ & $\eta^+ I^+ = \eta^+ Z^{1}$ \\
		16 & Death & $C$ & $\to$ & $D$     & $(\phantom{-}0,\phantom{-}0,\phantom{-}0,\phantom{-}0,\phantom{-}0,\phantom{-}0,\phantom{-}0,-1,\phantom{-}1)^\top$ & $\kappa C = \kappa Z^{3}$ \\
		17 & Transfer to ICU & $H$ & $\to$ & $C$     & $(\phantom{-}0,\phantom{-}0,\phantom{-}0,\phantom{-}0,\phantom{-}0,\phantom{-}0,-1,\phantom{-}1,\phantom{-}0)^\top$ & $\delta H = \delta Z^{2}$ \\
		\bottomrule\end{tabular}
	\caption{Transition vectors $\xi_k$ and transition intensities  $\lambda_k(t,X)$ for the $K=17$ transitions of the partially observable base \covid model, see the flowchart in Figure \ref{fig:BaseCovid}. The state process $X$ of dimension $d=9$  is  decomposed into the hidden state $Y=(I^-,R^-,S,R^+,V)^\top$, and the observable state $Z=(I^+,H,C,D)^\top$. Transitions 1-9 connect two hidden states. 
		Transitions 10–17 connect two states, at least one of which is an observable state. \newline The parameters $\beta,\gamma^\pm,\gamma^{H/C},\alpha,\mu,\rho_{1/2}^-,\rho^V,\eta^{\pm},\delta,\kappa$ are positive model parameters. In the numerical study, the infection, testing, and vaccination rates $\beta,\alpha,\mu$ are treated as time-varying quantities estimated from data, whereas the other parameters 
		are kept fixed. 
		}
	\label{Table_Info_Model1}
	\endgroup\end{table}


\ifthenelse{\preprintOn=0}{}{{\prepcolor
	\begin{remark}\label{rem:CovidTests} \quad \textbf{\covid tests}\\
		During the \covid pandemic, three main types of tests have been widely used for diagnosing the virus. The first two on the list are tests that can be used to confirm the presence or absence of the \covid virus. These are PCR tests (Polymerase Chain Reaction): PCR tests are considered the gold standard for \covid diagnosis. They detect the virus' genetic material and are highly accurate. However, they can take from a few hours to a few days to produce results, depending on laboratory capacity and test volume. In addition, we have antigenic tests. These tests detect specific proteins on the surface of the virus. They are faster and less expensive than PCR tests, but may be slightly less sensitive. Antigenic tests are often used for rapid screening in settings where immediate results are needed, such as workplaces, schools and community testing sites. Results can be available in less than 15 to 30 minutes. The third is an antibody test. They detect the antibodies produced by the body in response to an infection with \covidk. These tests can indicate whether a person has already been infected with the virus, i.e. they are used to detect whether a person has previously recovered from a \covid infection. Unfortunately, antibody tests could not provide information on immunity levels or protection against reinfection \citep{FDA2023AntibodyTesting}, limiting their usefulness for public health decision-making. Consequently, the focus shifted to tests that could diagnose active infections, such as PCR and antigenic tests, which were more important for controlling the spread of the virus. 
\end{remark}  
} }
 
\myparagraph{Compartments and Transitions}
In the following, we give a detailed and structured description of the compartments and the transitions of the  base \covid model. The flowchart of the model is presented in Figure \ref{fig:BaseCovid}, while  the transitions with the associated transition vectors and intensities are given in Table \ref{Table_Info_Model1}.

\myparagraph{ $\mathbf S $, Susceptible}
Individuals in $S$ can become infected by infectious individuals who are not in quarantine, i.e., by individuals in $ I^{-} $. They then transition to state $ I^- $, and when they receive the vaccine, they transition to  $ V $. In addition, individuals who have lost their immunity due to a previous infection or vaccination return to  $S$.

\myparagraph{$\mathbf I^{-} $, Hidden infected}  
These individuals may be asymptomatic and unaware of their infection status; they may know they are infected but choose to ``stay calm'' and ``hide''; or they may not be sick enough to go to the hospital and thus remain undetected.
They may test positive (via PCR tests), after which they transition to $ I^{+} $, become ill and go to the hospital $ H $, or simply recover, develop immunity, yet remain hidden and transition to $R^{-}$.
With the introduction of the vaccine, they are vaccinated because their actual status is unknown and they are considered susceptible in the vaccination campaign. They then transition to $V$.

\myparagraph{ $\mathbf R^{-} $,  Hidden recovered}
People in this compartment have recovered from a previous, undetected infection ($I^-$) and are temporarily immune. Over time, they lose their immunity, become susceptible, and return to the $S$ state. They are also eligible for vaccination, as their status is unknown within the context of the vaccination campaign.

\myparagraph{$\mathbf I^{+}$,  Confirmed infected}
People who have tested positive but are not hospitalized are considered as detected or confirmed cases, and form the $I^+$ compartment.
They have two options. Either their condition worsens and they must be admitted to the hospital, meaning they move to $H$, or they recover and move to $R^{+}$ as confirmed cases. For simplicity, it is assumed that people in $I^{+}$  are in quarantine and completely isolated, meaning they cannot transmit the virus to others.

\myparagraph{ $\mathbf R^{+}$,  Confirmed recovered}
Individuals in this compartment have recovered from a previous confirmed  infection ($I^+$), or from a confirmed infection treated in the hospital ($H$) or ICU ($C$),  and are temporarily immune. Over time, they lose their immunity, become susceptible, and return to the $S$ state.
In addition, they can be vaccinated regardless of their previous infection status and transition to  $V$.

\myparagraph{$\mathbf V$,  Vaccinated}
The vaccine is administered to individuals whose current epidemiological status is unknown. Specifically, those vaccinated may be susceptible ($S$), infected but hidden ($I^{-}$), or in the hidden recovery phase ($R^{-}$). In addition, individuals who have recovered from a confirmed infection ($R^{+}$) may also be vaccinated. After vaccination, all of these individuals gradually lose their immunity over time and eventually return to the group of susceptible individuals.

\myparagraph{$\mathbf H $, Hospitalized}  
This group includes infected individuals who are admitted to a hospital for inpatient treatment. They have two options: either their condition deteriorates and they require treatment in an ICU, then they move to $C$, or they recover and  move into the $ R^{+} $ compartment, as their recovery state is recorded. 

\myparagraph{ $\mathbf  C $, Intensive care hospitalization}
This compartment admits patients who are in an advanced stage of their illness and require ventilators and specialized medical care in an ICU. Patients in this unit are transferred from the $ H $ unit and face two possible outcomes: they may recover and  move to the $ R^+ $ compartment, or they may succumb to the disease and move to the deceased category $ D $.

\myparagraph{$\mathbf D $, Death from the disease}
People who die from the disease in an intensive care unit are assigned to state $D$, which is an absorbing state.
For simplicity, we assume that only hospitalized patients in ICUs can die from the disease.
The direct transition of infected individuals who are not in the hospital  $ (I^{-},I^{+}) $, or who are already in the $H$ hospital unit, to the $D$ state of the deceased is not taken into account.

\smallskip
 Among the $d=9$ states mentioned above and depicted in the flowchart in Figure \ref{fig:BaseCovid}, there are $d_2=4$ observable states. One of these is $I^+$, since both the inflow of people from $I^-$ who have tested positive  and the outflow due to recovery to $R^+$ can be tracked. In addition, $H,C$ and $D$ are also observable, as they refer to the hospital in which transitions are fully recorded.   The remaining $d_1=5$ states $S, I^-, R^-, R^+,V$ are hidden, as they are associated with at least one unobservable transition.  Although $R^+$ and $V$ show an observable inflow of confirmed recoveries and recorded vaccinations, respectively, the outflow toward $S$ cannot be tracked due to the loss of immunity and remains hidden. We therefore call these two compartments also as ``partially hidden''. In Subsects.~\ref{sec:cascade_model} and \ref{sec:extended_covid} we extend the present base \covid model and introduce additional so-called cascade states in order to use the information from observable inflows $R^+$ and $V$. The other hidden states $S, I^-, R^-$ are characterized by unobservable inflows and  outflows. The outflow to $V$ is observable, since vaccinations are recorded. However, the other outflows, which are associated with the transitions between the hidden compartments $S\to I^-, I^-\to R^-, R^-\to S$, are not observable. Therefore, we also refer to the states $S, I^-, R^-$ as ``fully  hidden''.
As can be seen in Figure \ref{fig:BaseCovid} and Table \ref{Table_Info_Model1}, the transitions either connect hidden and observable components, connect observable components, or remain entirely within the hidden component. This distinction determines which transition noises enter only the hidden-state equation and which ones enter both the hidden and observable equations. The transition intensities follow standard mass-action and individual-transition principles: the infection rate is bilinear in the susceptible and undetected-infected populations, whereas testing, recovery, hospitalization, ICU transfer, death, vaccination, and immunity-loss transitions are linear in the compartment from which individuals leave. These choices fully specify the CTMC transition rates and, through the diffusion approximation, the drift and diffusion coefficients used below.

\myparagraph{State dynamics}
The dynamics of the state process $X=( I^-,R^-,S,R^+,V,I^+,H,C,D)^\top$ of the  base \covid model  depicted in Figure \ref{fig:BaseCovid} are fully specified by  the settings in Table \ref{Table_Info_Model1}.  They define the   CTMC dynamics given  in \eqref{Eq_CTMC} as well as  the   diffusion approximation dynamics for $\XD$ given by the SDE  \eqref{DA1}, and  its decomposition into  the hidden component  $Y= (I^-,R^-,S,R^+,V)^\top$ and the observable component  $Z=(I^+,H,C,D)^\top$ satisfying the  system of coupled SDEs \eqref{fp11}. The block structure of the coefficients $\driftX,\diffX$ with respect to the hidden/observable partition follows directly from \eqref{block_SDE}.
Below we give the coefficients of the coupled system of SDEs \eqref{fp11} for $Y, Z$. Here, as in the previous model examples,   
the coefficients $ \overline{f}, ~ \overline{\sigma},~ \overline{g} ,~ \overline{\ell} $  are  non-linear in the  hidden  state $Y$, while the drift coefficient $\overline h$ of the SDE for $Z$ is linear and of the form $\overline{h}(t,y,z) = \overline{h}_0(t,z) + \overline{h}_1(t,z) y$. 
 The entries of the Brownian motion $W^1$ of dimension $k_1=9$ drive the transitions between the hidden states, which  in Table \ref {Table_Info_Model1} are labeled 1 through 9, while the entries of $W^2$ of dimension $k_2=8$ drive the remaining transitions, labeled  10 through 17, which connect at least one observable state.
The coefficients of the recursions \eqref{Discr_Model} describing the discrete-time state dynamics follow   from  \eqref{Discr_Model_coeff}.
	\begin{align}
		\overline{f}(t,y,z)&=\left(\hspace*{-0.5em} \begin{array}{c}		
			\beta\frac{y^1y^3}{N} -(\alpha + \gamma^{-} + \eta^{-} + \mu )y^1\\
			\gamma^{-}y^1 -\mu y^2 -\rho^{-}_{1}y^2\\
			-\beta\frac{y^1y^3}{N} -\mu y^3 +\rho^{-}_{1}y^2 +\rho^{-}_{2}y^4 +\rho^{V}y^5\\[1ex]
			\gamma^{+}z^1 + \gamma^{H}z^2 + \gamma^{C}z^3 - \rho^{-}_{2}y^4 - \mu y^4\\
			(y^1 + y^2 + y^4 + y^3)\mu - \rho^{V}y^5
		\end{array} \hspace*{-0.5em}\right),~
		\overline{h}_{0}(t,z)=\left(\hspace*{-0.5em}\begin{array}{c}
			-\gamma^{+}z^1 - \eta^{+}z^1\\[0.7ex]
			\eta^{+}z^1 - \delta z^2 - \gamma^{H}z^2\\
			\delta z^2 - \gamma^{C}z^3 -\kappa z^3\\
			\kappa z^3
		\end{array} \hspace*{-0.5em}\right), ~
		\overline{h}_{1}(t,z)=\left(\hspace*{-0.5em}\begin{array}{c@{\hspace{0.15em}}c@{\hspace{0.15em}}c@{\hspace{0.15em}}c@{\hspace{0.15em}}c}
			\alpha & 0 & 0&0 &0\\[0.5ex]
			\eta^{-} & 0 & 0 & 0 & 0\\[0.5ex]
			0 & 0 & 0 & 0 & 0\\[0.5ex]
			0 & 0 & 0 & 0 & 0
		\end{array} \hspace*{-0.5em}\right),\\[-2ex]
		\label{covid_base_sde_coefficients}
		\overline{\sigma}(t,y,z)&=
		\left(\begin{array}{*{9}{@{\hspace{0.2em}}c}@{}}
			\sqrt{\beta y^1y^3/N} & -\sqrt{\gamma^-y^1} & 0 & 0 & 0 & 0 & -\sqrt{\mu y^1} & 0 & 0\\
			0 & \sqrt{\gamma^-y^1} & -\sqrt{\rho^-_1y^2} & 0 & 0 & 0 & 0 & -\sqrt{\mu y^2} & 0\\
			-\sqrt{\beta y^1y^3/N} & 0 & \sqrt{\rho^-_1y^2} & \sqrt{\rho^-_2y^4} & \sqrt{\rho^Vy^5} & -\sqrt{\mu y^3} & 0 & 0 & 0\\
			0 & 0 & 0 & -\sqrt{\rho^-_2y^4} & 0 & 0 & 0 & 0 & -\sqrt{\mu y^4}\\
			0 & 0 & 0 & 0 & -\sqrt{\rho^Vy^5} & \sqrt{\mu y^3} & \sqrt{\mu y^1} & \sqrt{\mu y^2} & \sqrt{\mu y^4}\end{array}\right),\\[1ex]
		\overline{g}(t,y,z)&=
		\left(
		\begin{array}{ccccc@{\hspace*{2.8em}}c@{\hspace*{3.6em}}c@{\hspace*{3.5em}}c}		
			-\sqrt{\alpha y^1} & -\sqrt{\eta^{-}y^1} & 0 & 0 & 0 & 0 & 0 & 0\\
			0 & 0 & 0 & 0 & 0 & 0 & 0 & 0\\
			0 & 0 & 0 & 0 & 0 & 0 & 0 & 0\\
			0 & 0 & \phantom{-}\sqrt{\gamma^{+}z^1} & \phantom{-}\sqrt{\gamma^{H}z^2} & \phantom{-}\sqrt{\gamma^{C}z^3} & 0 & 0 & 0\\
			0 & 0 & 0 & 0 & 0 & 0 & 0 & 0
		\end{array} \right),\\[1ex]
		\overline{\ell}(t,y,z)&=
		\left(
		\begin{array}{cccccccc}			
			\phantom{-}\sqrt{\alpha y^1} & 0 & -\sqrt{\gamma^{+}z^1} & 0 & 0 & -\sqrt{\eta^{+}z^1} & 0 & 0 \\
			0 & \phantom{-}\sqrt{\eta^{-}y^1} & 0 & -\sqrt{\gamma^{H}z^2} & 0 & \sqrt{\eta^{+}z^1} & 0 & -\sqrt{\delta z^2}\\
			0 & 0 & 0 & 0 & -\sqrt{\gamma^{C}z^3} & 0 & -\sqrt{\kappa z^3} & \sqrt{\delta z^2}\\
			0 & 0 & 0 & 0 & 0 & 0 & \sqrt{\kappa z^3} & 0
		\end{array} \right).
	\end{align}

\subsection{Cascade State Modeling }\label{sec:cascade_model}
The base \covid model comprises $d=9$ states, including $d_1=5$ hidden states. Three of these hidden states ($S, I^-, R^-$) are referred to as ``fully hidden'', since the corresponding compartments have an inflow and at least one outflow that are not observable. The other two hidden states, which correspond to the compartments of detected recovered individuals ($R^+$) and vaccinated individuals ($V$) and are highlighted in orange in the flowchart of the model shown in Figure \ref{fig:BaseCovid}, are referred to as ``partially hidden''. Here, the inflow of individuals who have recovered from a confirmed infection and are in quarantine ($I^+$) , from the hospital ($H$), or from the  ICU ($C$) into $R^+$  can be tracked. The same applies to the inflow of vaccinated people into $V$, since vaccinations are reported. However, the outflow from both $R^+$ and $V$ to $S$ due to loss of immunity is not observable and remains hidden,  since we have no information on the timing of immunity loss. 	
This modeling approach has the disadvantage that, when statistically estimating the dark figures associated with hidden states  as considered, for example, in \cite{ouabo2025estimating}, \cite{achola2025stochastic}, and \cite{OuaboKamkumo2025}, the available information from the observable inflows to $R^+$ and $V$ cannot be taken into account. To address this limitation, we propose the following model extension.

Our approach is based on the assumption that, for many diseases, particularly  \covid, it is known that vaccination and recovery from the disease provide complete immunity only for a specific, known period of time. This is followed by a phase of waning immunity and, ultimately, a complete loss of immunity. From that point on, those individuals  must be considered susceptible to infection. Based on this information regarding the duration of full immunity, the idea is to take into account the ``recovery age'' and ``vaccination age''  that is the elapsed time since recovery or vaccination. In addition, we augment the state vector and introduce subgroups that include individuals with the same vaccination and recovery age. 

\subsubsection{Generic Cascade States}\label{sec:cascade_model_intro}
In order to simplify the presentation and to avoid repetitions, we explain our approach for a ``generic compartment'' denoted by $G$  that is  partially hidden and has an observable inflow and an unobservable outflow, as shown in Figure \ref{ch6_cascade1}. This also makes it easier to apply the method to models for other diseases. In the  \covid model under consideration, the symbol $G$ serves as a placeholder for the two states $R^+$ and $V$.  Note that the observable inflow does not necessarily originate from an observable state. For example, the inflow to $V$ stems exclusively from hidden states. However, vaccinations are reported, and our aim is to capture this information that is not directly contained in the subpopulation sizes related to the observable states. For simplicity, we assume that the unobservable outflow from $G$ is directed exclusively toward hidden states, as is the case in the basic \covid model.

The ``$G$-age'' is the time elapsed since the transition to this state. It is supposed to take only discrete values that are multiples of some $\Delta\tau>0$. This ensures that the augmented state remains  finite-dimensional. For simplicity, we assume that the horizon time $T$ is a multiple of $\Delta \tau$, i.e., $T=N_\tau\Delta \tau$ for some $N_\tau\in\N$.  A natural choice is $\Delta\tau=\Delta t$ when working with the time-discretized version \eqref{Discr_Model} of the SDEs \eqref{fp11} and a time step $\Delta t$. Then we have $N_\tau=N_t$.
Furthermore, full immunity persists after the transition to the $G$-state for a period of length $\LG \Delta \tau$, i.e., up to the $G$-age of    $\LG \Delta \tau$, where $\LG \in\N$ is given. 

\begin{table}[htb]
	\begin{center}
		\begin{tikzpicture}[>=stealth]
			\draw[->,line width=1.6pt,blue] (-4.0,0) -- (-1.15,0)
			node[midway,above=3pt,text=blue] {Observed inflow};
			\node[draw=orange!85!red,fill=orange!14,text=orange!85!red,
			line width=1.6pt,minimum width=2.2cm,minimum height=0.9cm]
			at (0,0) {$G$};
			\draw[->,line width=1.6pt,red] (1.15,0) -- (4.0,0)
			node[midway,above=3pt,text=red] {Unobserved outflow};
			\node[above=2pt] at (0,0.45) {{\color{black} Original partially hidden compartment}};
		\end{tikzpicture}\par\bigskip
		{\color{black}
\resizebox{\linewidth}{!}{%
\begin{tikzpicture}[>=stealth,
 box/.style={draw=cascpurple,fill=cascpurple!15,line width=.9pt,
 minimum height=1.15cm,minimum width=1.65cm,align=center,font=\normalsize},
 terminal/.style={box,draw=cascbrown,fill=cascbrown!10}]
\node[box] (g1) at (0,0) {age $\Delta\tau$\\[4pt]{\Large\textcolor{cascgreen}{$G_1$}}};
\node[box] (g2) at (2,0) {age $2\Delta\tau$\\[4pt]{\Large\textcolor{cascgreen}{$G_2$}}};
\node (dots1) at (3.55,0) {$\cdots$};
\node[box] (gi) at (5.1,0) {age $i\Delta\tau$\\[4pt]{\Large\textcolor{cascgreen}{$G_i$}}};
\node (dots2) at (6.65,0) {$\cdots$};
\node[box] (gl) at (8.2,0) {age $\LG\Delta\tau$\\[4pt]{\Large\textcolor{cascgreen}{$G_{\LG}$}}};
\node[terminal] (gm) at (10.55,0) {age $>\LG\Delta\tau$\\[4pt]{\Large\textcolor{red}{$G^-$}}};
\foreach \a/\b in {g1/g2,g2/dots1,dots1/gi,gi/dots2,dots2/gl,gl/gm}
  \draw[->,black] (\a.east)--(\b.west);
\node[left=1.15cm of g1,align=right,text=blue,font=\large] (in) {observed\\inflow};
\draw[->,blue] (in.east)--(g1.west);
\node[right=1.15cm of gm,align=left,text=red,font=\large] (out) {unobserved\\outflow};
\draw[->,red] (gm.east)--(out.west);
\end{tikzpicture}}}

		\captionof{figure}{{\color{black} Decomposition of $G$ into $\LG$ observable age compartments $G_1,\ldots,G_{\LG}$ and the hidden terminal state $G^-$.  Observed inflow enters $G_1$, and deterministic transfers occur at accounting dates. Unobserved immunity-loss outflow occurs only from $G^-$ after the full-immunity period $\LG\Delta\tau$.}}
		\label{ch6_cascade1}
	\end{center}
\end{table}

To capture the $G$-age, the $G$ compartment is split into a sequence or ``cascade'' of  $\LG$ new compartments $G_1,\ldots,G_{\LG}$, comprising individuals of the same $G$-age, as shown in Figure \ref{ch6_cascade1}. The first of these cascade compartments $G_1$ receives  the  observable inflow to the former partially hidden $G$ compartment, i.e., the  inflow of individuals as they enter the phase of full immunity. Thus, $G_1$ is now an observable compartment. If there are inflows from multiple compartments, they are aggregated and form a single inflow.   In the \covid  model, these are people who have recovered from a confirmed infection or who have been vaccinated.  

The former $G$ compartment is then replaced by the $G^-$ compartment including  individuals with waning immunity and $G$-age  exceeding  $\LG \Delta \tau$. From this point on, there is an unobservable outflow of individuals who have lost their immunity,  so that $G^-$, as the former $G$ compartment, is partially hidden. The outflow from $G^-$ replaces the outflow from the former $G$ compartment. However, the transition rates  for the loss of immunity must be adjusted to account for the elapsed time of complete immunity.
In the \covid  model, they become susceptible again and transition to state $S$. 

We assume that individuals transition to the next cascade compartment or from $G_{\LG}$ to $G^-$ only every $\Delta \tau$ time units, after their $G$-age has increased by $\Delta \tau$. These are simple, deterministic, and observable transitions.

In summary, we can conclude that the cascade state approach increases the dimension $d$ of the state vector $X$ by $\LG$. The new states $G_1,\ldots, G_{\LG}$ are observable and are assigned to the state component $Z$, while the state $G^-$, which emerged from the previously hidden state $G$, is again hidden and is assigned to the state component $Y$.

\subsubsection{Cascade State Dynamics}
\label{sec:cascade_dynamics}
We will now derive equations that describe the dynamics of the cascade states contained in the extended state vector and resulting from the decomposition of the original state $G$ explained above.
First we formulate such equations in continuous time  and adopt a \emph{hybrid continuous-discrete} approach. A further time-discretization leads to recursions that describe the time-discrete dynamics. To do this, we introduce equidistant accounting dates $\tau_l = l\,\Delta\tau$, $l = 0,1,\ldots,N_\tau$, and assume that the $G$-age of the individuals increases by $\Delta\tau$ only at these dates, while remaining constant between two such dates. We denote by $G_1,\ldots,G_{\LG}, G^-$ the number of individuals in the respective compartments.

\myparagraph{Continuous-time dynamics between two accounting dates}
The population sizes $G_2,\ldots,G_{\LG}$ in the ``internal cascade compartments''  remain constant, while   $G_1$ and $G^-$ follow the usual continuous-time dynamics.  $G_1$ receives the random but observable inflow  that, in the original model, was directed to the partially hidden compartment $G$,  while $G^-$ emits the hidden outflow. 
However, the outflow transition rate for the loss of immunity and the transition to $S$ has to be modified such that the known and fixed period of full immunity spent in ``internal cascade compartments'' is taken into account. Let $\rho^G$ denote the original rate for the loss of immunity for individuals in state $G$. 
Then under the assumptions of the underlying stochastic CTMC model, the  residence time in state $G$ before the loss of immunity and transition to $S$ follows an exponential distribution with a mean of $1/\rho^G$.	We make the reasonable assumption that this average time exceeds the period of full  immunity, i.e., $1/\rho^G>\LG\Delta\tau$. A simple and practical approach is to adjust the rate $\widetilde{\rho}^{\,G}$ for the transition from the final state of the cascade $G^-$ to $S$  so that the average residence time in the cascade states $G_1,\ldots, G_{\LG},G^-$ matches that of the original state $G$. This leads to the equation   $1/\rho^G= \LG\Delta\tau +1/\widetilde{\rho}^{\,G}$ and the
	modified outflow rate
 \begin{align} \label{cascade_outflow_rate}
 	\widetilde{\rho}^{\,G}=(1/\rho^G-\LG\Delta\tau)^{-1}.
 \end{align}
Please note that the transition rates from $G^-$ to other states, for example, in the context of vaccinations, do not need to be adjusted, as they are not affected by the loss of immunity.	 
Furthermore, the states associated with the other compartments of the model follow their usual continuous-time dynamics.

\myparagraph{Updates at accounting dates}
At each accounting date $\tau_l$ the G-age increases by $\Delta \tau$, leading to a transition of individuals to the next  cascade compartment described by 
\begin{align}\label{cascade_update}
G_{k+1}(\tau_l)=G_k(\tau_l-),\quad k=1,\ldots,\LG-1,\qquad
G_1(\tau_l)=0,\qquad 
G^-(\tau_l)=G^-(\tau_l-)+G_{\LG}(\tau_l-).
\end{align}
Here, $G_k(\tau_l-), G^-(\tau_l-)$ denote the population sizes immediately before the update.  The first ``cascade compartment'' $G_1$ is completely emptied, i.e.,  $G_1(\tau_l)=0$. It will then accommodate the individuals that arrive continuously over time  until the next accounting date $\tau_{l+1}$. In the last “cascade compartment” $G^-$, the population size at $\tau_l$ is increased by the inflow of individuals from $G_{\LG}$ with G-age $\LG \Delta \tau$, i.e., the group that was previously the oldest with complete immunity is now at risk of losing its immunity.	
The dynamics of the remaining compartments of the model is not affected by these updates.

\myparagraph{Grouping of cascade compartments}
Depending on the duration of complete immunity and the size of $\Delta \tau$, the number $\LG$ can become quite large, and the decomposition of $G$ described above generates an excessive number of new compartments. In this case, a solution is to accept a certain loss of temporal resolution in recording the $G$-age and to aggregate several consecutive cascade compartments into a single one.

Let $\KG\in\N$ with $\KG\le\LG$ be the desired number of grouped cascade compartments, and let $\PG_1,\ldots,\PG_{\KG}\in\mathbb N$ be the numbers of consecutive original compartments aggregated in each group, so that $\PG_1+\cdots+\PG_{\KG}=\LG$.
The first grouped compartment, again denoted by $G_1$, contains individuals with a $G$-age between $1$ and $\PG_1$ time steps of length $\Delta\tau$. The compartment $G_2$ covers ages from $\PG_1+1$ to $\PG_1+\PG_2$ time steps, and so forth. The last grouped compartment $G_{\KG}$ contains the oldest complete-immunity ages, up to $\LG$ time steps. Figure \ref{ch6_cascade2} illustrates this construction.

\begin{table}[htbp]
	\begin{center}		
\resizebox{\linewidth}{!}{%
\begin{tikzpicture}[>=stealth,
 age/.style={draw=cascpurple,fill=cascpurple!15,line width=.9pt,
 minimum height=.95cm,minimum width=1.4cm,align=center,font=\Large,inner sep=4pt},
 group/.style={draw=cascmagenta,fill=cascpink,rounded corners=4pt,inner sep=5pt,line width=.9pt}]
\node[age] (a11) at (0,0) {age\\[2pt]$\Delta\tau$};
\node (d1) at (1.15,0) {$\cdots$};
\node[age] (a12) at (2.4,0) {age\\[2pt]$\PG_1\Delta\tau$};
\node[age] (a21) at (5.2,0) {age\\[2pt]$(A_1+1)\Delta\tau$};
\node (d2) at (6.9,0) {$\cdots$};
\node[age] (a22) at (8.3,0) {age\\[2pt]$A_2\Delta\tau$};
\node (between) at (10.1,0) {$\cdots$};
\node[age] (ak1) at (12.1,0) {age\\[2pt]$(A_{\KG-1}+1)\Delta\tau$};
\node (dk) at (14.2,0) {$\cdots$};
\node[age] (ak2) at (15.7,0) {age\\[2pt]$\LG\Delta\tau$};
\begin{scope}[on background layer]
\node[group,fit=(a11)(d1)(a12)] (g1) {};
\node[group,fit=(a21)(d2)(a22)] (g2) {};
\node[group,fit=(ak1)(dk)(ak2)] (gk) {};
\end{scope}
\foreach \a/\b in {a11/d1,d1/a12,a21/d2,d2/a22,ak1/dk,dk/ak2}
 \draw[->,black] (\a.east)--(\b.west);
\node[below=6pt of g1,text=cascmagenta,font=\huge] {$G_1$};
\node[below=6pt of g2,text=cascmagenta,font=\huge] {$G_2$};
\node[below=6pt of gk,text=cascmagenta,font=\huge] {$G_{\KG}$};
\node[draw=cascbrown,fill=cascbrown!10,line width=.9pt,minimum height=.95cm,
 align=center,font=\Large,right=1.1cm of gk] (gm)
 {age\\[2pt]$>\LG\Delta\tau$};
\coordinate (gmb) at (gm.south |- gk.south);
\node[below=6pt of gmb,text=cascmagenta,font=\huge] {$G^-$};
\draw[->,black] (g1.east)--(g2.west);
\draw[->,black] (g2.east)--(between.west);
\draw[->,black] (between.east)--(gk.west);
\draw[->,black] (gk.east)--(gm.west);
\node[left=1.3cm of g1,align=right,text=blue,font=\Large] (in) {Observed\\inflow};
\draw[->,blue] (in.east)--(g1.west);
\node[right=1.3cm of gm,align=left,text=red,font=\Large] (out) {Unobserved\\outflow};
\draw[->,red] (gm.east)--(out.west);
\end{tikzpicture}}

		\captionof{figure}{Grouping into $\KG$ compartments of sizes $\PG_1,\ldots,\PG_{\KG}$, where $A_k=\sum_{j=1}^k\PG_j$, $A_0=0$, and $A_{\KG}=\LG$. Age ranges in the boxes are in units of $\Delta\tau$.  Each group $G_k$ aggregates adjacent ages, followed by the hidden terminal state $G^-$. }
		\label{ch6_cascade2}
	\end{center}
\end{table}

The updating formulas given in \eqref{cascade_update} have to be modified to account for the fact that, within the  grouped compartments \( G_{k} \), the different individual $G$-ages can no longer be distinguished. 
Whereas \eqref{cascade_update} moves all individuals in each of the $\LG$ original compartments at every accounting date, the grouped approximation transfers only the fraction $\psiG_k=1/\PG_k\in(0,1]$ of the population in $G_k$. The grouped update is
\begin{align}\label{cascade_update_group}
	 	\begin{split}
	 	G_{1}(\tau_{l})&= (1-\psiG_1)G_1(\tau_{l}-),\\
	 	G_k(\tau_{l})&=(1-\psiG_k)G_k(\tau_{l}-)+\psiG_{k-1}G_{k-1}(\tau_{l}-),\quad k=2,\ldots,\KG,\\
	 	G^{-}(\tau_{l})&=G^{-}(\tau_{l}-)+\psiG_{\KG}G_{\KG}(\tau_{l}-).
	 \end{split}	
\end{align}
This approach introduces a ``grouping error'' if the age distribution within the grouped compartments is not uniform. For $\KG=\LG$ and $\PG_1=\cdots=\PG_{\KG}=1$, no grouping occurs and \eqref{cascade_update_group} reduces to the exact update \eqref{cascade_update}.

\myparagraph{Discrete-time dynamics}
The hybrid continuous-discrete approach to model the cascade state dynamics described above can also be integrated into the time discretization of SDE models. 
A natural choice of accounting times is $\tau_l=t_l, l=1,2,\ldots$ and $\Delta \tau=\Delta t$, i.e., updates to the $G$-age occur at each of the discrete time points.	
We refer to the next subsection, in which we derive the stochastic recursions for the discrete-time  dynamics of the state vector $X$ in the extended \covid model, which is decomposed into the hidden and observable components $Y$ and $Z$.

\subsection{Extended \covid Modeling with Cascade States  }\label{sec:extended_covid}
\begin{table}[htbp]
\centering
\includegraphics[width=11cm,height=6cm]{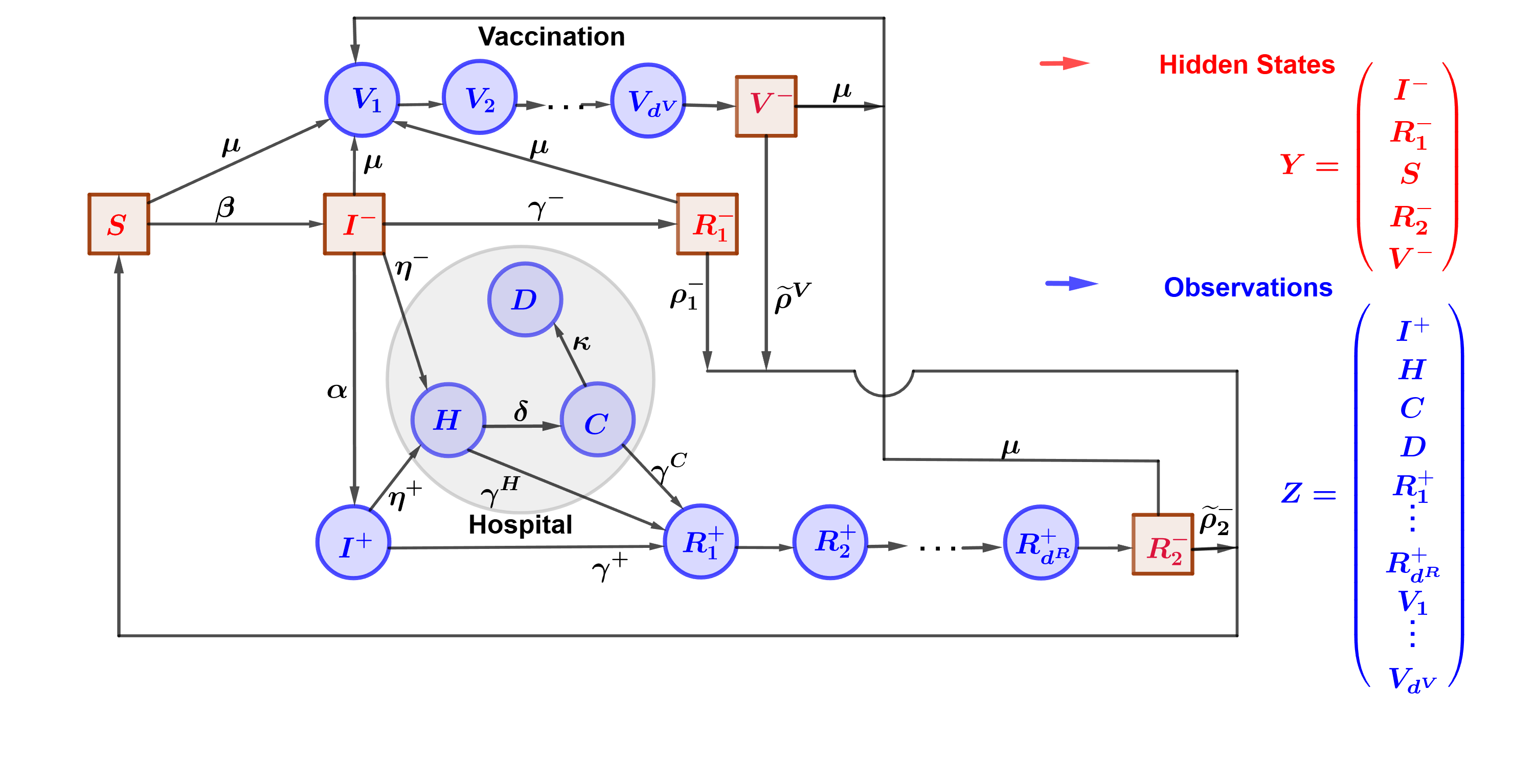}
\captionof{figure}{Flow diagram of the partially observable extended \covid model, 
	$d=9+\KR+\KV$ states, $K=18$ random transitions, and  $\KR+\KV$ deterministic cascade transitions. The states include  $d_1=5$ hidden states ($I^-,R^-_1,S,R^-_2,V^-$, colored red), and $d_2=4+\KR+\KV$ observable states (colored blue) comprising the $\KR$ cascade states $R^+_1,\ldots,R^+_{\KR}$ of recovered individuals with full immunity,  $\KV$ cascade states $V_1,\ldots,V_{\KV}$ of vaccinated individuals with full immunity, and $I^+,H,C,D$.  Of the $K$ random transitions, $k_1=5$ connect hidden states, whereas $k_2=13$ involve at least one observable state.}
\label{fig:ExtendedCovid}
\end{table} 

	In this subsection we apply the  cascade state approach to the base \covid model, see Subsection \ref{sec_base_covid} and Figure \ref{fig:BaseCovid}, and derive a refined model that captures the information from observable inflows into partially hidden states. 
	
	\subsubsection{New compartments and transitions}
	Instead of a single generic state $G$ as in the previous subsection, we now  decompose two 
	 partially hidden  compartments,  $R^+$ and $V$,  into cascades of observable compartments followed by a terminal hidden compartment as depicted in  Figure \ref{fig:ExtendedCovid}. Specifically, $R^+$ is replaced by $R^+_1,\ldots,R^+_{\KR},R^-_2$, and $V$ by $V_1,\ldots,V_{\KV},V^-$. 
	  To simplify the notation, we will occasionally use the placeholder notation $\dagger = R,V$ below to refer to both types of cascade states.  Complete immunity following recovery or vaccination lasts for a period  $\LD\Delta t$. This period is grouped into $\KD\leq\LD$  intervals of lengths $\PD_1,\ldots,\PD_{\KD}\in\N$ with $\PD_1+\ldots+\PD_{\KD}=\LD$, and we set $\psiD_k=1/\PD_k, k=1,\ldots,\KD$.  For the model calibration and the numerical simulation in Sections \ref{D_calibration} and \ref{sec:NumericalResults} we will work with  $\Delta t=1$ day, $\LD=90$ time steps, $\KD=3$ groups, each comprising  $\PD_k=30$ time steps, so that $\psiD_k=1/30, k=1,2,3$. 
	  
	 The former base model state $R^-$, comprising  individuals that recovered from a hidden infection, is now denoted by $R_1^-$, to better distinguish it from the  new terminal cascade state $R^-_2$. Both states are hidden, just as  $R^-$ and $R^+$ are in the base model,  as the loss of immunity after recovery and the return to state $S$ cannot be tracked.  In addition, individuals in states $R^-_1$ and $R^-_2$ can receive a vaccination, since their actual state is unknown to the vaccination campaign. Instead of transitioning to state $V$, which was a hidden state in the base model, they now transition to the new cascade state $V_1$, which is observable. It should be noted that the transition $R^-_2\to V_1$, which replaces the transition of the base model $R^+\to V$, is a transition between the two types of cascade states that is not covered by the  generic cascade state  approach discussed in the previous subsection. 	 
	 
	 Motivated by the fact that, during the \covid pandemic, vaccinations were repeatedly administered because they provide complete immunity only for a short period, we also introduce a new transition from $V^-$ to $V_1$ in the extended model. This means that vaccinated individuals with waning immunity in $V^-$ can receive another vaccination, since  the   full immunity period of $\LV\Delta t$ has expired, and  the actual status of individuals in $V^-$ is unknown within the context of the vaccination campaign.
	 
	  In the extended model, one must take into account the  periods of full immunity following recovery and vaccination, which individuals spend in the  cascade states $R^+_k, V_k$, and adjust the base model's rates of immunity loss $\rho^-_2, \rho^V$ for the transitions $R^-_2 \to S$ and $V^- \to S$, respectively. For this adjustment   we follow the approach in Subsection \ref{sec:cascade_dynamics},  see \eqref{cascade_outflow_rate}, and obtain  
	 	\begin{align}\label{covid_adjusted_immunity}
	 		\widetilde\rho^{\,-}_2&=(1/\rho^-_2-\LR\Delta t)^{-1} \quad\text{and}\quad 
	 		\widetilde\rho^{\,V}=(1/\rho^V-\LV\Delta t)^{-1},
	 	\end{align}
	 	provided $\LR\Delta t<1/\rho^-_2$ and $\LV\Delta t<1/\rho^V$. The rates for the transition from $R^-_2,V^-$ to $V_1$ do not have to be adjusted, as they are given exogenously and characterize the intensity of the vaccination campaign and are therefore uniformly set to $\mu$ for all transitions from the states $S, I^-, R^-_1, R^-_2, V^-$ to $V_1$.

\subsubsection{Discrete-time state dynamics} 
We now derive the  recursions of the form \eqref{Discr_Model} which describe the discrete-time state dynamics of the state vector $X$ of the extended model with its hidden and observable components $Y$ and $Z$. The starting point is the corresponding recursions for the base model, for which the coefficients are derived from \eqref{Discr_Model_coeff} and \eqref{covid_base_sde_coefficients}.
Let $X^B,Y^B,Z^B$ denote the corresponding states of the base model, which are given in Subsection~\ref{sec_base_covid} as 
\begin{align}\label{cascade_state_decomp_base}
	X^B=\binom{Y^B}{Z^B} \quad\text{with}\quad  Y^B=(I^-,R^-,S,R^+,V)^\top,\quad  Z^B=(I^+,H,C,D)^\top.
\end{align}
Note that the components of $Y^B$ are arranged such that the partially hidden states $R^+,V$ with observable inflow and hidden outflow form the last entries.
For the state $X$ of the extended model these two entries  are replaced by the hidden terminal cascade states $R^-_2,V^-$,  and the observable cascade states $R^+_1,\ldots,R^+_{\KR},V_1,\ldots,V_{\KV}$ are concatenated with $Z^B$ to form $Z$, i.e., 
\begin{align}\label{cascade_state_decomp}
	X=\binom{Y}{Z} \quad\text{with}\quad  Y=(I^-,R^-_1,S,R^-_2,V^-)^\top,\quad Z=(I^+,H,C,D,R^+_1,\ldots,R^+_{\KR},V_1,\ldots,V_{\KV})^\top.
\end{align}
We rewrite the recursions \eqref{Discr_Model}  for the base model's state  in the form
\begin{align}\label{covid_base}
	\begin{array}{rlrl}
		Y^B_{n+1}&=f^B+\sigma^B\NoiseB^{1}_{n+1}+g^B\NoiseB^{2}_{n+1},
		&Y^B_0&=y^B_0,\\[1ex]
		Z^B_{n+1}&=h^B+\ell^B\NoiseB^{2}_{n+1},
		&Z^B_0&=z^B_0,
	\end{array}
\end{align}
 with the noise vectors $\NoiseB^1,\NoiseB^2$ of dimensions $k_1^B=9,k_2^B= 8$,  and aim to derive the corresponding recursions for the extended model in the form 
\begin{align}\label{covid_ext}
	\begin{array}{rlrl}
		Y_{n+1}&=f+\sigma\Noise^1_{n+1}+g\Noise^2_{n+1},
		&Y_0&=y_0,\\[1ex]
		Z_{n+1}&=h+\ell\Noise^2_{n+1},
		&Z_0&=z_0.
	\end{array}
\end{align}
For the sake of brevity, we have omitted the dependence of the coefficients on the time index $n$ and the respective state variables  in the recursions above.
	
\myparagraph{Noise vectors} 	Since the base model transitions to $V$ (see transition 6 through 9 in Table  \ref{Table_Info_Model1}), in the extended model are directed toward the observable state $V_1$, they are no longer transitions between hidden states. Thus,  the corresponding entries in the noise vector $\NoiseB^1$ of the base model, namely, the last four entries $\NoiseB^{1,6},\ldots,\NoiseB^{1,9}$, are moved to the noise vector $\Noise^2$ of the extended model, which appears in the recursions for both the state components $Y$ and $Z$.  This vector  is of dimension $k_2=13$ and is formed by concatenating these four entries with $\NoiseB^2$, and finally adding a last entry that drives the new transition $V^-\to V_1$. The remaining entries of $\NoiseB^1$  form the vector $\Noise^1$ of dimension $k_1=5$. Summarizing we obtain the independent Gaussian vectors (suppressing the subscript with time index $n$)
	\begin{align}\label{covid_routed_noise}
		\Noise^1&=(\NoiseB^{1,1},\ldots,\NoiseB^{1,5})^\top\sim \mathcal N(0,\one_5),\quad\text{and}\quad  
		\Noise^2=(\NoiseB^{2,1},\ldots,\NoiseB^{2,8},\NoiseB^{1,6},\ldots,\NoiseB^{1,9},\Noise^{2,13})^\top \sim \mathcal N(0,\one_{13}).
	\end{align}
	\myparagraph{Discrete-time diffusion coefficients} 
	The above structure of the noise vectors allows us to derive the diffusion coefficients $\sigma,g,\ell$ in the recursion \eqref{covid_ext} from the corresponding coefficients  $\sigma^B,g^B,\ell^B$ in the recursions of the base model  \eqref{covid_base}. The latter coefficients are derived from the SDE coefficients given in \eqref{covid_base_sde_coefficients}, which are multiplied by $\sqrt{\Delta t}$, see \eqref{Discr_Model_coeff}.  Note that the diffusion coefficients are not affected by the internal cascade transitions, since these are deterministic transitions.   
	
	In the following,  for a generic matrix, we use the notation $A_{i:j,k:l}$ to denote the block matrix consisting of the entries of $A$ in rows $i$ through $j$ and columns $k$ through $l$, where $i\le j, k\le l$. Further,  $A_{:,k:l}$ and $A_{i:j,:}$ denote the block matrices with entries from all rows and all columns, respectively. Then the diffusion coefficients of the extended \covid model read as
	\begin{align}\label{covid_extended_diffusion_blocks}
		\begin{split}
			{\sigma}={\sigma}^B_{:,1:5}, \qquad  {g}&=\begin{pmatrix}
				~~~{g}^B_{1:3,:}~~~&{\sigma}^B_{1:3,6:9}&0_{3\times1}\\
				0_{1\times8}& ~~~~~~~~~\bigl(0,0,0,-\sqrt{\mu y^4}\bigr)\sqrt{\Delta t}~~~~~~~~~&0\\
				0_{1\times8}&0_{1\times4}&-\sqrt{\mu y^5}\sqrt{\Delta t}
			\end{pmatrix},\\[1ex]
			{\ell}&=\begin{pmatrix}
				{\ell}^B&0_{4\times4}&0_{4\times1}\\
				{g}^B_{4,:}&0_{1\times4}&0\\
				0_{(\KR-1)\times8}&0_{(\KR-1)\times4}&0_{(\KR-1)\times1}\\[1ex]
				0_{1\times8}&\bigl(\sqrt{\mu y^3},\sqrt{\mu y^1},\sqrt{\mu y^2},\sqrt{\mu y^4}\bigr)\sqrt{\Delta t}&\sqrt{\mu y^5}\sqrt{\Delta t}\\
				0_{(\KV-1)\times8}&0_{(\KV-1)\times4}&0_{(\KV-1)\times1}
			\end{pmatrix}.
		\end{split}		
	\end{align}
	Note that since the last four entries of $\NoiseB^1$ related to the base model's transitions to $V$ are now included as entries 9 through 12  of $\Noise^2$, the respective entries in the last four columns of  the matrix $\sigma^B$, i.e., ${\sigma}^B_{1:3,6:9}$, have been moved to the columns  9 through 12  of the matrix $g$. The remaining entries of $\sigma^B$ form the matrix $\sigma$, where in the entries related to the transitions $R^-_2\to S$ and $V^-\to S$ the rates $\rho^-_2,\rho^V$ are replaced by the adjusted rates $\widetilde\rho^{\,-}_2,\widetilde\rho^{\,V}$ from \eqref{covid_adjusted_immunity}. 
	
	\myparagraph{Discrete-time drift coefficients} For the drift coefficients $f,h$ which determine the conditional mean of the hidden and observable state component $X_{n+1}$ given $X_n$ via  $\E[Y_{n+1}|X_n]= f(n,Y_n,Z_n)$ and $\E[Z_{n+1}|X_n]= h(n,Y_n,Z_n)$ we have to combine the expressions following from the SDE discretization, see \eqref{covid_base_sde_coefficients}, with the updating formulas given in \eqref{cascade_update_group}. The latter lead to additional terms for the entries related to the terminal cascade states $R^-_2, V^-$, which form the last two entries of $Y$, and the observable cascade states  $R^+_1,\ldots,R^+_{\KR},V_1,\ldots,V_{\KV}$, which form  the last $\KR+\KV$ entries of $Z$. Using for a generic vector $a$ and $i\le j$ the notation $a_{i:j}$ to denote the vector $(a_i,\ldots,a_j)^\top$, and $\psiR_k,\psiV_k$ to denote the weighting factors for the updating of $R^+_k,V_k$ in \eqref{cascade_update_group},  the coefficient $f$ is given by 
	\begin{align}\label{covid_extended_means}
		f_{1:3}=f^B_{1:3}, \quad f_4=\big(1-(\widetilde\rho^{\,-}_2+\mu)\Delta t\big)y^4+\psiR_{\KR}z^{4+\KR}, \quad 
		f_5 = \big(1-(\widetilde\rho^{\,V}+\mu)\Delta t\big)y^5+\psiV_{\KV}z^{4+\KR+\KV}.		
	\end{align}
	It can be seen that $f_4$ includes the outflow from  $R^-_2$ to $V_1,S$, and $f_5$ includes the outflow from  $V^-$ to $V_1,S$.  Accordingly, in the inflow terms $\rho^-_2y^4+\rho^Vy^5$ of $f^B_3$ (entry for $S$) the rates are replaced by $\widetilde\rho^{\,-}_2,\widetilde\rho^{\,V}$.
	For the coefficient $h$ we obtain
	\begin{align}\label{covid_extended_observation_mean}
		\begin{split}
			h_{1:4}=h^B, \hspace*{5em}
			h_5&=(1-\psiR_1)z^5+	(\gamma^+z^1+\gamma^Hz^2+\gamma^Cz^3)\Delta t,\\
			h_{4+k}&=(1-\psiR_k)z^{4+k}+\psiR_{k-1}z^{3+k}, \hspace*{3.8em} 
			k=2,\ldots,\KR,\\
			h_{5+\KR}&=(1-\psiV_1)z^{5+\KR}+\mu\textstyle\sum_{j=1}^5 y^j\,\Delta t,\\
			h_{4+\KR+k}&=(1-\psiV_k)z^{4+\KR+k}	+\psiV_{k-1}z^{3+\KR+k}, \quad 
			k=2,\ldots,\KV.
		\end{split}		
	\end{align}
	Here, $h_5$ contains the observable inflow of recoveries to $R_1^+$ from $I^+,H,C$, and $h_{5+\KR}$ the observable inflow of vaccinations to $V_1$ from five hidden compartments.  Note that $h$  remains affine linear in $y$, since $h^B$  and the additional entries from the updating of cascade states are affine linear.

\section{Model Calibration}\label{D_calibration}
In this section, we calibrate the extended \covid model introduced in the previous section and depicted in Figure~\ref{fig:ExtendedCovid} using epidemiological data from Germany.  We focus on adjusting key model parameters, specifically the time-dependent transition rates for transmission, tests, and vaccination,  so that they are consistent with the observed course of the epidemic. In Section~\ref{sec:NumericalResults}, we show that the calibrated model enables realistic simulations of the spread of the disease that are consistent with observed patterns. Thus, well-calibrated stochastic epidemic models can be used to evaluate various public health strategies, such as vaccination campaigns, social-distancing guidelines, and lockdown measures.

\FloatBarrier
\subsection{Data Set}\label{sec:CalibrationDataSet}
In Germany, local health authorities collect infection data in accordance with the ``Infektionsschutzgesetz (IfSG)'' and forward this data to the national public health institute, the ``Robert Koch-Institut (RKI)''. The primary dataset for our analysis covers a three-year period,  from March 20, 2020, through March 20, 2023. It includes the daily incidence figures reported to local health authorities as well as the vaccination data reported to the RKI. Additionally, we incorporate values of the current effective reproduction number $\Reproeff$ published by the RKI as they provide a comprehensive picture of infection dynamics,  for details, please see Remark \ref{rem:reproduction_numbers} below.  The daily number of cases and the $\Reproeff$ values are taken from the RKI dataset \citep{anderheiden2023nowcasting,anderheiden2020nowcasting}. Vaccination data also come from the RKI dataset \citep{rki2023impfungen}.

The daily reporting suggests using discrete time points $t_n = n$, where $n = 0, 1, \ldots$, and a time step of $\Delta t = 1$ day.
Due to delays in data transmission between the testing centers, local health authorities, and the RKI, on day $n$ not only is new data added, but the data for the preceding days ${n-1}, {n-2}, \ldots$ is also updated. This means that the number of reported cases is underestimated, particularly for the most recent time points. Therefore, one has to account for this data artifact.  For example, one may choose to disregard some of the most recent data points, as they are most affected by this artifact. Another solution is to estimate the systematic bias based on datasets published in the past; we refer to \citet{refisch2022data}. That study examines how the data can be adjusted to avoid a bias in the predictions toward lower incidence rates.  Since we focus on calibrating the model parameters, we disregard the most recent data points, which are most affected by this reporting artifact.

\FloatBarrier
\subsection{Estimation of  Model Parameters}\label{sec:CalibrationModelParameters}
The calibration method we introduce here can be viewed as a hybrid approach because it combines real epidemiological data from the RKI with the model dynamics of the extended \covid model to estimate time-dependent key model parameters. We will work with the recursions \eqref{covid_ext} describing the discrete-time dynamics of the state vector
\begin{align}
X=\binom{Y}{Z}		=\bigl(I^-,R^-_1,S,R^-_2,V^-,I^+,H,C,D,	R^+_1,R^+_2,R^+_3,V_1,V_2,V_3\bigr)^\top
\end{align}
that is decomposed into the hidden component $Y=(I^-,R^-_1,S,R^-_2,V^-)^\top$ and the observable component $Z$ containing the remaining entries. The transition rates for transmission, tests, and vaccination, $\beta,\alpha,\mu$, appearing in the coefficients of that recursion are assumed to be time-dependent, and we will write them as $\beta_n,\alpha_n,\mu_n$, respectively.

\ifthenelse{\preprintOn=1}{{\prepcolor%
\begin{remark}\label{rem_time_dependent} \quad \textbf{Time-dependent Transition Rates}\\
There are various reasons to consider parameters of epidemic models as time-dependent. First of all, government intervention measures: during the \covid pandemic, authorities implemented social distancing, mask requirements, testing campaigns, vaccination rollouts, and travel restrictions, all of which significantly impact the transmission rate $\beta$, the test rate $\alpha$, and the vaccination rate $\mu$. Secondly, seasonal variability: diseases such as \covid exhibit seasonal patterns in transmission influenced by temperature, humidity, and human behavior, requiring time-dependent transmission rates. Another important reason is the emergence of new virus variants: in the case of \covidk, the Alpha, Beta, Gamma, Delta, and Omicron variants each produced noticeable changes in transmissibility. In addition, public adherence to preventive measures changes over time in response to shifting risk perceptions, restriction fatigue, and evolving government communications. Finally, population heterogeneity in susceptibility, contact patterns, and healthcare capacity may vary due to population mobility, demographic changes, or the progression of the epidemic itself.
\end{remark}
}}{}%

 Table \ref{Table_Param_Covid19-1} comprises all parameters of the extended \covid model. As announced, we focus here on the three time-dependent parameters $\beta_n,\alpha_n, \mu_n$.  The remaining parameters are chosen or assumed according to the information provided in the cited references.
 The idea of the calibration is the following. Starting with initial values of the observable states $Z_0$, initial guesses for the hidden states $Y_0$, and given the assumed values of the model parameters (except $\beta_n,\alpha_n,\mu_n$), we consider at each day the mean-square optimal estimate of the state values of the next day, which is known to be given by the conditional expectation $\E[X_{n+1}|X_n]$.  Using the recursions \eqref{covid_ext} for the discrete time state dynamics, we obtain
 \begin{align}\label{Covid-19_dynamic-mean}
 	\E\big[X_{n+1}|X_n\big] = \E\Big[\binom{Y_{n+1}}{Z_{n+1}} \Big|\binom{Y_n}{Z_n}\Big]&=
 	\begin{pmatrix}	f(n,Y_n,Z_n) \\[0.5ex] h(n,Y_n,Z_n)\end{pmatrix},
 \end{align}
 since  the  noise terms in the recursions \eqref{covid_ext} are additive, and the coefficients $\sigma,g,\ell$ in front of the Gaussian vectors $\Noise^1_n,\Noise^2_n$ are $X_n$-measurable and independent of these random vectors. Therefore, these terms vanish under the conditional expectation.
 Note that the ``diffusion coefficients''  $\sigma,g,\ell$, which no longer appear in the conditional expectation above, do not depend on any additional model parameters not included in the  ``drift coefficients'' $f,h$.

 The resulting expressions of the state prediction for the next day in \eqref{Covid-19_dynamic-mean} still depend on the parameters to be estimated, $\beta_n,\alpha_n,\mu_n$.
 They can be determined by adjusting the observed daily data on newly reported infections and vaccinations to the predicted increments in the state components $I^+$ and $V_1$, i.e., $\E[I^+_{n+1}|X_n]-I^+_n$ and $\E[V_{1,n+1}|X_n]-V_{1,n}$. In addition, the reported values for the effective reproduction number $\Reproeff$ can be used to specify a third condition that the three parameters must satisfy,  details are given  below in Remark \ref{rem:reproduction_numbers}.  
 Solving the resulting system of equations yields the desired parameter estimates for this day, denoted by $\widehat \beta_n,\widehat \alpha_n,\widehat  \mu_n$. They are then used in the next iteration step for the computation of $\E[X_{n+2}|X_{n+1}]$  where the entries of $X_{n+1}$, which represent $I^+_{n+1}$ and $V_{1,n+1}$, are updated based on the data, while the remaining entries are replaced by their conditional mean from the previous step. Then, the  parameter estimates for the next day can  be derived.
 This iteration is followed by a second step in which we smooth the  raw estimates from the previous iteration over time to reduce the effects of observation errors.
 
 \begin{remark}\label{rem:reproduction_numbers}\quad\textbf{Basic and effective reproduction number}\\
 The effective reproduction number $\Reproeff$ is related to the basic reproduction number, which is usually denoted by $\Repronull$. The latter indicates the average number of secondary infections caused by a single infected person in an otherwise fully susceptible population. It is a commonly used metric to assess whether an outbreak of the disease will occur and what proportion of the population must be immunized, e.g., through vaccination, to eradicate the disease. 
 
 Later in the course of the epidemic, the population is no longer fully susceptible, as there are already larger numbers of infected or of recovered and immune individuals.
 Therefore, the effective reproduction number $\Reproeff$, which takes into account the actual proportion of susceptible individuals, is used to quantify the current spread of the disease within the population. Following the next-generation approach for reproduction numbers, see  \citep{diekmann1990definition,van2002reproduction},  and assuming perfect isolation of confirmed  infections in $I^+$,  the effective reproduction number at day $n$ in the extended \covid model  is obtained as 
 		\begin{equation}\label{reproduction_number_t}
 			\Reproeffn=\frac{S_n}{N}\frac{\beta_n}{\alpha_n+\gamma^-+\eta^-+\mu_n}.
 		\end{equation}
  Here $\beta_n S_n/N$ is the infection-generation rate per hidden infectious individual at time $n$, and the denominator is its total exit rate through testing, recovery, hospitalization, or vaccination. 
 \end{remark}

\begin{table}[h]
	\begin{center} 
		\footnotesize
		\phantom{\citep{bhapkar2020revisited}}\\[-2ex]
		\begin{tabular}{ c |l| c | l }
			Parameters & Description & Value  & Reference \\ \hline
			$\beta_n$ & Rate of transmission & \multicolumn{2}{l}{Time-dependent  } \\
			$\alpha_n$ & \phantom{Rate of} tests &\multicolumn{2}{l}{Fitted to German \covid data,} \\
			$\mu_n$ & \phantom{Rate of}  vaccination & \multicolumn{2}{l}{between 2020 and 2023} \\\hline
			$\gamma^{-}$ & Rate of recovery for hidden infected & $ 1/14 $ & \citep{bhapkar2020revisited}\\
			$\gamma^{+}$ & \phantom{Rate of recovery for}  confirmed infected & $ 1/14 $ &  \citep{bhapkar2020revisited}\\
			$\gamma^{H}$ & \phantom{Rate of recovery for}  hospitalized & $ 0.048$ &  \citep{charpentier2020covid}\\
			$\gamma^{C}$ & \phantom{Rate of recovery for}  ICU & ${ 1/60}$ & Assumed\\
			$\eta^{+}$ & Rate of hospitalization for confirmed infected & $ 0.0023$ &  \citep{charpentier2020covid}\\
			$\eta^{-}$ & \phantom{Rate of hospitalization for} hidden infected & $ 0.0023$ &  \citep{charpentier2020covid}\\
			$\delta$ & Rate of transfer to ICU  & ${ 1/25}$ & Assumed\\
			$\kappa$ & Rate of death  & ${ 1/10}$ & Assumed\\
			$\rho^{-}_{1}$ & Rate of losing immunity for hidden recovered &  $ 1/120$ & Assumed\\
			 $\rho^{-}_{2}$ & Rate of losing immunity for confirmed recovered  in $R^+$ (base model)&  $ 1/120$ & Assumed\\[0.5ex]
			$\widetilde\rho^{\,-}_{2}$ & \phantom{Rate of losing immunity} for confirmed recovered  in $R_2^-$ & $ 1/30$ &  Derived  from \eqref{covid_adjusted_immunity}\\[0.5ex]
			 $\rho^{V}$ &  Rate of losing immunity for vaccinated in $V$ (base model)&  $ 1/290$ & Assumed\\
			$\widetilde\rho^{\,V}$ & \phantom{Rate of losing immunity} for vaccinated  in $V^-$ & $ 1/200 $ &  Derived from \eqref{covid_adjusted_immunity}\\
			$\overline c_1,\overline c_2$ & Tuning parameters in Algorithm \ref{param_estim} &  $0.1,\;0.01$ &  Assumed\\
			\hline
			&&&\\[-1.5ex]
			$N$ &  Total population size  & ${ 8.3\times10^7}$ & { Germany} \\
			$T$ &   Horizon time & $3$ years & Assumed\\
			$\Delta t$ &   Time step & $1$ day & Assumed\\
			$\LV=\LR$ & Number of time steps with complete immunity & $90$ & Assumed\\
			$\KV=\KR$ & Number of grouped cascade compartments & $3$ & Assumed\\
			$\PD_k$,  $\dagger=R,V$ & Number of time steps in each grouped compartment & $30$ & Assumed\\[0.3ex]
			$\psiD_k$, ~~$k=1,2,3$ & Weighting factors & $1/30$ & Derived from $\psiD_k=1/\PD_k$ \\
			\hline

		\end{tabular}
	\end{center}
	\caption{Model parameter values for the extended \covid model. The three time-dependent transition rates $\beta_n,\alpha_n,\mu_n$ are estimated from data. All remaining parameters are either literature-based or assumed. 		
		}
	\label{Table_Param_Covid19-1}
\end{table}

Now we formalize the above idea of parameter estimation. Consider the model dynamics described in Equation \eqref{Covid-19_dynamic-mean}.
Let us denote by $\Delta I^+_n, \Delta V_{1,n}$ the numbers of new incidences and vaccinations generated at day $n$ that are reported at day $n+1$.  These observed values are equated to the predicted increments of $I^+,V_1$ for that day, i.e.,
\[\Delta I^+_n		=\E[I^+_{n+1}|X_n] -I^+_n \quad \text{and}\quad \Delta V_{1,n}		=\E[V_{1,n+1}|X_n]-V_{1,n}\]
using the relation \eqref{Covid-19_dynamic-mean}. Further, for the reported effective reproduction numbers $\Reproeffn$ we use the relations given in  \eqref{reproduction_number_t}. This leads to the system of three equations shown below in \eqref{Infection_rate_t}.
   Given the observations $\Delta I^+_n,\Delta V_{1,n}$,  $\Reproeffn$ for $n=0,\ldots, N_t-1$, using the assumed values for the remaining parameters and propagating the conditional mean of the state $X$ according to \eqref{Covid-19_dynamic-mean}, these equations are solved in each time step for the parameters $\beta_n,\alpha_n,\mu_n$ leading to their  raw estimates $\widehat\beta_n,\widehat\alpha_n,\widehat\mu_n$.
 	\begin{align}\label{Infection_rate_t}
 		\begin{split}
 			\Delta I^+_n		&=\alpha_n I^-_n-(\eta^++\gamma^+)I^+_n,\\
 			\Delta V_{1,n}		&=\mu_n\bigl(I^-_n+R^-_{1,n}+S_n+R^-_{2,n}+V^-_n\bigr)	-\psiV_1 V_{1,n},\\
 			\Reproeffn	&=   \frac{\beta_n}{\gamma^-+\alpha_n+\eta^-+\mu_n} \frac{S_n}{N}.
 		\end{split}
	\end{align}
	The first relation of \eqref{Infection_rate_t} equates reported new infections to the inflow into $I^+$ from $I^-$ minus the outflow to the hospital $H$ and $R^+_1$. The second relation equates reported vaccinations to the inflow into $V_1$ minus the outflow to $V_{2}$. Here, inflow comes from $I^-$, $R^-_1$, $S$, $R^-_2$, and $V^-$ (re-vaccination of waning-immunity individuals).

\subsection{Model Regularization}
Due to observation errors in the time series, the daily estimates of the three time-dependent parameters are prone to significant estimation errors and often fluctuate sharply over time. Therefore, it is necessary to smooth these estimates and filter out the noise in order to obtain more robust estimates. We will now consider one such regularization technique.

Without assuming a specific underlying model for the observation errors, we decided to apply a finite impulse response (FIR) filter for temporal smoothing of the parameter estimates. Denote by $\widetilde{\beta}_{n}$, $\widetilde{\alpha}_{n}$, $\widetilde{\mu}_{n}$ the smoothed transmission rate, test rate and vaccination rate. From the FIR filters, they are obtained as follows:
\begin{align}
	\widetilde{\beta}_{n}=~\sum\limits_{j=1}^{J_{1}}a_{j}\widehat{\beta}_{(n-j)} + a_0\label{constr_1}, \quad
	\widetilde{\alpha}_{n}=~\sum\limits_{l=1}^{J_{2}}b_{l}\widehat{\alpha}_{(n-l)} + b_0,
	\quad
	\widetilde{\mu}_{n}=~\sum\limits_{m=1}^{J_{3}}c_{m}\widehat{\mu}_{(n-m)} + c_0,
\end{align}
where $J_1, J_2,J_3\in\N$ are the orders, and $a_j,b_l,c_m$, are the coefficients of the three FIR filters, where $J_1, J_2, J_3 < N_{t}-1$, $j=1,\ldots,J_1, l=1,\ldots,J_2$, and $m=1,\ldots,J_3$, and $a_0,b_0,c_0$ are intercepts.

To estimate the FIR filter coefficients, various machine learning techniques can be employed, such as ordinary least squares (OLS) or ridge regression \citep{chen2020time,ferrari2021modeling}. In this work, we adopt ridge regression, which solves the following optimization problems:
\begin{align}
	\min_{a_j} \; \sum_{n=J_1}^{N_{t}-1} \left( \widehat{\beta}_{n} - \widetilde{\beta}_{n} \right)^{2}
	+ { \regupar_{\beta}} \sum_{{j=1}}^{J_1} a_{j}^{2}, \label{Obj_1}~~~
	\min_{b_l} \; \sum_{n=J_2}^{N_{t}-1} \left( \widehat{\alpha}_{n} - \widetilde{\alpha}_{n} \right)^{2}
	+ { \regupar_{\alpha}} \sum_{{ l=1}}^{J_2} b_{l}^{2}, ~~~
	\min_{c_m} \; \sum_{n=J_3}^{N_{t}-1} \left( \widehat{\mu}_{n} - \widetilde{\mu}_{n} \right)^{2}
	+ { \regupar_{\mu}} \sum_{{ m=1}}^{J_3} c_{m}^{2}.
\end{align}
The intercepts are not penalized, as in the \texttt{glmnet} implementation used below. Here, $\regupar_{\beta}$, $\regupar_{\alpha}$, and $\regupar_{\mu}$ are the regularization parameters for the regressions for the transmission, tests, and vaccination rate, respectively.
The use of ridge regression is motivated by the limitations of OLS estimation. Specifically, OLS can yield high-variance estimates in the presence of multicollinearity among predictors or when the sample size is small relative to the number of predictors. Ridge regression, also known as $\mathcal{L}_2$-regularization, addresses this issue by penalizing the magnitude of the coefficients. The penalty parameters $\regupar_{\beta},\regupar_{\alpha},\regupar_{\mu}$ shrink the estimated coefficients, thereby reducing variance and improving model generalization.

\begin{algorithm}[htb]
	\caption{Parameter estimation and FIR smoothing of the time-dependent parameters} \label{param_estim}
	\DontPrintSemicolon

	\KwData{ $\Delta I^{+}_{n}$, $\Delta V_{1,n}$,  $\Reproeffn$, ~~$n=0,\ldots,N_{t}-1$}
	\KwIn{Initial value $X^{}_{0}$, assumed parameters $\eta^{\pm}$,$\gamma^{\pm},\gamma^{H}, \gamma^{C}$,$\delta$,$\kappa$,$\rho^{V}, \rho^{-}_{1},\rho^{-}_{2},\overline{c}_1,\overline{c}_2$}
	\KwOut {Raw estimates $\widehat{\alpha}_n,\widehat{\beta}_n,\widehat{\mu}_n$ and smoothed estimates $\widetilde{\alpha}_n,\widetilde{\beta}_n,\widetilde{\mu}_n$, ~~$n=0,\ldots,N_{t}-1$}

	\For{$n:=0$ to $N_{t}-1$}
	{	 Use $X_n$ with $I^+_n,V_{1,n}$ updated from reported data. Propagate $\E[X_{n+1}\mid X_n]$ using \eqref{Covid-19_dynamic-mean}.\\
			Solve \eqref{Infection_rate_t} to update $\widehat{\alpha}_{n},\widehat{\beta}_{n},\widehat{\mu}_{n}$\\[1ex]
			$\displaystyle \hspace*{2em}
			\widehat{\alpha}_{n}  =  \overline{c}_1\dfrac{\Delta I^{+}_n + (\eta^{+} + \gamma^{+})I^{+}_{n}}{I^{-}_n},\quad
			\widehat{\mu}_{n}=\overline{c}_2\dfrac{\Delta V_{1,n}+\psiV_1V_{1,n}}{I^{-}_n+R^{-}_{1,n}+S_n+R^{-}_{2,n}+V^-_n},\quad
			\widehat{\beta}_n	= \Reproeffn\frac{N}{S_n} (\gamma^{-} + \widehat{\alpha}_n + \eta^{-} + \widehat{\mu}_{n})
		$
	}
	{
		Solve the three ridge-regression problems in \eqref{Obj_1} using the  raw estimates  $\widehat{\beta}_n$, $\widehat{\alpha}_n$ and $\widehat{\mu}_n$ \\to obtain the FIR coefficients $a_j$, $b_l$, and $c_m$.\\
		Compute $\widetilde{\beta}_n$, $\widetilde{\alpha}_n$, and $\widetilde{\mu}_n$ from the FIR relations in \eqref{constr_1}.
	}
\end{algorithm}

In this study, the lag orders $J_1$, $J_2$, and $J_3$ are set to $7$, reflecting the inherent weekly periodicity observed in \covid data. Daily reporting practices introduce systematic fluctuations: for example, testing and reporting volumes are often reduced on weekends and elevated during weekdays. A $7$-day period thus captures this recurring cycle.

Given historical data  $(\Delta I^{+}_{n}, \Reproeffn, \Delta V_{1,n})$ for $n=0,\ldots,N_{t}-1$, we first estimate $\widehat{\beta}_n$, $\widehat{\alpha}_n$, and $\widehat{\mu}_n$ as explained in the previous subsection. Subsequently, we solve the ridge regression problems \eqref{Obj_1}, subject to the constraints in \eqref{constr_1}, to obtain the FIR filter coefficients \(a_j, b_l,c_m\).
These estimated coefficients are then used to construct smoothed estimates for the infection, test, and vaccination rates.
The complete calibration method is summarized in Algorithm \ref{param_estim}.

\begin{remark} \label{rem:tuning}
	In Algorithm \ref{param_estim}, the values $\overline{c}_1$ and $\overline{c}_2$ serve as tuning parameters (or hyper-parameters). They scale the raw rate estimates to ensure that the estimated test and vaccination rates remain within biologically plausible bounds, specifically that they do not exceed the maximum admissible outflow from the $I^-$ compartment. The use of manually specified hyperparameters and parameter constraints to enforce admissible ranges is also encountered in infectious-disease model calibration; see, e.g., \citet{refisch2022data}. They are not directly learned from the data during the training process. Rather, they need to be set manually to improve the algorithm's effectiveness.
\end{remark}

\subsection{Calibration Results, Application to the \covid Model}\label{Ridge_Regression}

\begin{figure}[pos=!htbp]
	\centering
	\subfloat[ Daily new cases in Germany, March 2020--March 2023. Source: RKI  \citep{anderheiden2023nowcasting}.]{\includegraphics[width=0.48\textwidth]{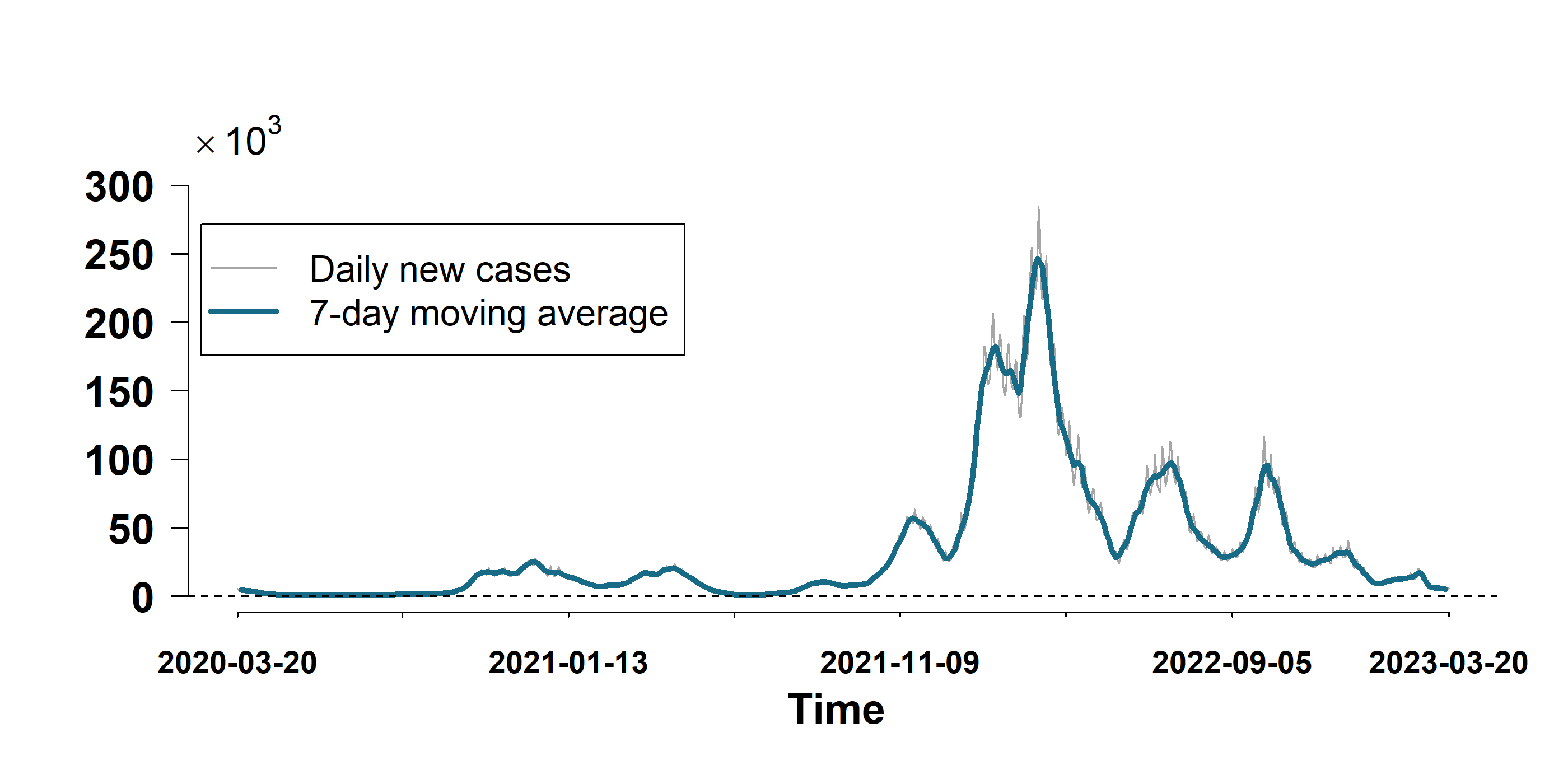}\label{fig:f16}}
	\hfill
	\subfloat[Daily new vaccinations in Germany (all doses). Source:  RKI \citep{rki2023impfungen}.]{\includegraphics[width=0.48\textwidth]{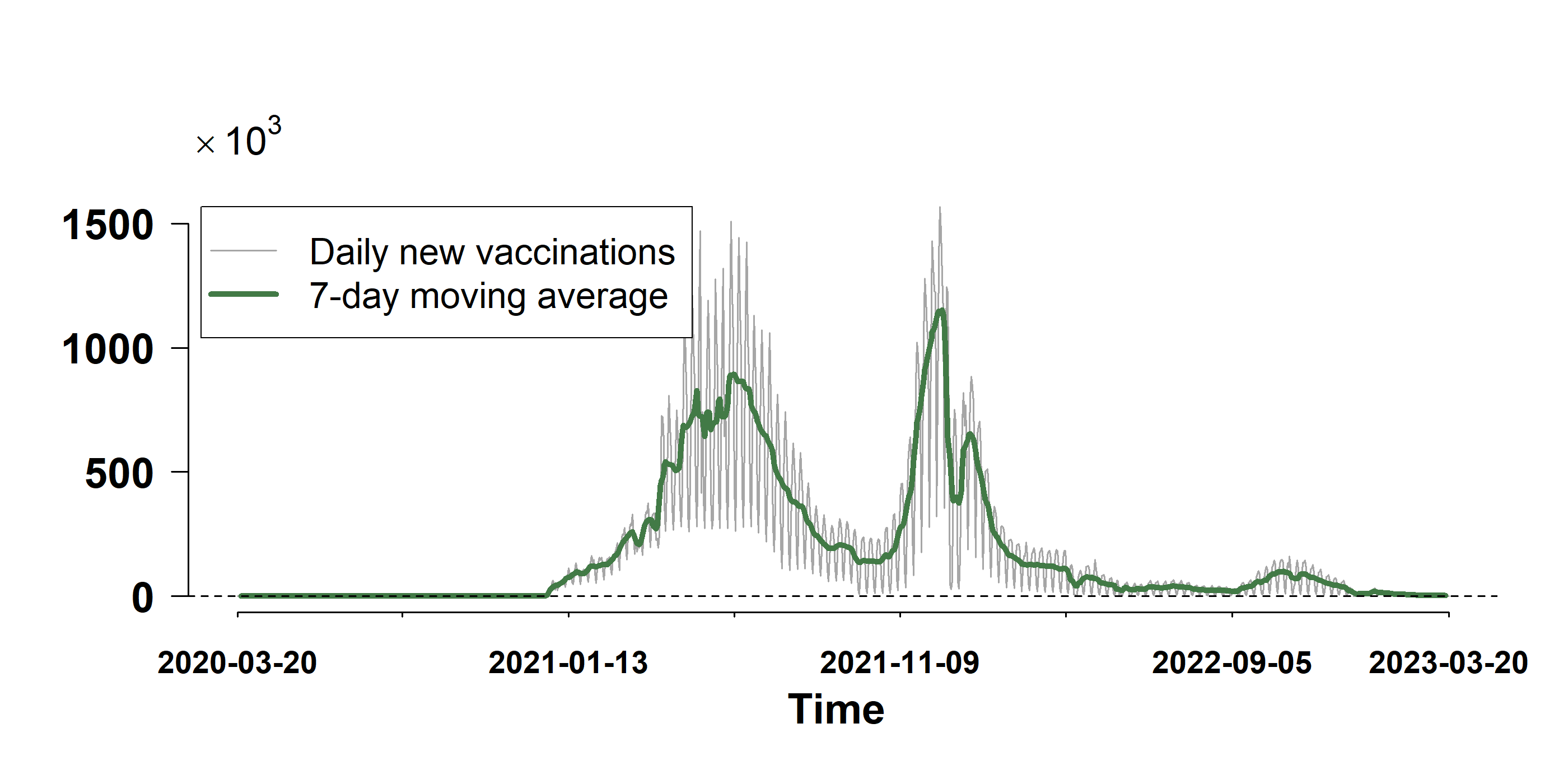}\label{fig:f17}}\\
	\subfloat[RKI seven-day estimates of the effective reproduction number $\Reproeff$ for Germany, March 2020--March 2023. Source: RKI  \citep{anderheiden2023nowcasting}. The horizontal dashed line marks the threshold level $1$.]{\includegraphics[width=0.48\textwidth]{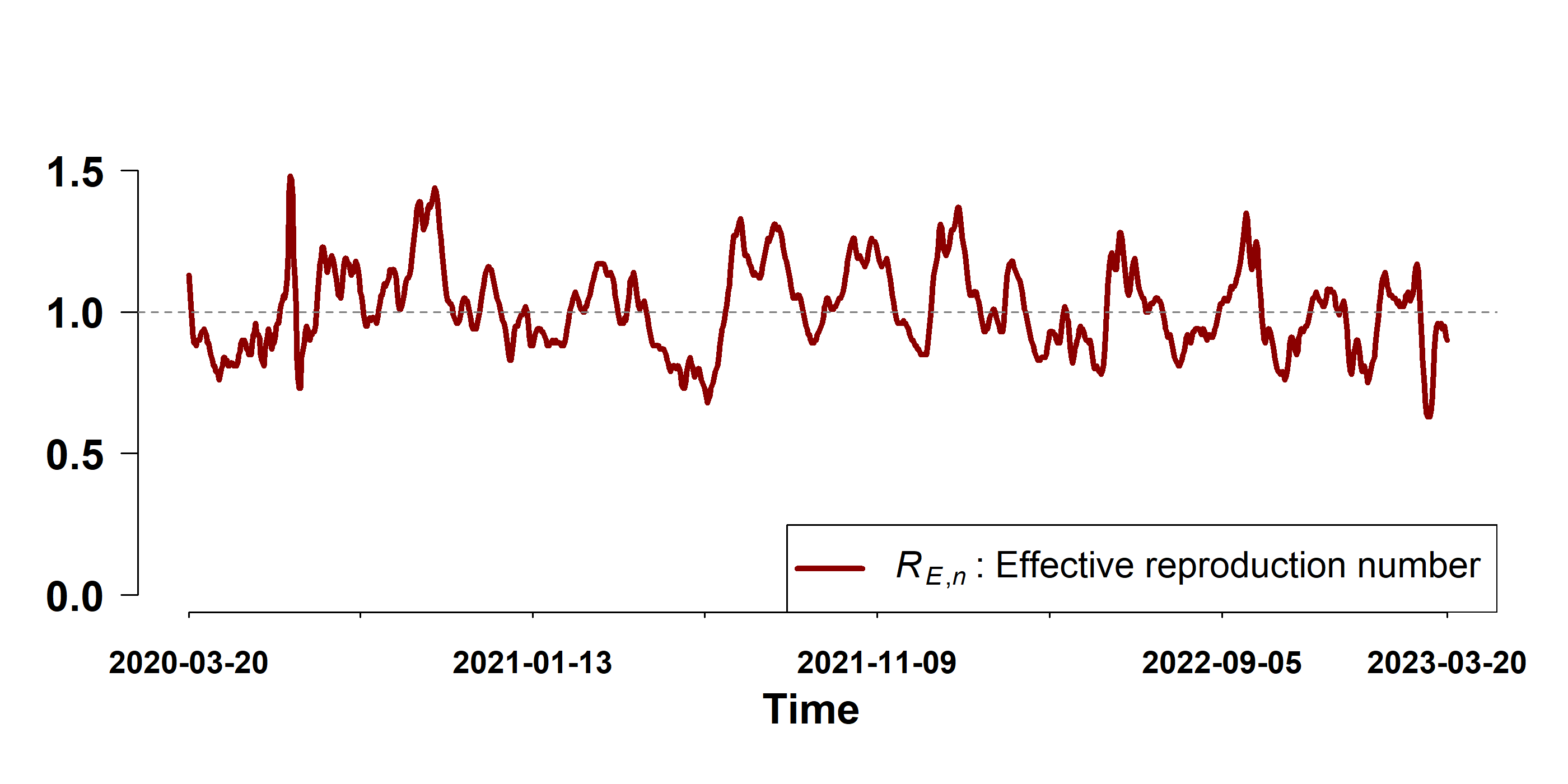}\label{fig:f_Rt}}
	\caption{Input data used for calibration}
\label{fig:calibration_input_data}
\end{figure}

We now present numerical results of the above described regularization method for parameter estimation of the time-dependent transmission, test, and vaccination rates. The study is based on the RKI dataset introduced in Subsection~\ref{sec:CalibrationDataSet} that provides daily reported \covid data in Germany  from 20 March 2020 through 20 March 2023.

Figure~\ref{fig:calibration_input_data} summarizes the  three observed time series used in the calibration. In Figure~\ref{fig:f16}, the grey line represents daily reported new cases in Germany, while the cyan line shows the corresponding 7-day moving average. These daily new cases are interpreted as new arrivals of detected infections in $I^+$ and provide information about the time-dependent test rate. Figure~\ref{fig:f17} reports daily new vaccinations forming the inflow of compartment $V_1$. Note that in Germany vaccinations started in December 2020, and that in our model we aggregate individuals who have received at least one \covid vaccine dose, without distinguishing between first and later doses.  Figure~\ref{fig:f_Rt} shows the RKI seven-day estimates of the effective reproduction number $\Reproeff$ used to calibrate the transmission rate. The reference line at $\Reproeff=1$ marks the boundary between epidemic growth and decline.

The values of the assumed model parameters are given in Table~\ref{Table_Param_Covid19-1}. For the initial values of the 10 observed state components in $Z_0$  the initial values are set to $I^{+}_{0}=5328$ for confirmed infections and $R^{+}_{1,0}=455$ for confirmed recoveries, based on the RKI	nowcast from March 2020.  For the period of full immunity after recovery or vaccination we choose  $\LR\Delta t=\LV\Delta t=90$ days.  This choice is supported by \citet{townsend2022duration}, who study waning protection after infection and vaccination, and by \citet{polack2020safety}, who report high short-term vaccine efficacy against symptomatic \covidk. Our approach is a simplification and an idealization. We are aware that these immunity periods do not last exactly 90 days for every individual. We divide these immunity  periods into $\KR=\KV=3$ subperiods, each lasting 30 days. The observable cascade compartments for confirmed recovered are initialized with $R^+_{2,0}=R^+_{3,0}=0$. Since no vaccine was available in March 2020, the cascade compartments for vaccinated individuals are also initialized as empty, i.e., $V_{1,0}=V_{2,0}=V_{3,0}=0$.

The states related to the hospital are initialized with $C_0=323$ \covid patients in intensive care reported by the DIVI intensive care register on 20 March 2020 \citep{divi2020intensivregister}, $D_0=31$ deaths reported by the RKI \citep{rki2020lagebericht0320,pz2020kein}, and we assume that $H_0=C_0=323$, since no national data on hospitalized non-ICU patients are available for this period. For the initial estimates of the five hidden states, we assume that the number of unconfirmed infections and recoveries is five times higher than that of confirmed cases. Then we obtain  $I^{-}_{0}= 26640, R^{-}_{1,0}=2275$.  The hidden terminal cascade compartments are initially empty, i.e., $R^{-}_{2,0}=V^{-}_{0}=0$. Finally, the initial number of susceptibles $S_0$ is obtained from the normalization property assuming a constant total population size $N= 8.3\times10^7$ (population of Germany).

\ifthenelse{\preprintOn=0}{}{{\prepcolor
\begin{figure}[pos=!htbp]
	\centering
	\subfloat[{Raw estimates of transmission rate $\widehat{\beta}_n$}]{\includegraphics[width=0.31\textwidth]{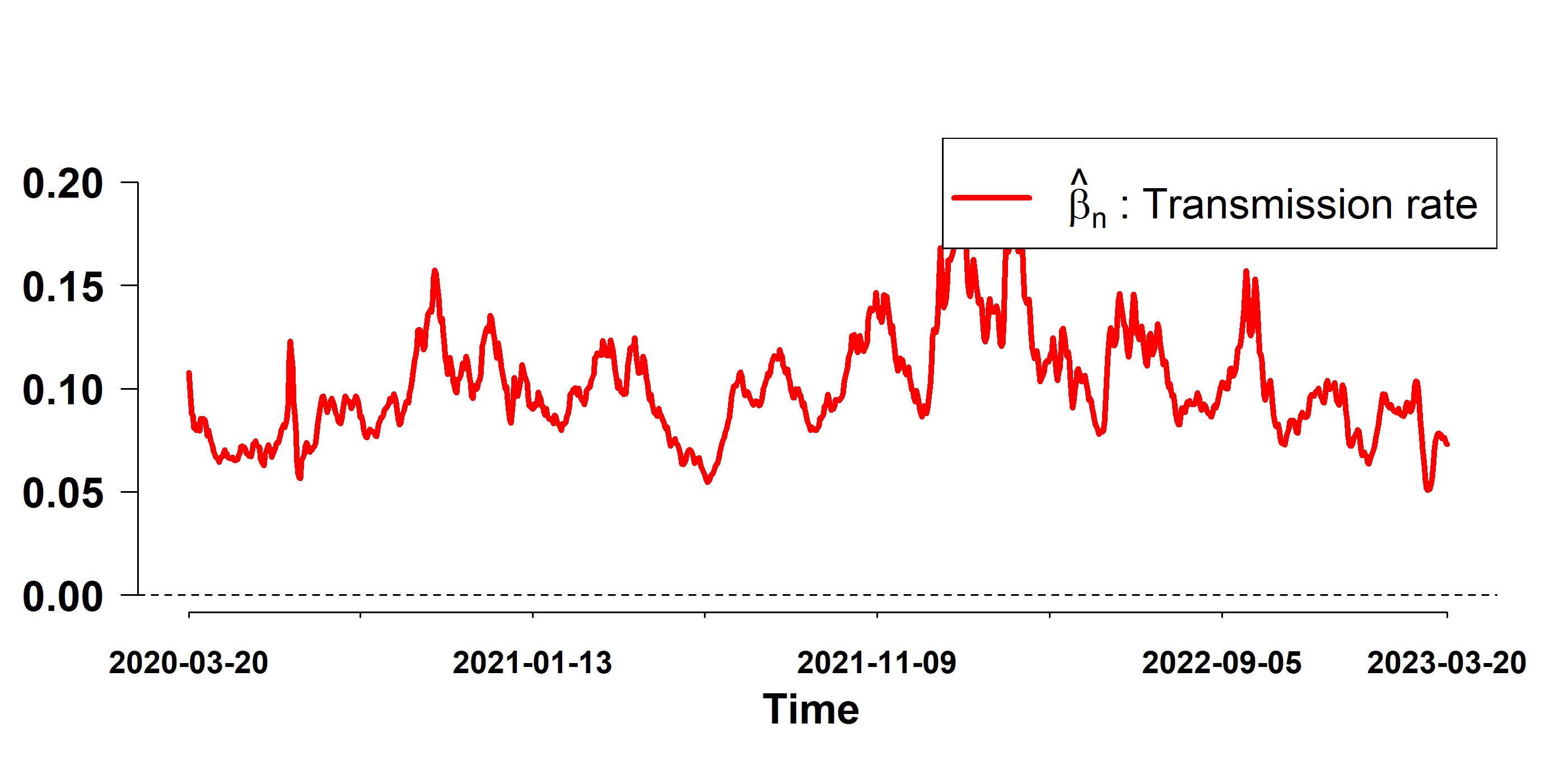}\label{fig:f18}}
	\hfill
	\subfloat[{Raw estimates of test rate $\widehat{\alpha}_n$}]{\includegraphics[width=0.31\textwidth]{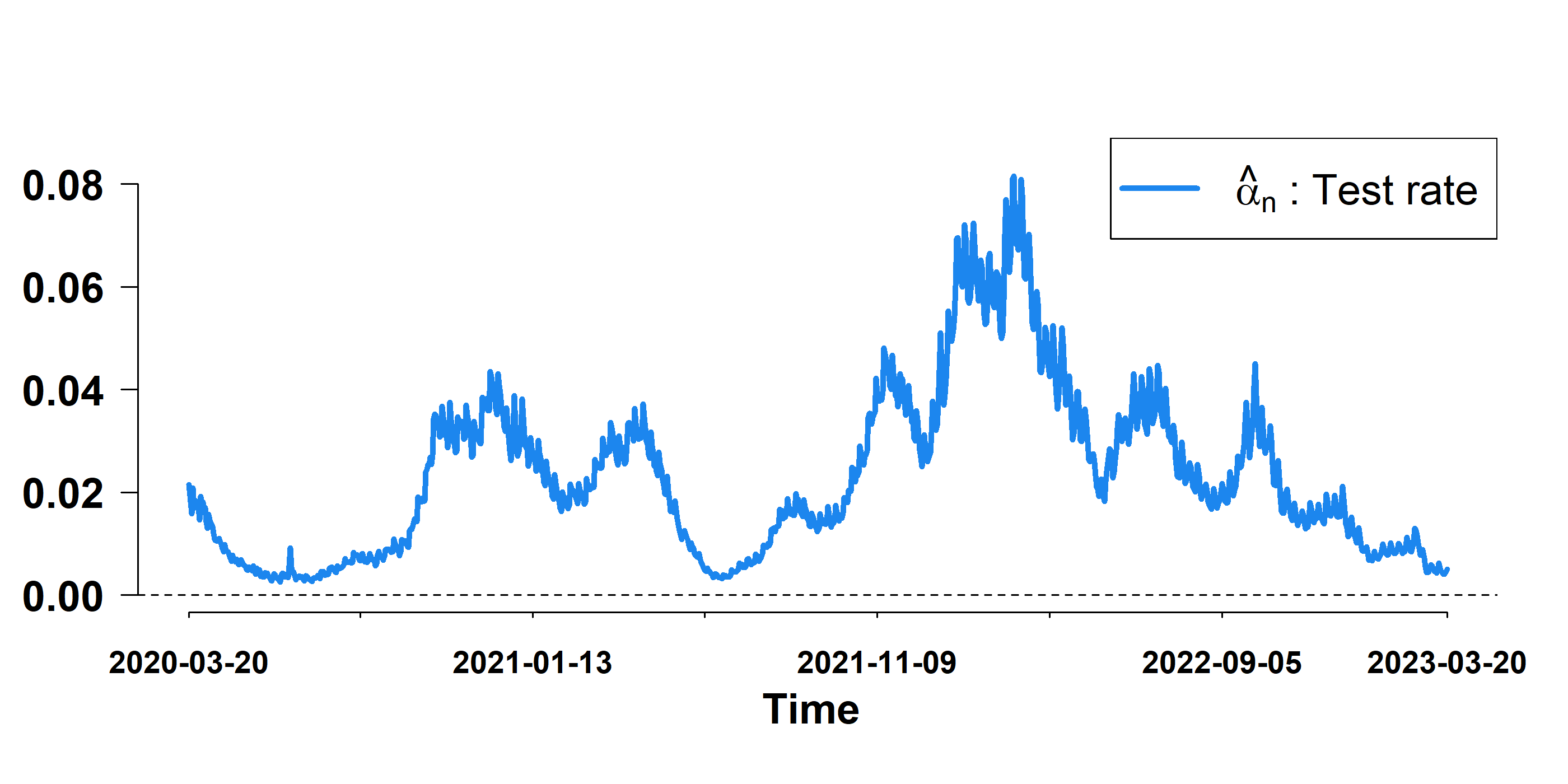}\label{fig:f19}}
	\hfill
	\subfloat[{Raw estimates of vaccination rate $\widehat{\mu}_n$}]{\includegraphics[width=0.31\textwidth]{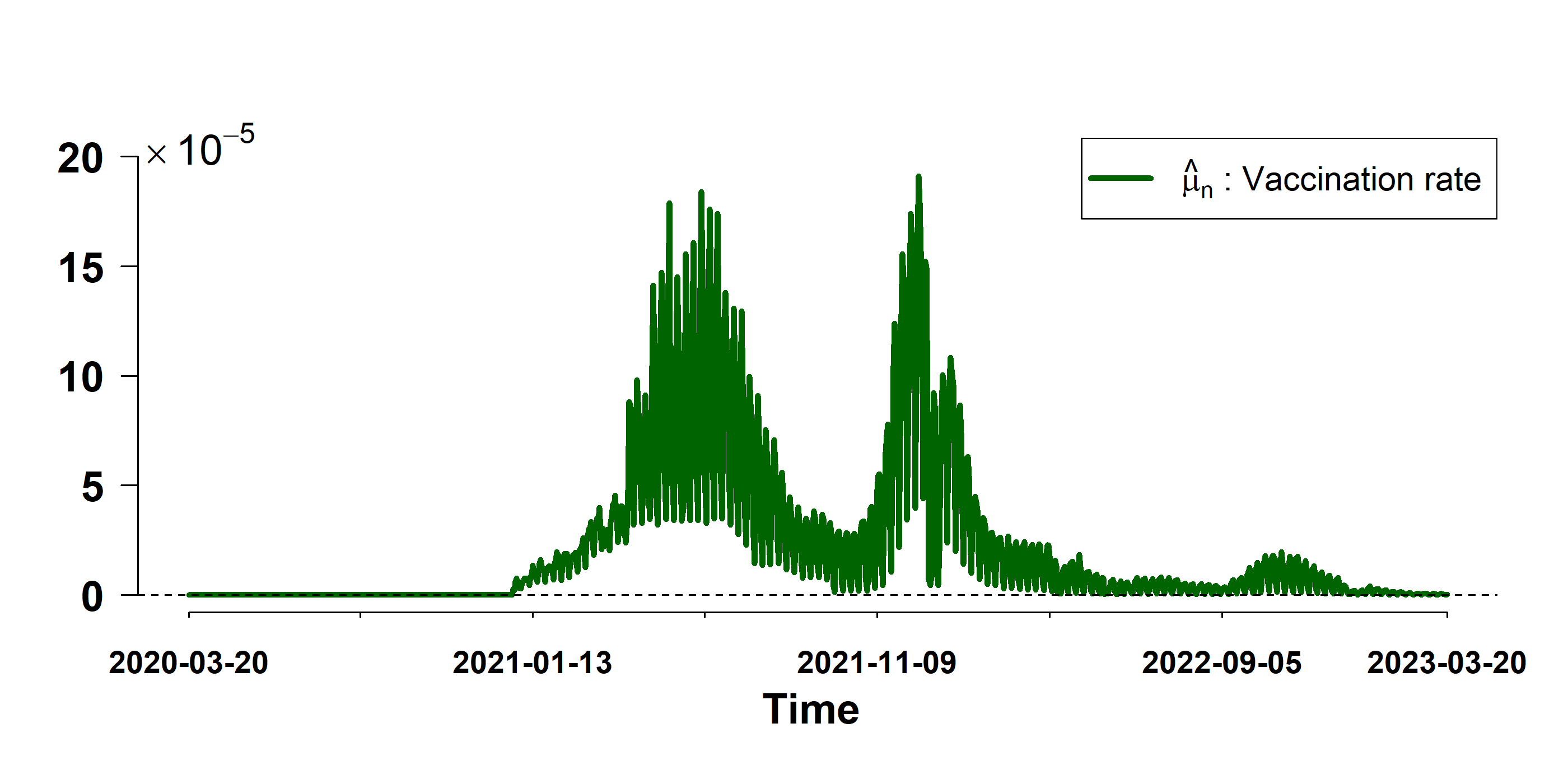}\label{fig:f20}}

	\caption{Raw estimates of the transmission rate $\widehat{\beta}_n$, test rate $\widehat{\alpha}_n$, and vaccination rate $\widehat{\mu}_n$ obtained by solving \eqref{Infection_rate_t} at each day.}
	\label{cha6_MSE}
\end{figure}
}}

After this initialization, we apply Algorithm~\ref{param_estim} to compute the  raw estimates $\widehat{\alpha}_{n}$, $\widehat{\beta}_{n}$, and $\widehat{\mu}_{n}$. 
\ifthenelse{\preprintOn=0}{}{{\prepcolor The results are  shown in Figure~\ref{cha6_MSE}.}}
We then smooth the  raw estimates by ridge regression, for which the minimization problems given in \eqref{Obj_1} are solved numerically using the \texttt{glmnet} package in R, which provides ridge-regularized estimates for the transmission, test, and vaccination rates. Here, we first split the time series of  raw estimates $(\widehat{\beta}_n,\widehat{\alpha}_{n},\widehat{\mu}_{n})$ into a training set containing $70\%$ of the values and a test set containing the remaining $30\%$. The corresponding regularization parameters $\regupar_{\beta}$, $\regupar_{\alpha}$, and $\regupar_{\mu}$ are selected by cross-validation.  The resulting FIR coefficients $a_j,b_l,c_m$ from \eqref{constr_1} are reported for $J_1=J_2=J_3=7$ in Table \ref{Table_FIR_coefficients}. 
\ifthenelse{\preprintOn=0}{ In \citet[Section 6.4]{ouabo2026partial} we provide additional figures that show the raw estimates and  illustrate the cross-validation procedure.}{}

\ifthenelse{\preprintOn=0}{}{{\prepcolor
\begin{figure}[pos=!htbp]
	\centering
	\subfloat[{CV curve for $\ln(\regupar_{\beta})$}]{\includegraphics[width=0.31\textwidth]{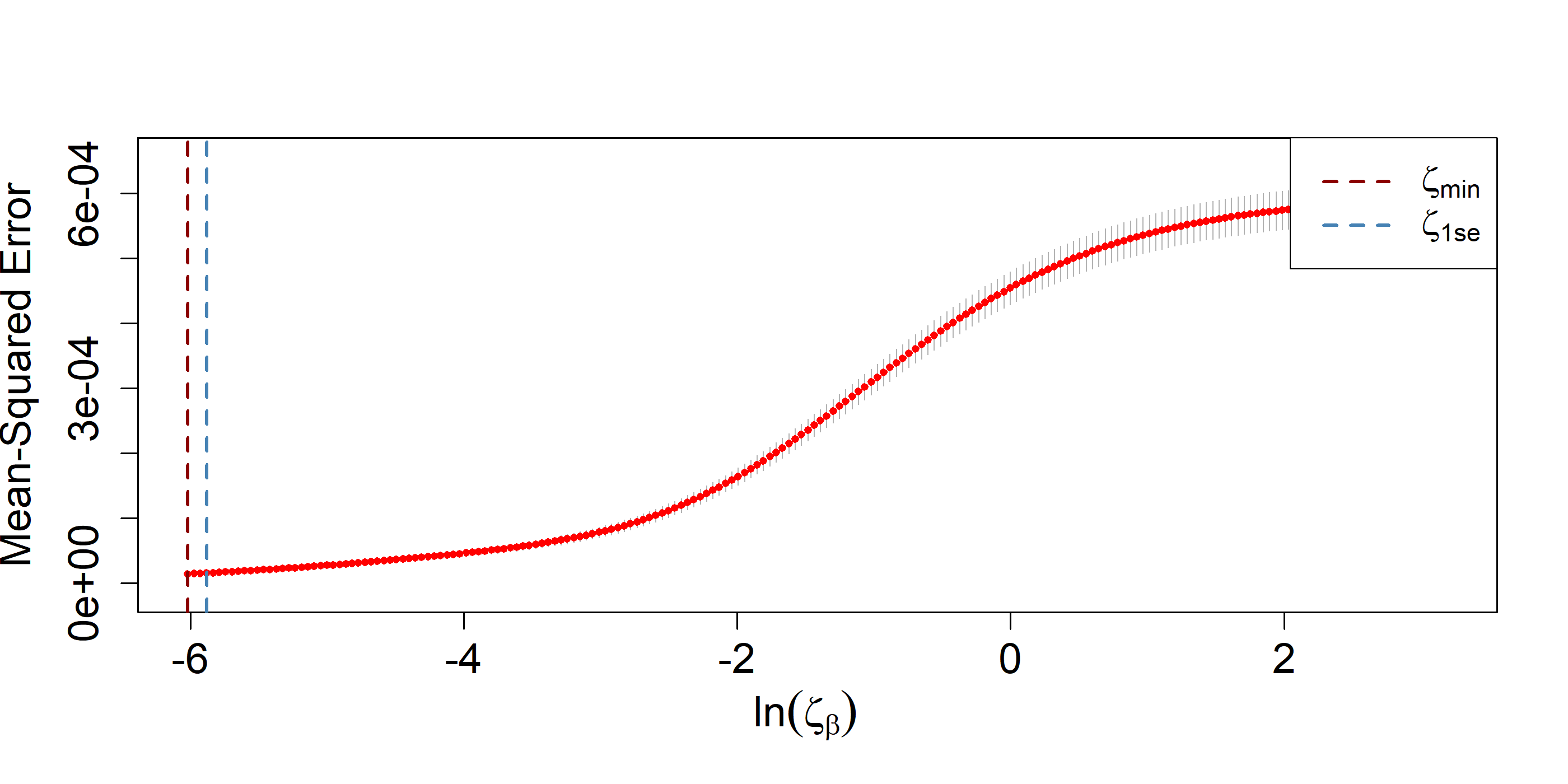}\label{fig:f21}}
	\hfill
	\subfloat[{CV curve for $\ln(\regupar_{\alpha})$}]{\includegraphics[width=0.31\textwidth]{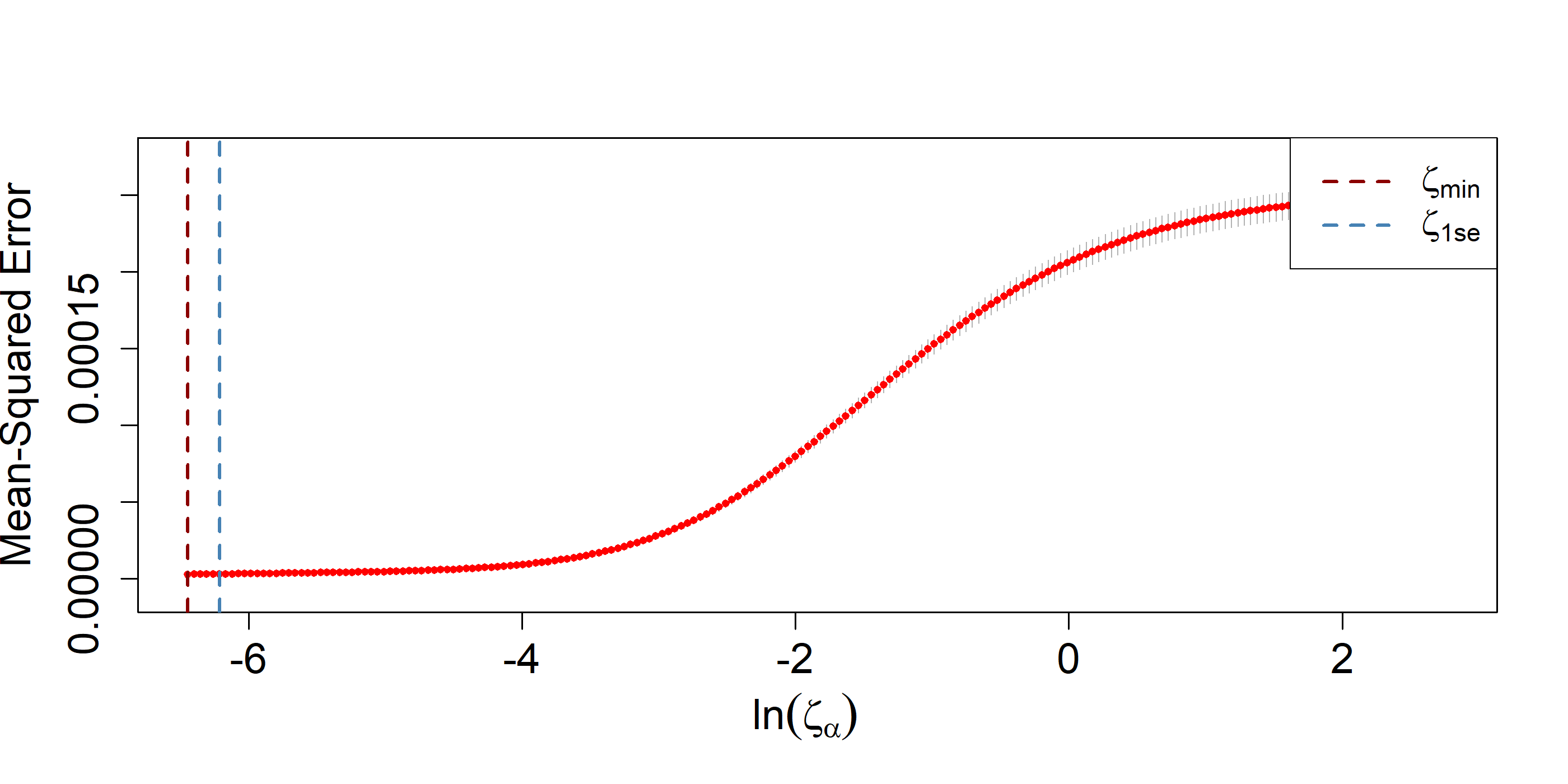}\label{fig:f24}}
	\hfill
	\subfloat[{CV curve for $\ln(\regupar_{\mu})$}]{\includegraphics[width=0.31\textwidth]{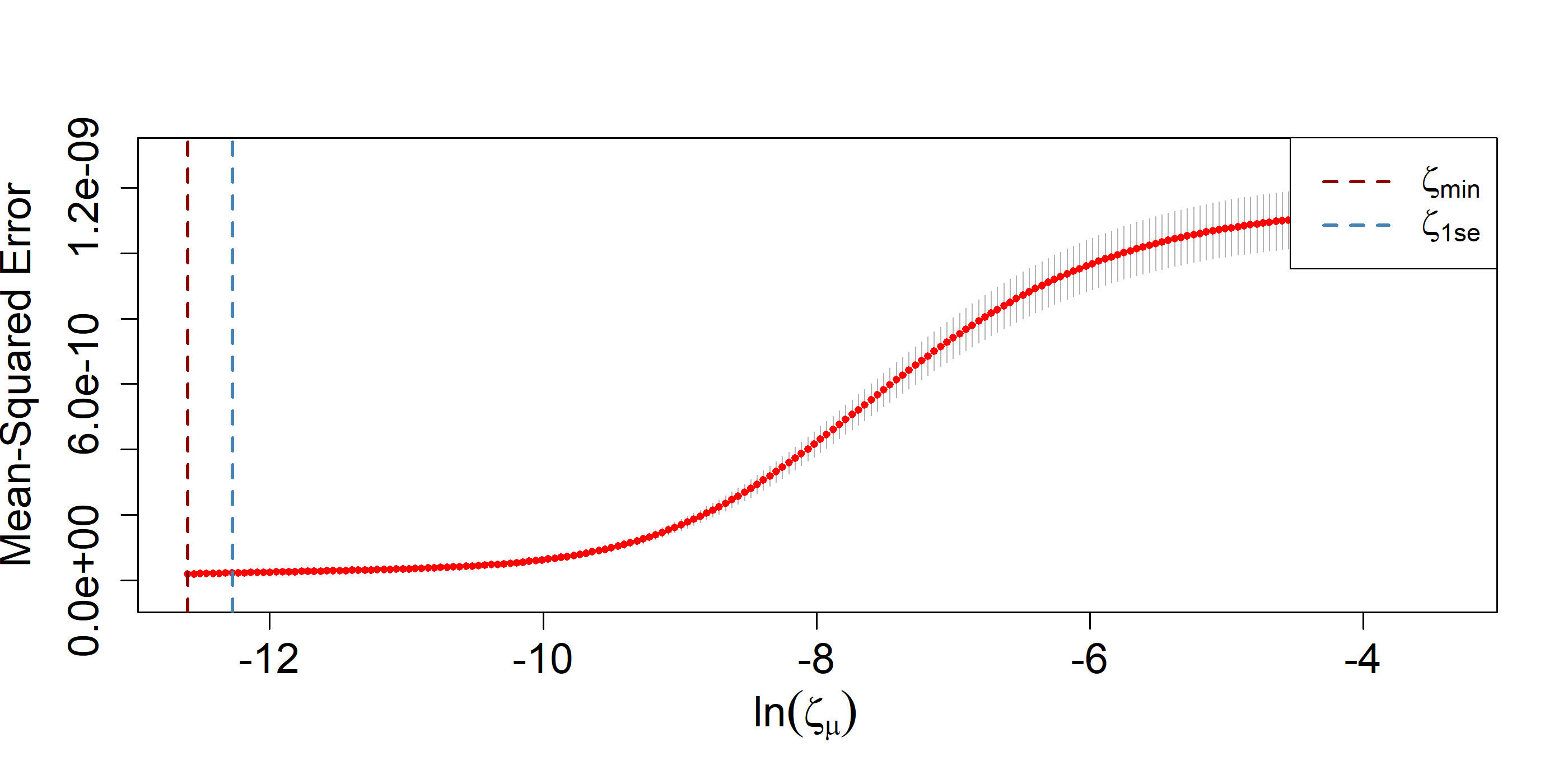}\label{fig:f25}}

	\caption{Cross-validation(CV) curves for the ridge penalty $\regupar$ of each time-dependent rate.  The grey error bars show the mean $\pm\,1$ standard error of the  ten-fold CV MSE. The  dark red dashed vertical line marks $\regupar_{\min}$ (the penalty minimizing cross-validated MSE) and the  blue dashed vertical line marks $\regupar_{1\text{se}}$ (the largest penalty within one standard error of the minimum). We use $\regupar_{\min}$ for prediction.}
	\label{fig:crossvalidation}
\end{figure}

The cross-validation is illustrated in Figure~\ref{fig:crossvalidation}. For each of the three parameters, the regression uses all $J_1=J_2=J_3=7$ lagged variables as candidate predictors.
The graphs illustrate the cross-validation process used in ridge regression to assess model performance for different values of the relevant regularization parameter $\regupar_\star$, where $\star=\beta,\alpha,\mu$. For each value of $\regupar_\star$, the model is trained and validated across multiple  subsets of the data set  that are called folds. The performance is measured using the mean square error (MSE), defined in general as
$
\text{MSE} = \frac{1}{n} \sum_{j=1}^{n} \left( y_j - \hat{y}_j \right)^2.
$
Here, $y_j$ is the true output, $\hat y_j$ is the predicted value for validation sample $j$, and $n$ is the number of validation samples. The MSE values are averaged across  folds for each regularization strength $\regupar_\star$. The x-axis shows $\ln(\regupar_\star)$, while the y-axis shows the corresponding MSE. A lower MSE indicates better predictive fit.
 The  red dots represent the cross-validated MSE, and the  grey error bars represent its standard error. The  dark red dashed vertical line marks $\regupar_{\min}$, which minimizes cross-validated MSE, and the  blue dashed vertical line marks $\regupar_{1\text{se}}$, the largest penalty whose error is within one standard error of the minimum. For each parameter series we use $\regupar_{\min}$ and evaluate the selected model on the test data. 
}}
 

\begin{table}[htb]\centering
\begin{tabular}{lrrrrrrrr}\toprule
Rate & Intercept & 1 & 2 & 3 & 4 & 5 & 6 & 7\\\midrule
{$\beta$} & { $4.3\!\times\!10^{-3}$} & { $.434$} & { $.306$} & { $.188$} & { $.091$} & { $.019$} & { $-.028$} & { $-.052$}\\
{ $\alpha$} & { $3.9\!\times\!10^{-4}$} & { $.243$} & { $.204$} & { $.167$} & { $.132$} & { $.102$} & { $.077$} & { $.056$}\\
{ $\mu$} & { $6.1\!\times\!10^{-7}$} & { $.288$} & { $.226$} & { $.171$} & { $.126$} & { $.085$} & { $.052$} & { $.027$}\\\bottomrule
\end{tabular}
\caption{Intercepts and lag coefficients of the fitted FIR filters, obtained with the penalty $\regupar_{\min}$.}\label{Table_FIR_coefficients}
\end{table}

\begin{figure}[pos=!htbp]
	\centering
	\subfloat[{{Transmission rate $\beta$ (full period)}}]{\includegraphics[width=0.47\textwidth]{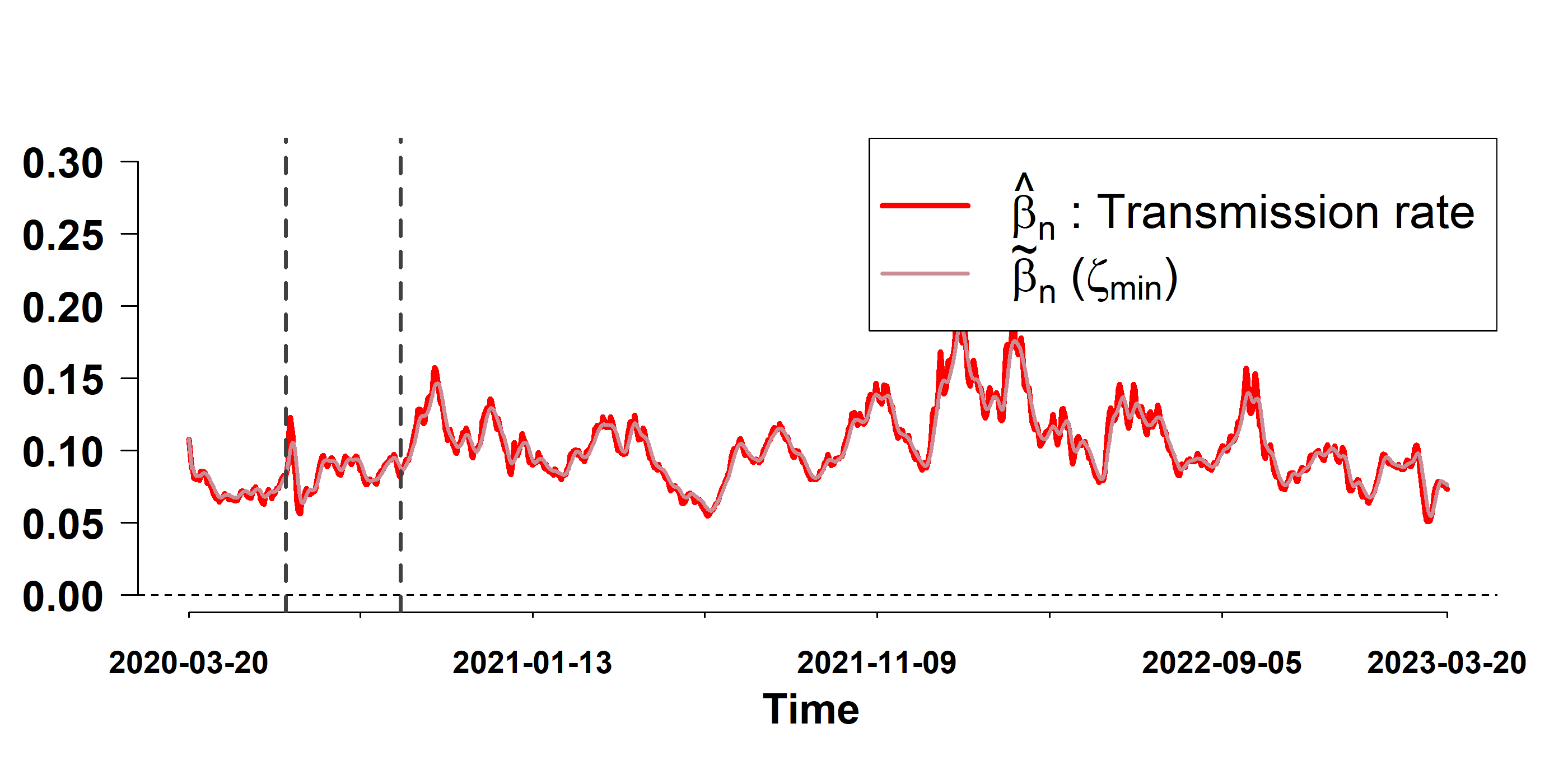}\label{fig:f22}}
	\hfill
	\subfloat[{{Transmission rate $\beta$ ({ Jun 2020-Sep 2020})}}]{\includegraphics[width=0.47\textwidth]{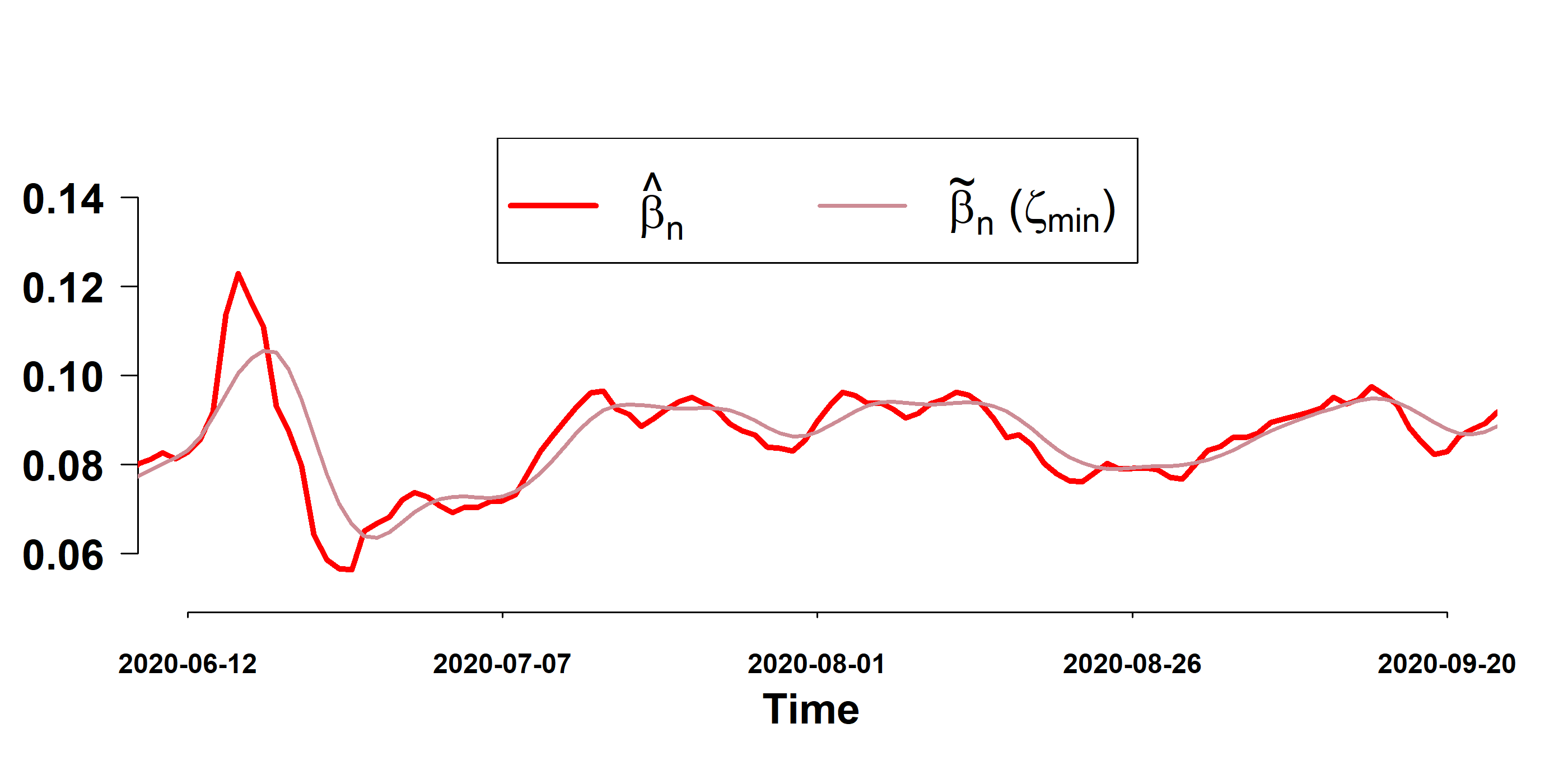}\label{fig:f23}}
	\\[-2.5ex]
	\subfloat[{{Test rate $\alpha$ (full period)}}]{\includegraphics[width=0.47\textwidth]{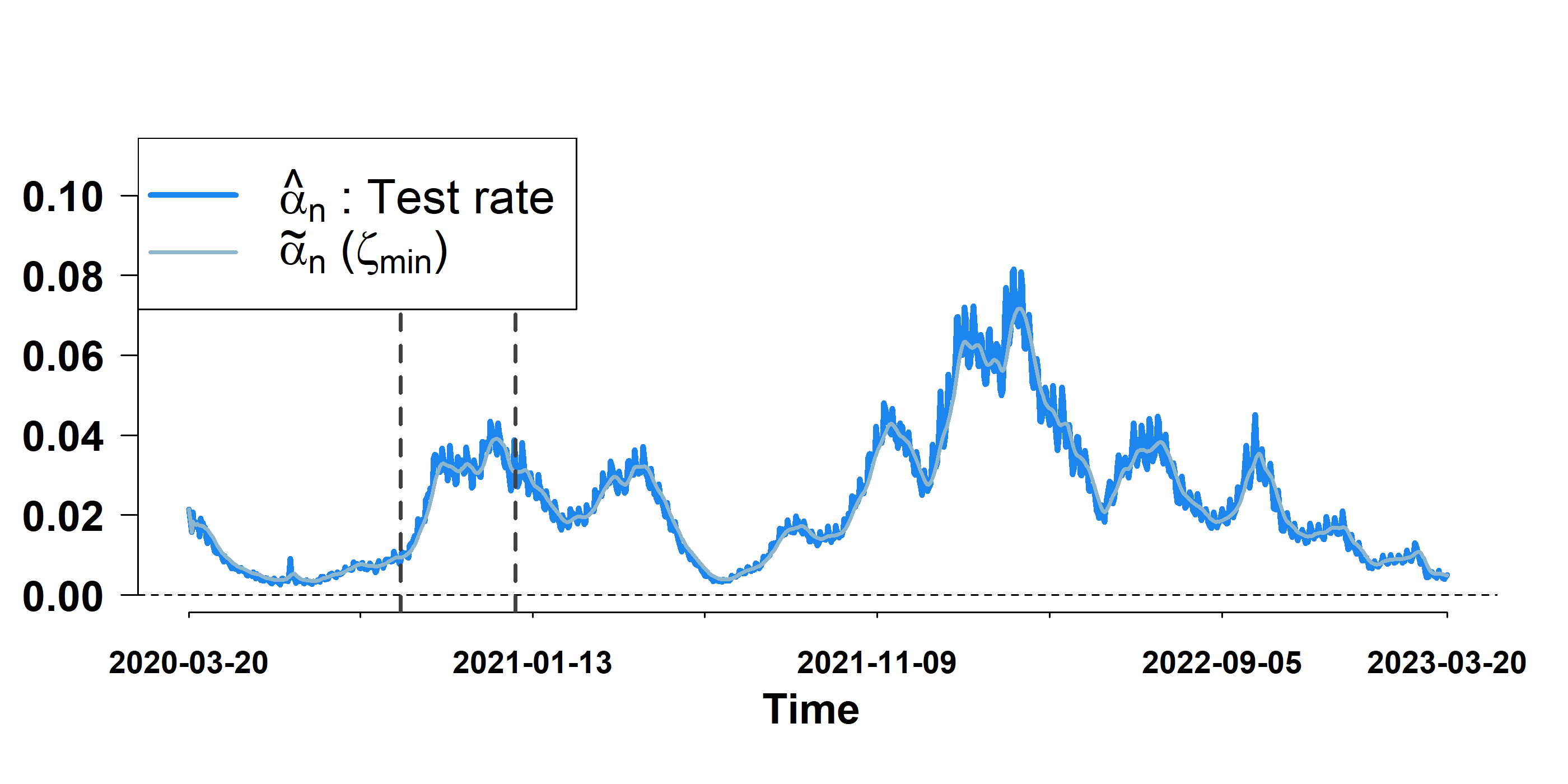}\label{fig:f26}}
	\hfill
	\subfloat[{{Test rate  $\alpha$ ({ Sep 2020-Dec 2020})}}]{\includegraphics[width=0.47\textwidth]{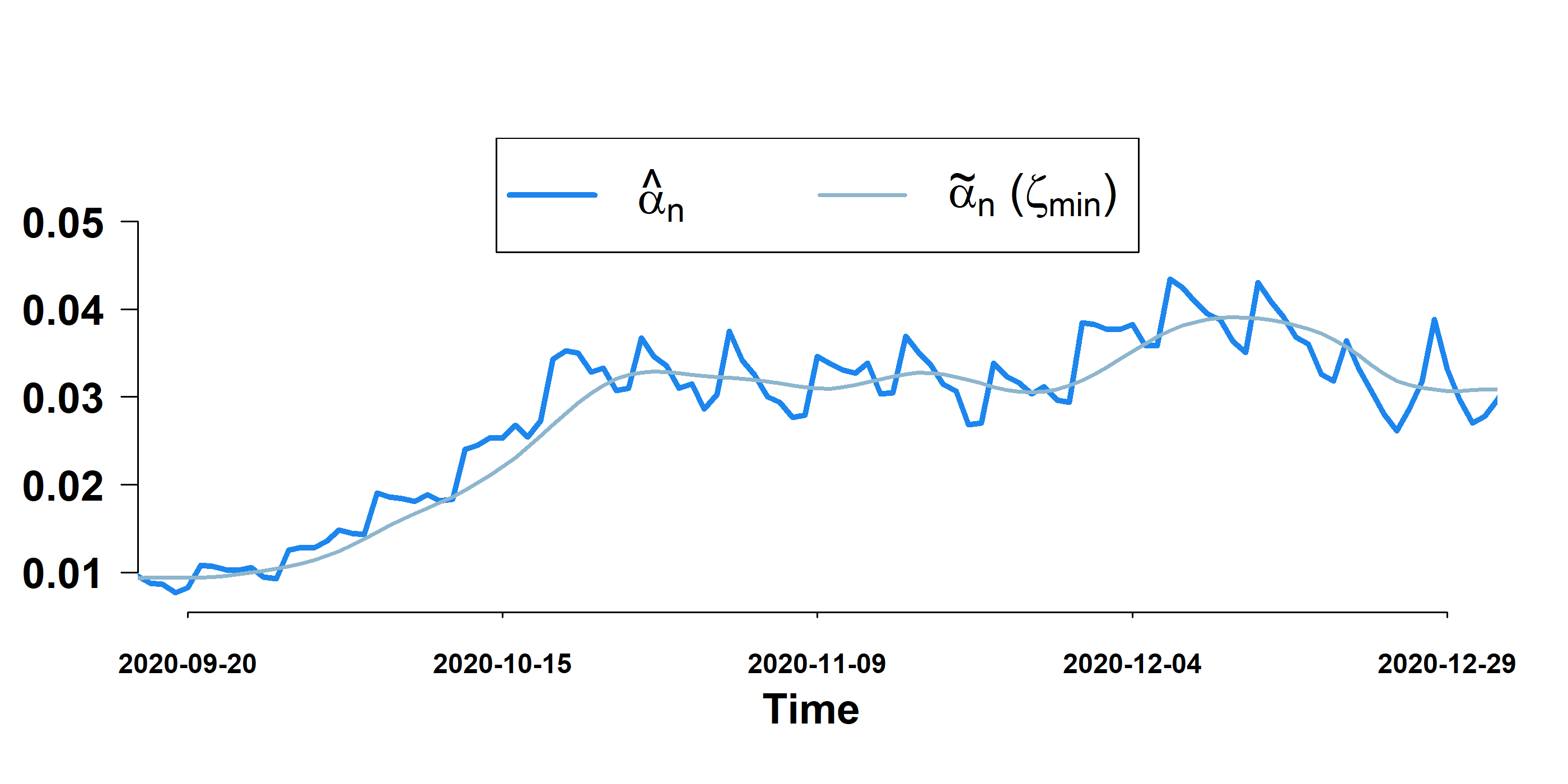}\label{fig:f27}}
	\\[-2.5ex]
	\subfloat[{{Vaccination rate $\mu$ (full period)}}]{\includegraphics[width=0.47\textwidth]{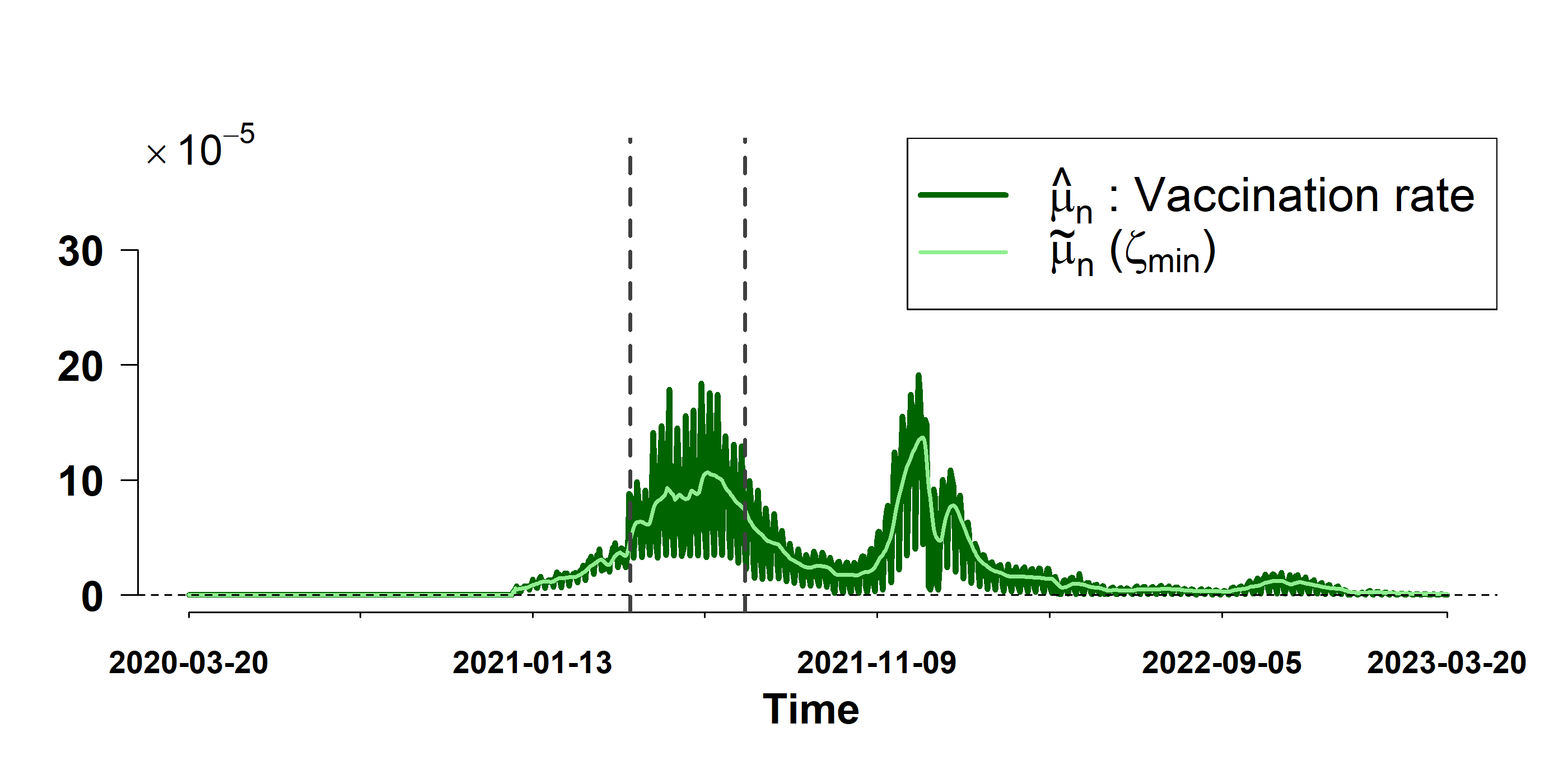}\label{fig:f28}}
	\hfill
	\subfloat[{{ Vaccination rate $\mu$ ({ Apr 2021-Jul 2021})}}]{\includegraphics[width=0.47\textwidth]{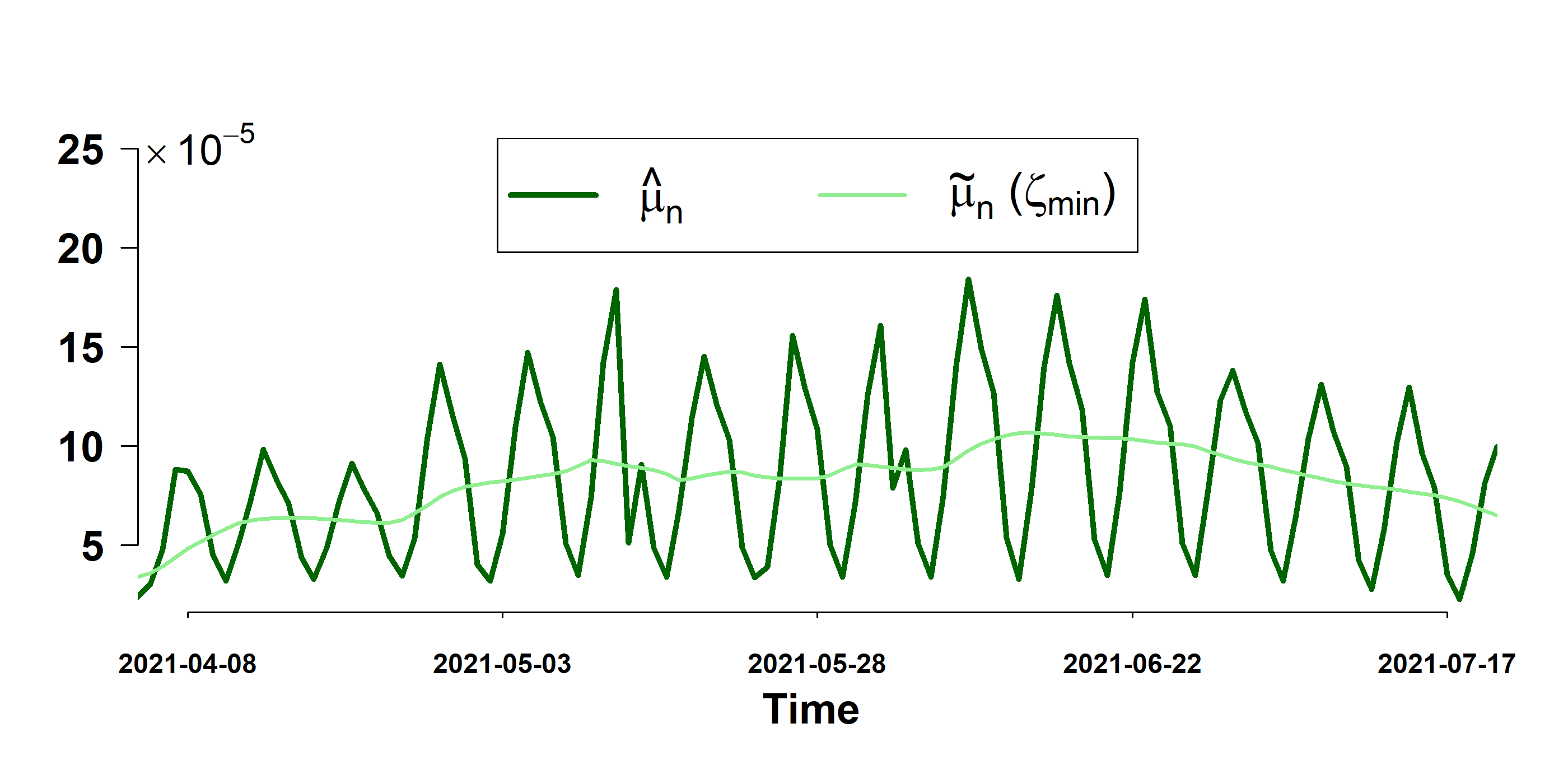}\label{fig:f30}}
	\captionof{figure}{ Raw and smoothed estimates of the transmission rate ${\beta}_n$ (a-b), test rate ${\alpha}_n$ (c-d), and vaccination rate ${\mu}_n$ (e-f). Each full-period panel (left) is paired with a panel that zooms  in  a three-month period  (right), whose boundaries are marked by the two dashed vertical lines in the left panel. The smoothed estimates are obtained with the penalty $\regupar_{\min}$ selected by cross-validation.}
	\label{chap6_cros_val_mu}
\end{figure}

At this stage, we apply the ridge-regression regularization to the  raw estimates $\widehat{\beta}_{n}$, $\widehat{\alpha}_{n}$, $\widehat{\mu}_{n}$ to filter out noise, and obtain the smoothed estimates $\widetilde{\beta}_{n}$, $\widetilde{\alpha}_{n}$, $\widetilde{\mu}_{n}$. 
Figures~\ref{fig:f22}-\ref{fig:f23} show the  raw and smoothed estimates of the transmission rate $\widetilde{\beta}_n$  for the full period and for a three-month period.
Similar results are obtained for the test rate $\widetilde{\alpha}_{n}$ and the vaccination rate $\widetilde{\mu}_{n}$, and are shown in the other subfigures of Figure~\ref{chap6_cros_val_mu}.


{\let\FloatBarrier\relax
\section{Numerical Simulation}\label{sec:NumericalResults}}
We will now use the extended \covid model, for which we calibrated the time-dependent transition rates $\beta,\alpha,\mu$ using epidemiological data from Germany in the previous section, to simulate paths of the state process. The remaining model parameters are constant and provided in Table~\ref{Table_Param_Covid19-1}. The initial conditions used for the simulations are given in Subsection~\ref{Ridge_Regression}. The simulations cover the same time period for which calibration data were available  that is,  from 20 March 2020 through 20 March 2023. This allows us to compare the simulated course of the epidemic with its actual course.

\begin{figure}[pos=!htbp]
	\centering
	\subfloat[{Hidden infected $I^-$ and confirmed infected $I^+$}]{%
		\includegraphics[width=0.45\textwidth]{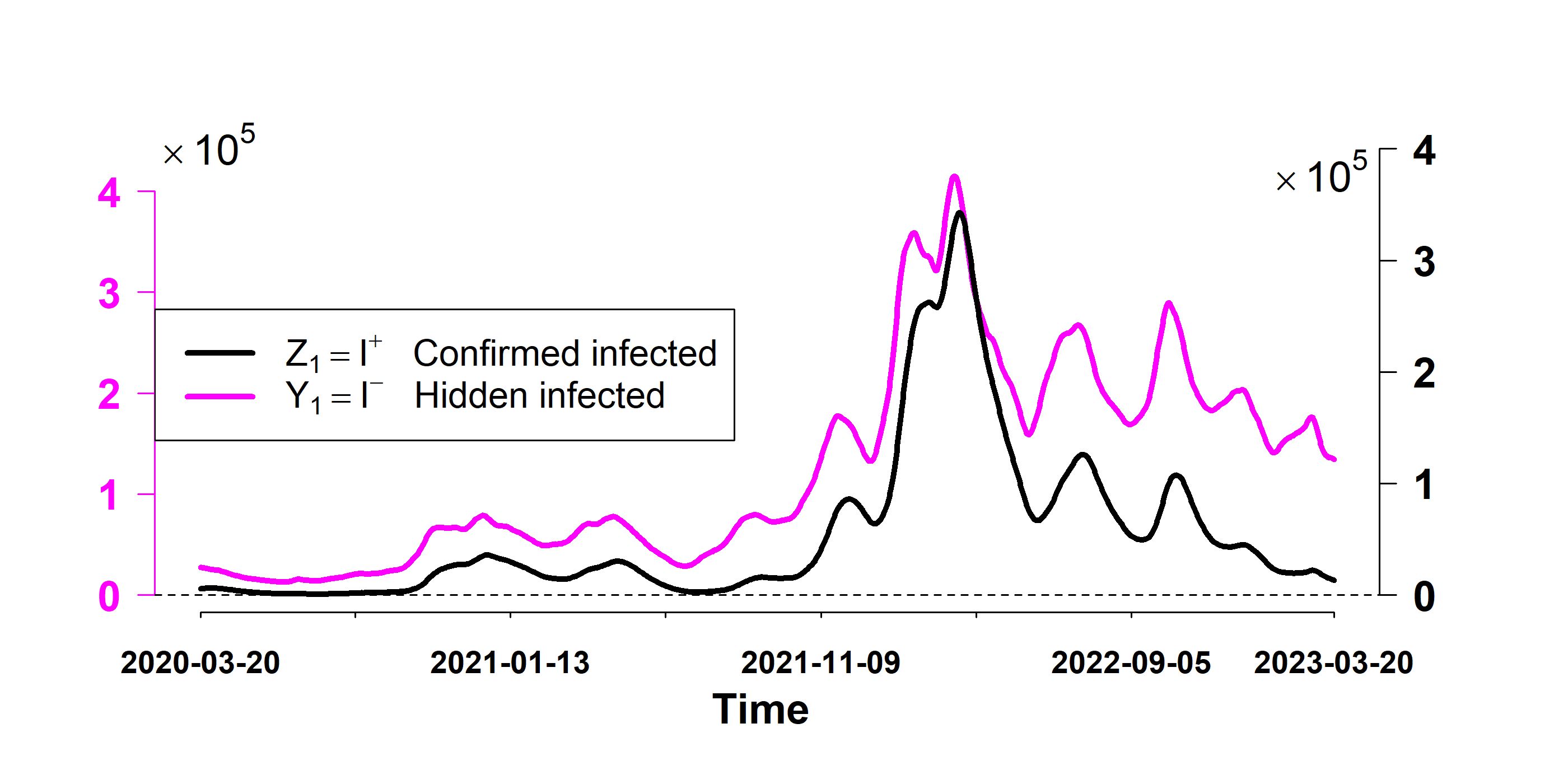}%
		\label{fig:f31}}
	\hfill
	\subfloat[Total number of infected individuals]{%
		\includegraphics[width=0.45\textwidth]{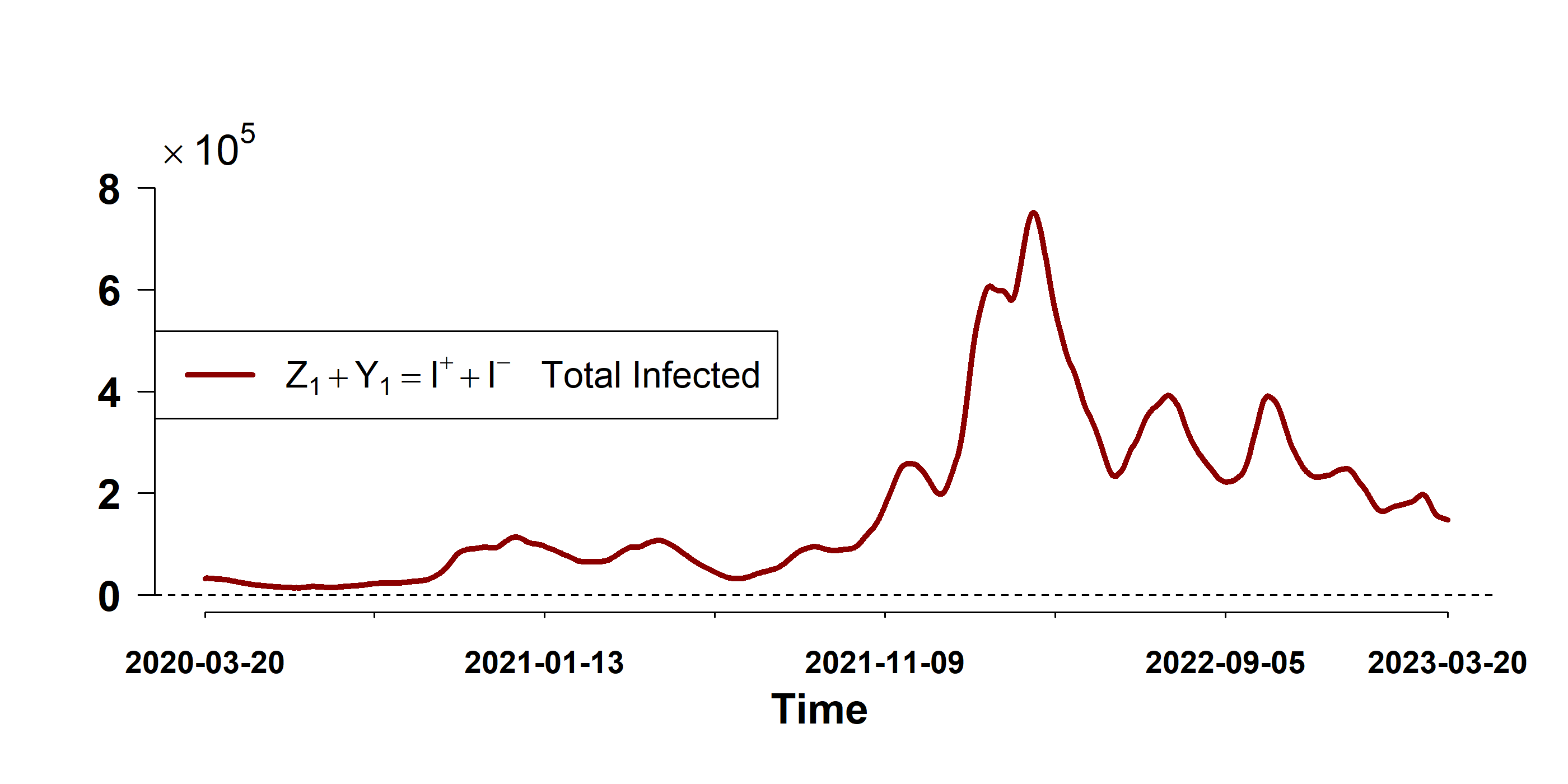}%
		\label{fig:f32}}
	\\[-2.5ex]
	\subfloat[Hospitalized, ICU and dead individuals]{%
		\includegraphics[width=0.45\textwidth]{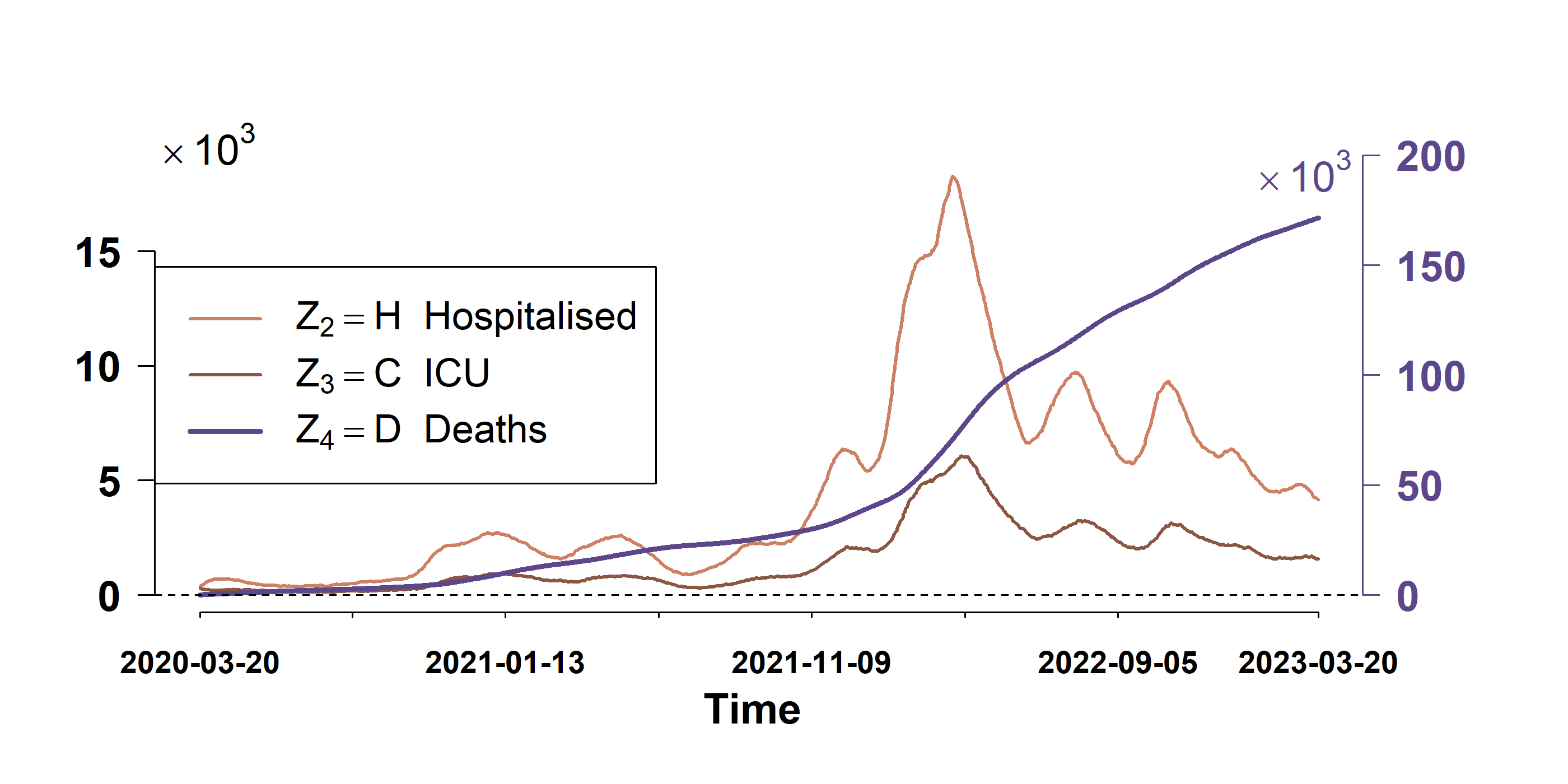}%
		\label{fig:f33}}
	\hfill
	\subfloat[{Hidden recovered individuals $R^-_1$}]{%
		\includegraphics[width=0.45\textwidth]{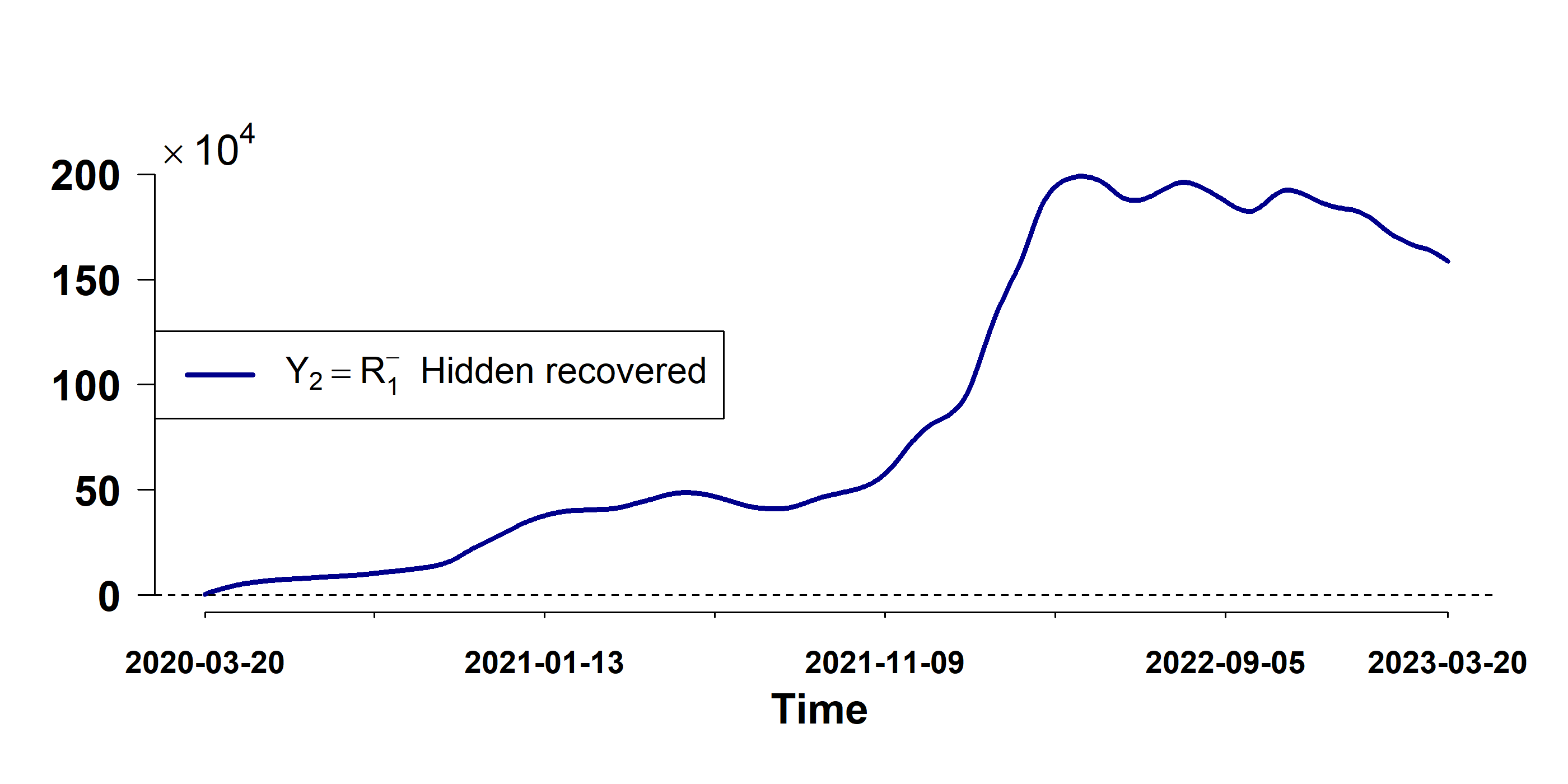}%
		\label{fig:f36}}
	\\[-2.5ex]
	\subfloat[{Cascade states for vaccinated individuals ($V_1, V_2, V_3$)}]{%
		\includegraphics[width=0.45\textwidth]{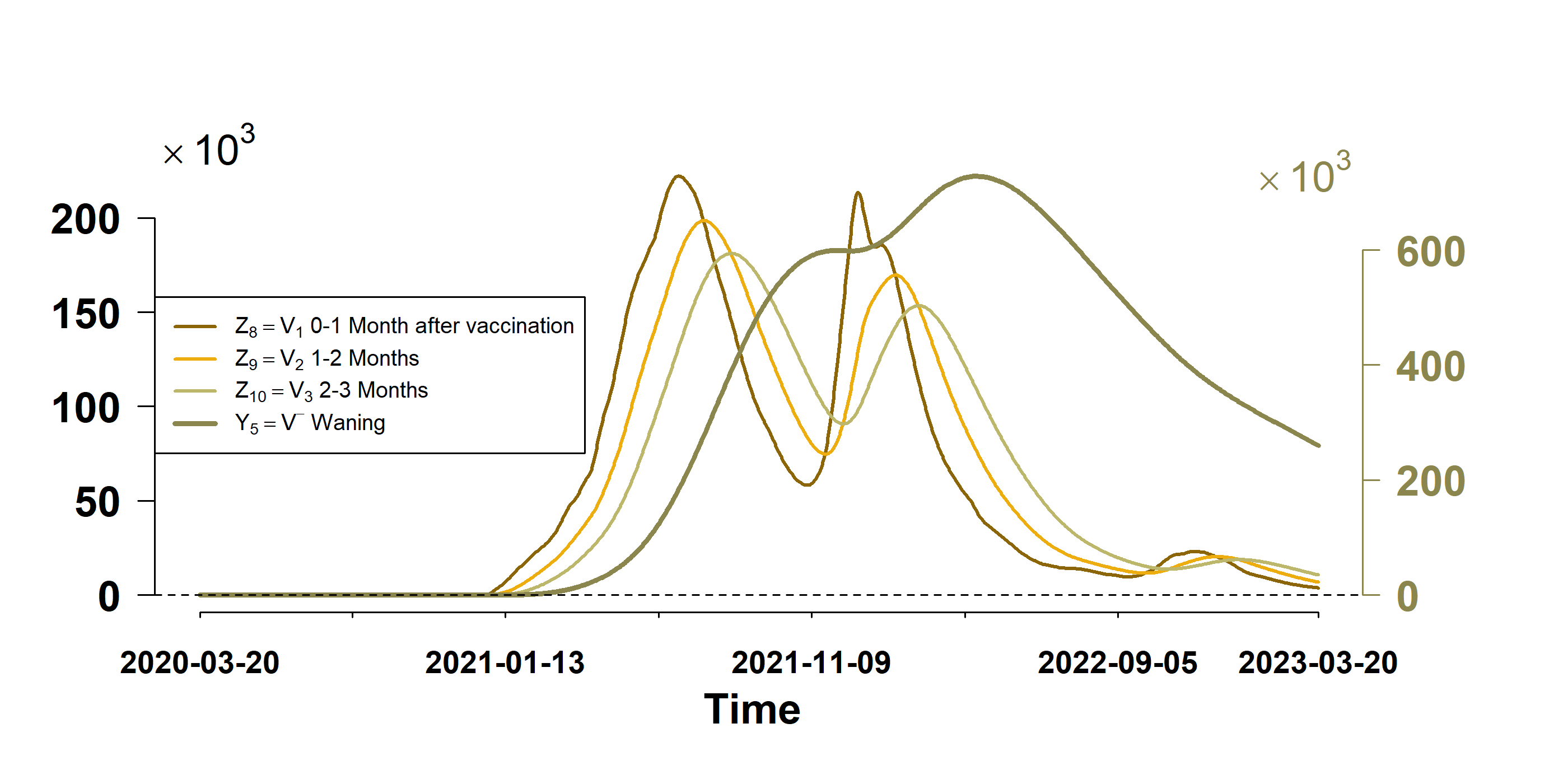}%
		\label{fig:f35}}
	\hfill
	\subfloat[{Cascade states for recovered individuals ($R^+_1, R^+_2, R^+_3$)}]{%
		\includegraphics[width=0.45\textwidth]{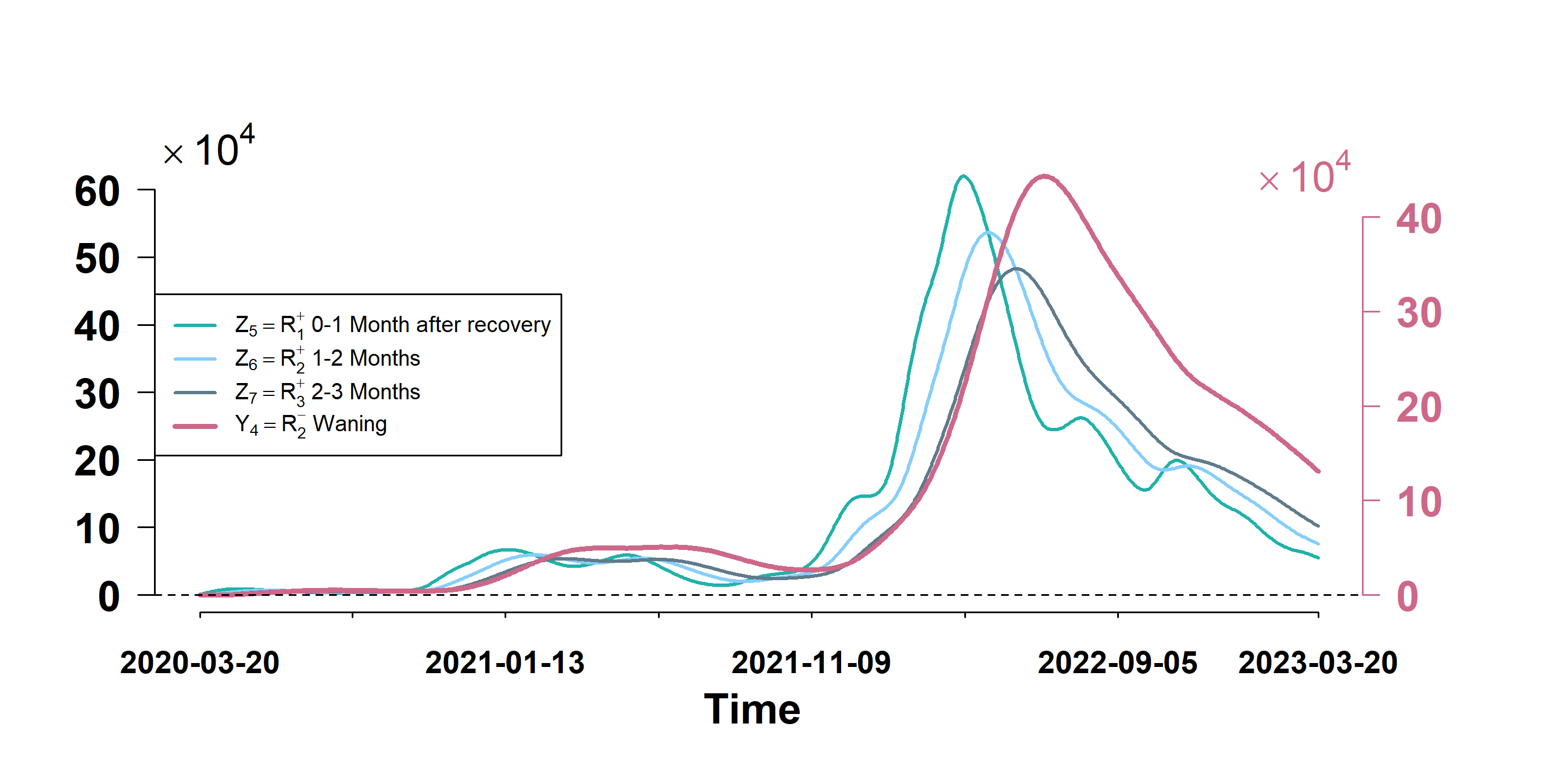}%
		\label{fig:f35-2}}
	
	\caption{Dynamics of the state process: (a) hidden infected $I^-$ versus confirmed infected $I^+$; (b) total infected; (c) hospitalized, ICU, and cumulative model deaths; (d) hidden recovered $R^-_1$; (e) cascade states for vaccinated individuals, labeled by time elapsed since vaccination; (f) cascade states for confirmed recovered individuals, labeled by time  elapsed since  recovery.}
	\label{fig:covid_dynamics}
\end{figure}

Figure~\ref{fig:f31} illustrates the time evolution of both hidden and confirmed infected individuals, and Figure~\ref{fig:f32} shows the total number of infections, calculated as the sum of hidden and confirmed cases.
The trend pattern shown largely corresponds to the confirmed cases of infection in Germany during the same period, as illustrated in Figure~\ref{fig:f16}. It should be noted that the model also simulates trends for the number of hidden cases. This information, which was not available during the pandemic, could provide valuable insights for decision-making.  The decline of hidden infections in spring 2022 follows the period of intensive testing in early 2022, shown  in Figure~\ref{fig:f26}, and the booster vaccination campaign of winter 2021/22, shown  in Figure~\ref{fig:f28}. Figure~\ref{fig:f31} also shows that from around May 2022 the number of hidden infections began to diverge from the curve of confirmed infections. This coincides with a marked drop of the test rate, while the vaccination rate had already returned to a low level. By the summer of the same year, the number of hidden cases began to rise noticeably, which may be due to the considerable reduction in both testing and vaccination efforts.

We also visualize additional metrics, such as the number of hospitalized and ICU patients as well as deaths, see Figure~\ref{fig:f33}. Recall  that  our \covid model takes into account only deaths among patients in ICUs,  which is a simplifying assumption, see Section \ref{Model_example}.
Now, we shift our attention to individuals who recovered from the disease without being identified or confirmed as infected. By monitoring hidden infections in Figure~\ref{fig:f31} and their subsequent recovery in Figure~\ref{fig:f36}, we can uncover the hidden dynamics of disease spread, which is often underestimated in real-world reports. This information is highly valuable for public health strategies, as it helps identify the extent of undetected cases that could silently continue propagating the infection within the population.

Finally, we present the dynamics of the cascade states, which allow us to track the progression of recovered or vaccinated individuals through various phases until they lose their full immunity and transition to the state of recovered or vaccinated individuals with waning immunity (see  Figures~\ref{fig:f35} and \ref{fig:f35-2} for vaccinated and recovered individuals, respectively).

\section{Conclusion}\label{sec:conclusion}
This study examines partially observable stochastic epidemic models that account for uncertainties in the course of epidemics. The uncertainties inherent in the various transitions between the states of the epidemic are first mathematically modeled using   continuous-time Markov chains,  and then approximated by diffusion processes. This leads to systems of stochastic differential equations  and, after temporal discretization, to stochastic recursions that determine the state dynamics.
Particular attention was paid to accounting for the uncertainties about the current state, and to the fact that not all state components can be directly observed, due to asymptomatic or unreported infections. 

The theoretical modeling approach was illustrated by various examples, with a particular focus on \covid models. Here, a cascade state approach is used to extract  all available information delivered by observed epidemiological data  related to partially hidden compartments.  The latter model is  calibrated to epidemiological data for Germany and used to conduct numerical simulations that are able to reproduce typical observed patterns during the \covid pandemic.

The derived epidemic models with partial information are used in our work on the statistical estimation of hidden states and ``dark figures'' from observable data using extended Kalman filter techniques, see   \cite{ouabo2025estimating,mbouandi2025convergence,achola2025stochastic,OuaboKamkumo2025}. Furthermore, they are incorporated into  \cite{mbouandi2025stochastic,MbouandiNjiasse2025}, where  the cost-optimal containment of epidemics  is investigated as optimal control problems under partial information.

\bigskip
{\ifthenelse{\preprintOn=0}{}{\footnotesize}
	\ifthenelse{\preprintOn=0}{
		{\printcredits}   

		\smallskip\noindent\textbf{Declaration of competing interest}
		The authors declare that they have no known competing financial interests or personal relationships that could have appeared to influence the work reported in this paper.
				
		\smallskip\noindent\textbf{Data availability} The data analyzed in this study are publicly available, as indicated in the sources cited in the text.
			}{}

	\smallskip\noindent\textbf{Code availability}
	The code for the R implementation used to generate the results in this work is available from the corresponding author upon reasonable request.
	
	\smallskip\noindent\textbf{Acknowledgments~}	 
	The authors thank       Gerd Wachsmuth,   Armin Fügenschuh, Markus Friedemann, Jesse Beisegel (BTU Cottbus--Senftenberg),  Lillian Achola  Oluoch (The Technical University of Kenya), Kewani Welay Brhane (Mekelle University, Ethiopia \& AIMS Cameroon)	and Olivier Menoukeu Pamen (University of Liverpool) for insightful discussions and valuable suggestions that improved this paper.

	\smallskip\noindent
	\textbf{Funding~} 	
	The  authors gratefully acknowledge the  support by the Deutsche Forschungsgemeinschaft (DFG), award number 458468407,  and by the  German Academic Exchange Service (DAAD), award number 57417894.

	\par
}

\appendix
\section{List of Notations}
{\footnotesize	
	\begin{longtable}{p{0.3\textwidth}p{0.68\textwidth}}	 	
		\multicolumn{2}{l}{\textbf{Acronyms}}\\[0.5ex]
		SIR & Susceptible Infected Recovered (with lifelong immunity) \\
		SIRS & Susceptible Infected Recovered (with loss of immunity) \\
		$SI^{\pm}RS$ & Susceptible Infected confirmed/hidden Recovered (with loss of immunity) \\
\ifnum\measles=1 		 $SEI^{\pm}R^{\pm}V$ & Susceptible Infected confirmed/hidden Recovered confirmed/hidden Vaccinated \\\fi		
		\covid & Coronavirus disease 2019 \\
		ICU & Intensive Care Unit \\
		LLN & Law of Large Numbers \\
		CLT & Central Limit Theorem \\
		ODE & Ordinary Differential Equation \\
		SDE & Stochastic Differential Equation \\
		CTMC & Continuous Time Markov Chain \\
		OLS & Ordinary Least Squares \\
		FIR & Finite Impulse Response \\[1ex]	
		
		\multicolumn{2}{l}{\textbf{Population Compartments}}\\[0.5ex]
		$S$ & Susceptible individuals \\
		\ifnum\measles=1 $E$ & Exposed individuals \\\fi
		$I, I^{-}/I^{+}$ & Infected / hidden  / confirmed infected\\
		$R, R^{-},R^{+}$ & Recovered / hidden  / confirmed individuals \\
		$R_{k}^{+}$ & Observable cascade compartments for observable recovered \\
		$R_{2}^{-}$ &  Terminal  cascade compartments for recovered individuals  with fading immunity\\
		$V$ & Vaccinated individuals \\
		$V_{k}$ & Observable cascade compartments for vaccinated individuals \\
		$V^{-}$ & Terminal  cascade compartments for vaccinated individuals  with fading immunity\\
		$C$ & Individuals in the ICU \\
		$H$ & Hospitalized individuals \\
		$D$ & Individuals who died in ICU \\
		$G_k$ & Generic cascade compartment  $k$ \\
		$N$ & Total population size \\[1ex]
		
		\multicolumn{2}{l}{\textbf{Epidemiological Parameters}}\\[0.5ex]
		$\alpha$ & Test rate \\
		$\beta$ & Transmission (infection) rate \\
		$\gamma, \gamma^+,\gamma^-$ & Recovery rate \\
		$\delta$ &  Hospitalization rate in ICU \\
		$\eta^{-}/\eta^{+}$ & Hospitalization rate from hidden / confirmed infected individuals \\
		$\kappa$ & Death rate from ICU \\
		$\mu$ & Vaccination rate \\
		\ifnum\measles=1 $\upsilon$ &  Infection rate from  $E$ to $I^{-}$ \\\fi
		$\rho,  \rho_1^-, \rho_2^-,\rho^V$ & Immunity loss rate \\
		$\widetilde\rho^{\,-}_2,\widetilde\rho^{\,V}$ &  Adjusted   immunity loss rates in extended \covid model\\
		$ \Reproeff, \Reproeffn$ &  effective reproduction number \\[1ex]

		\multicolumn{2}{l}{\textbf{Stochastic Epidemic Models}}\\[0.5ex]		
		$T,t$ & Time horizon, time  \\
		$X,\Xrel$ & State vector of absolute / relative population sizes \\
		$(\Omega, \mathcal{F}, \mathbb{F}, \mathbb{P})$ & Filtered probability space\ with sample space $\Omega$, probability measure $\P$\\
		$\mathbb{F}=\{\mathcal{F}_t\}_{t\ge0}$ & Filtration generated by $X$\\
		$\lambda=\lambda(t,x)$ & Transition intensity depending on time and  absolute population size $x$\\
		$\nu=\nu(t,\xrel)$ & Transition intensity depending on time and relative population size $\xrel$ \\		
		$\xi_k$ & Transition vector \\
		$K$ & Total number of transitions \\		
		$\countP$ & Counting process for transition $k$ \\
		$\Pi_k$ & Standard Poisson process \\[1ex]
		
		\multicolumn{2}{l}{\textbf{Diffusion Approximation and Time Discretization}}\\[0.5ex]
		$\XD,\XDrel$ & Diffusion approximation of absolute and relative subpopulation sizes\\
		$\diffX,\diffXrel$ & Diffusion coefficients of $\XD,\XDrel$ \\
		$\driftX,F$ & Drift coefficients of diffusion approximation of $\XD,\XDrel$ \\
		$f,f_0,f_1,\overline f,\overline f_0,\overline f_1$ & Drift coefficients (hidden states) \\		
		$h,h_0,h_1,\overline h,\overline h_0,\overline h_1$ & Drift coefficients (observable states) \\
		$\sigma,g,\overline \sigma,\overline  g$ & Diffusion coefficients (hidden state) \\		
		$\ell, \overline \ell$ & Diffusion coefficients (observable states) \\
		$W,W^1,W^2$ & Multi-dimensional Brownian motions \\
		$\Noise,\Noise^1,\Noise^2$ & Standard normally distributed random vectors \\
		 $\Gauss=\Gauss(t)$ &  Gaussian process in CLT \\
		$N_t$ & Total number of time steps \\
		$t_0,\ldots,t_{N_t}$ & Discrete time points \\[1ex]
		
		\multicolumn{2}{l}{\textbf{Model Structure}}\\[0.5ex]
		$X,\XrelN$ & State vectors (absolute / relative size) \\	
		$Y$ & Hidden state \\
		$Z$ & Observable state \\		
		$d,d_1,d_2$ & Number of total, hidden and observable states \\[1ex]		
		
		\multicolumn{2}{l}{\textbf{Cascade States}}\\[0.5ex]
		$\tau_l$ & accounting dates \\
		$\LG,\LR,\LV$ & Time steps with complete immunity: generic / \covid model \\
		$\KG,\KR,\KV$ & Number of grouped cascade compartments: generic / \covid model  \\
		$\PG,\PR,\PV$ & Number of grouped time steps: generic / \covid model   \\
		$\psiG, \psiR,\psiV$ & Weighting factors in cascade recursion:  generic / \covid model \\		

		\multicolumn{2}{l}{\textbf{Estimation and Filtering}}\\[0.5ex]		
		$\widehat{\alpha},\widehat{\beta},\widehat{\mu}$ & Raw parameter estimates\\
		$\widetilde{\alpha},\widetilde{\beta},\widetilde{\mu}$ & Smoothed parameter estimates \\
		$J_1,J_2,J_3$ & Order of the FIR \\
		$a_j,b_l,c_m$ & FIR coefficients \\
		$\regupar_{\beta}, \regupar_{\alpha},\regupar_{\mu}$ &  Regularization parameters \\[1ex]

	\end{longtable}
}

%
%

\bibliographystyle{cas-model2-names}

\bibliography{References_Epi_Model}

\end{document}